\documentclass[12pt]{article}

\newwrite\XTR
\AtBeginDocument{\immediate\openout\XTR\jobname.xtr}
\AtEndDocument{\immediate\closeout\XTR}

\newcommand{\defineXtrCommand}[2]{%
	\write\XTR{%
		\string\newcommand%
		{\csname myedit#1\endcsname}%
	}%
	\write\XTR{{#2}}%
}

\newcommand{\myedit}[2]{
	\defineXtrCommand{#1}{%
		``\expandafter\string#2''
		\null \hfill \mbox{(Page\string~\thepage)}%
	}%
	\defineXtrCommand{#1String}{%
		\expandafter\string#2%
	}%
	\defineXtrCommand{#1Page}{%
		\thepage%
	}%
	#2%
}

\newcommand{\myeditNoPrint}[2]{
	\defineXtrCommand{#1}{%
		``\expandafter\string#2''
		\null \hfill (Page\string~\thepage)%
	}%
	\defineXtrCommand{#1String}{%
		\expandafter\string#2%
	}%
	\defineXtrCommand{#1Page}{%
		\thepage%
	}%
}

\usepackage{natbib}
\PassOptionsToPackage{hyphens}{url} 
\usepackage[hidelinks]{hyperref} 
\usepackage{graphicx}
\usepackage{subcaption}
\usepackage{amsmath}
\usepackage{amssymb}
\usepackage{amsthm}
\usepackage{bm}
\usepackage{dsfont}
\usepackage{relsize}
\usepackage{calc}
\usepackage{placeins} 
\usepackage{algorithm} 
\usepackage[noend]{algpseudocode} 
\usepackage{algorithmicx}
\algrenewcommand\algorithmicindent{.75em} 
\newcommand{\commentMarker}{\#}
\algrenewcommand\algorithmiccomment[1]{\hfill \commentMarker\ #1}

\usepackage[subpreambles=true, mode=buildmissing]{standalone}
\usepackage{tikz}
\pgfrealjobname{zigzag_hmc} 
\usepackage[norefs, nocites]{refcheck}

\usepackage{multibib} 
\newcites{Supp}{References for Supplement}

\newtheorem{theorem}{Theorem}[section]
\newtheorem{lemma}[theorem]{Lemma}
\newtheorem{corollary}[theorem]{Corollary}

\makeatletter
\@ifundefined{c@rownum}{%
  \let\c@rownum\rownum
}{}
\@ifundefined{therownum}{%
  \def\therownum{\@arabic\rownum}%
}{} 
\makeatother

\definecolor{jhuBlue}{RGB}{0, 45, 114} 
\definecolor{jhuSpiritBlue}{RGB}{114, 172, 229} 
\definecolor{jhuSecondaryRed}{RGB}{118, 35, 47} 
\definecolor{jhuSecondaryOrange}{RGB}{207, 69, 32} 
\definecolor{jhuSecondaryGreen}{RGB}{0, 115, 119} 
\definecolor{jhuSecondaryYellow}{RGB}{241, 196, 0} 
\definecolor{jhuAccentBrown}{RGB}{79, 44, 29} 
\definecolor{jhuAccentBrownishYellow}{RGB}{203, 160, 82} 
\definecolor{turquoise}{rgb}{0.19, 0.84, 0.78}
\definecolor{mediumturquoise}{rgb}{0.28, 0.82, 0.8}
\definecolor{lava}{rgb}{0.81, 0.06, 0.13}

\colorlet{markovColor}{gray!6}
\colorlet{defaultNutsColor}{jhuAccentBrownishYellow!8} 
\colorlet{altNutsColor}{jhuSecondaryGreen!8} 
\colorlet{hzzManualColor}{jhuBlue!8}
\colorlet{harmonicHmcColor}{jhuSecondaryRed!8}
\colorlet{botevColor}{jhuAccentBrown!10}

\newcommand{\given}{\, | \,}
\newcommand{\diff}{\operatorname{\mathrm{d}}\!{}}
\newcommand{\sign}{\operatorname{sign}}
\newcommand{\diag}{\mathrm{diag}}

\DeclareMathOperator*{\argmin}{argmin}

\renewcommand{\complement}{{\raisebox{1.5pt}{\scriptsize $\mathsf{c}$}}}
\newcommand{\transpose}{\intercal} 
\newcommand{\lowerscript}[1]{\raisebox{-2pt}{\scriptsize $#1$}}
\newcommand{\yesnumber}{\addtocounter{equation}{1}\tag{\theequation}}
\newcommand{\thinnerspace}{\mskip.5\thinmuskip}
\newcommand{\spaceBeforePartial}{\mskip\thinmuskip}

\newcommand{\nParam}{d}

\newcommand{\principalComp}{\bm{u}}
\newcommand{\ball}{B}
\newcommand{\bmu}{\bm{\mu}}
\newcommand{\bPhi}{\bm{\Phi}}
\newcommand{\Id}{\bm{I}}
\newcommand{\observation}{y}
\newcommand{\bobservation}{\bm{\observation}}
\newcommand{\eigenValue}{\nu}

\newcommand{\position}{x}
\newcommand{\bposition}{\bm{\position}}
\newcommand{\momentum}{p}
\newcommand{\bmomentum}{\bm{\momentum}}
\newcommand{\velocity}{v}
\newcommand{\bvelocity}{\bm{\velocity}}
\newcommand{\positionMarginal}{\pi_{\scalebox{.65}{$X$}}} 
\newcommand{\momentumMarginal}{\pi_{\scalebox{.65}{$P$}}} 

\newcommand{\bPhiVelocity}[1][]{\bm{\varphi}_{\bvelocity #1}}
\newcommand{\bPhiPosition}[1][]{\bm{\varphi}_{\bposition #1}}
\newcommand{\switchRate}{\lambda}
\newcommand{\velocityFlipOperator}{F}
\newcommand{\potential}{U}
\newcommand{\kinetic}{K}

\newcommand{\problematicSet}{S}
\newcommand{\zeroMeasureSet}{\Omega}
\newcommand{\solutionOp}{\bm{\Psi}}
\newcommand{\momentumFlipOp}{\bm{R}}
\newcommand{\integrationTime}{T}
\newcommand{\baseIntegrationTime}{\Delta \integrationTime}
\newcommand{\baseIntegrationTimeMultiplier}{\baseIntegrationTime_{\textrm{rel}}}
\newcommand{\interEventTime}{t}

\newcommand{\dt}{{\Delta t}}
\newcommand{\superH}{{\scalebox{.65}{$\mathrm{H}$}}}
\newcommand{\superM}{{\scalebox{.65}{$\mathrm{M}$}}}
\newcommand{\eventIndicator}{I}
\newcommand{\gradBound}{L}
\newcommand{\hessianBound}{\kappa}

\newcommand{\positionPosition}[1][t]{%
	\dfrac{\partial \bposition_#1}{\partial \bposition}%
}
\newcommand{\positionMomentum}[1][t]{%
	\dfrac{\partial \bposition_#1}{\partial \bmomentum}%
}
\newcommand{\momentumPosition}{%
	\dfrac{\partial \bmomentum_t}{\partial \bposition}%
}
\newcommand{\momentumMomentum}{%
	\dfrac{\partial \bmomentum_t}{\partial \bmomentum}%
}
\newcommand{\spEventIndex}{i} 

\newcommand{\nLinearConstraint}{L}
\newcommand{\velocityFlipIndexSet}{I}

\newcommand{\probability}{\mathbb{P}}

\newcommand{\variance}{\mathrm{Var}}
\DeclareMathOperator*{\covariance}{Cov}
\newcommand{\indicator}{\mathds{1}}
\newcommand{\eqDistribution}{\mathrel{\raisebox{-.2ex}{$\overset{\scalebox{.6}{$\, d$}}{=}$}}}
\newcommand{\iidSim}{\mathrel{\raisebox{-.3ex}{$\overset{\text{i.i.d.}}{\sim}$}}}
\newcommand{\expDist}{\mathrm{Exp}}
\newcommand{\laplaceDist}{\mathrm{Laplace}}
\newcommand{\unifDist}{\mathrm{Unif}}
\newcommand{\normalDist}{\mathcal{N}}
\newcommand{\uniformRv}{u}
\newcommand{\buniformRv}{\bm{\uniformRv}}
\newcommand{\altUniformRv}{w}
\newcommand{\baltUniformRv}{\bm{\altUniformRv}}

\newcommand{\hmc}{\textsc{hmc}}
\newcommand{\Hmc}{\textsc{Hmc}}
\newcommand{\nuts}{\textsc{nuts}}
\newcommand{\Nuts}{\textsc{Nuts}}
\newcommand{\ess}{\textsc{ess}}
\newcommand{\Ess}{\textsc{Ess}}
\newcommand{\mcmc}{\textsc{mcmc}}

\newcommand{\pdmp}{\textsc{pdmp}}
\newcommand{\Pdmp}{\textsc{Pdmp}}
\newcommand{\bps}{the bouncy particle sampler}
\newcommand{\zigzag}{the zigzag sampler}
\newcommand{\hiv}{\textsc{hiv}}

\newcommand{\jhpce}{\textsc{jhpce}}

\newcommand{\phylogeny}{\mathcal{T}}
\newcommand{\latentData}{\bposition}
\newcommand{\traitCovariance}{\bm{\Gamma}}
\newcommand{\nTaxa}{n}
\newcommand{\nTraits}{m}

\usepackage{booktabs}
\usepackage{longtable}
\usepackage{array}
\usepackage{multirow}
\usepackage{wrapfig}
\usepackage{float}
\usepackage{colortbl}
\usepackage{pdflscape}
\usepackage{tabu}
\usepackage{threeparttable} \usepackage{threeparttablex} \usepackage[normalem]{ulem} \usepackage{makecell}
\usepackage{xcolor}

\newcommand{\blind}{1}

\begin{document}

\def\spacingset#1{\renewcommand{\baselinestretch}%
{#1}\small\normalsize} \spacingset{1}


\date{}
\newcommand{\titleString}{%
	Zigzag path connects two Monte Carlo samplers: Hamiltonian counterpart to a piecewise deterministic Markov process%
}
\if1\blind
{
  \title{\bf \titleString}
  \author{Akihiko Nishimura\\
    Department of Biostatistics, Bloomberg School of Public Health,\\ 
    Johns Hopkins University\\
    \null \\
    Zhenyu Zhang\\
    Department of Biostatistics, University of California, Los Angeles\\
    \null \\
    Marc A.\ Suchard\\
    Department of Biostatistics, Computational Medicine, and Human Genetics,\\
    University of California, Los Angeles}
  \maketitle
  \vspace*{-\baselineskip} 
} \fi

\if0\blind
{
  \bigskip
  \bigskip
  \bigskip
  \begin{center}
    {\LARGE\bf \titleString}
\end{center}
  \medskip
} \fi

\bigskip
\begin{abstract}
Zigzag and other piecewise deterministic Markov process samplers have attracted significant interest for their non-reversibility and other appealing properties for Bayesian posterior computation.
Hamiltonian Monte Carlo is another state-of-the-art sampler, exploiting fictitious momentum to guide Markov chains through complex target distributions.
We establish an important connection between the zigzag sampler and a variant of Hamiltonian Monte Carlo based on Laplace-distributed momentum.
The position and velocity component of the corresponding Hamiltonian dynamics travels along a zigzag path paralleling the Markovian zigzag process; 
however, the dynamics is non-Markovian in this position-velocity space as the momentum component encodes non-immediate pasts.
This information is partially lost during a momentum refreshment step, in which we preserve its direction but re-sample magnitude.
In the limit of increasingly frequent momentum refreshments, we prove that Hamiltonian zigzag converges strongly to its Markovian counterpart.
This theoretical insight suggests that, when retaining full momentum information, Hamiltonian zigzag can better explore target distributions with highly correlated parameters by suppressing the diffusive behavior of Markovian zigzag.
We corroborate this intuition by comparing performance of the two zigzag cousins on high-dimensional truncated multivariate Gaussians, including a 11,235-dimensional target arising from a Bayesian phylogenetic multivariate probit modeling of \hiv{} virus data. 
\end{abstract}

\noindent%
\vfill

\newpage
\spacingset{1.75} 

\section{Introduction}
\label{sec:intro}
Emergence of Monte Carlo methods based on continuous-time, non-reversible processes has been hailed as fundamentally new development \citep{fearnhead2018piecewise_deterministic_process} and has attracted an explosion of interest \citep{dunson2020hastings_algorithm}.
The two most prominent among such algorithms are the bouncy particle sampler \citep{bouchard2018bouncy_particle_sampler} and zigzag sampler \citep{bierkens2019zigzag_original}, both having roots in the computational physics literature \citep{peters2012rejection_free_monte_carlo, turitsyn2011irreversible_monte_carlo}.
These algorithms draw samples from target distributions by simulating piecewise deterministic Markov processes (\pdmp{}), which move along linear trajectories with instantaneous changes in their velocities occurring at random times according to inhomogeneous Poisson processes.

One known issue with \bps{} is its near-reducible behavior in the absence of frequent velocity refreshment \citep{bouchard2018bouncy_particle_sampler,  fearnhead2018piecewise_deterministic_process}.
In case of a high-dimensional i.i.d.\ Gaussian, \cite{bierkens2018scaling_of_pdmp_sampler} show that \bps{}'s optimal performance is achieved when refreshment accounts for as much as 78\% of all the velocity changes.
Such frequent velocity refreshment can lead to ``random-walk behavior,'' hurting computational efficiency of the sampler \citep{neal2010hmc, fearnhead2018piecewise_deterministic_process, andrieu2019peskun_ordering_beyond_reversible}.

\expandafter\MakeUppercase \zigzag{} on the other hand is provably ergodic without velocity refreshment \citep{bierkens2019ergodicity} and appears to have a competitive edge in high-dimensions \citep{bierkens2018scaling_of_pdmp_sampler}.
Much remains unknown, however, about how these samplers perform relative one another and against other classes of algorithms.
Early empirical results, while informative, remain limited by their focus on low-dimensional or synthetic examples. 
For example, the simulated logistic regression examples in \cite{bierkens2019zigzag_original} only have 16 regression coefficients and the most complex examples of \cite{bierkens2018scaling_of_pdmp_sampler} comprise 256-dimensional multi-variate Gaussians and spherically symmetric t-distributions.
Generating further theoretical and empirical insight on the practical performance of these algorithms thus stands out as one of the most critical research areas \citep{fearnhead2018piecewise_deterministic_process, bierkens2019zigzag_original, dunson2020hastings_algorithm}.

This article brings novel insight into \zigzag{}'s performance by revealing its intimate connection to a version of Hamiltonian Monte Carlo (\hmc{}) \citep{duane1987hmc, neal2010hmc}, another state-of-the-art paradigm for Bayesian computation.
\Hmc{} exploits an auxiliary momentum variable and simulates Hamiltonian dynamics to guide exploration of the original parameter space.
This notion of ``guidance'' provided by momentum has historically inspired the earliest examples of non-reversible Monte Carlo algorithms \citep{diaconis2000nonreversble_chain}.
Beyond this analogy and heuristic, however, the precise nature of relation between \hmc{} and non-reversible methods has never been explored.
In fact, the \pdmp{} sampler literature has so far mentioned \hmc{} only in passing or used it merely as a computational benchmark \citep{bierkens2019zigzag_original, bouchard2018bouncy_particle_sampler, sherlock2021discrete_bps}.
A faint link between the two paradigms appears in \cite{deligiannidis2021randomized}, who observe a univariate randomized \hmc{} to emerge as the first-coordinate marginal of \bps{} in a high-dimensional limit.
Their weak convergence result, however, only concerns the univariate marginal and offers little insight as to \textit{why} a randomized \hmc{} may appear as the limit.

We study a less explored version of \hmc{}, based on the momentum variable components $\momentum_i$ having independent Laplace distributions.
The corresponding Hamiltonian dynamics follows a zigzag path akin to that of the Markovian zigzag process. 
We thus call the resulting algorithm \emph{zigzag \hmc{}} and, to differentiate the dynamics underlying the two samplers, refer to them as \emph{Hamiltonian} and \emph{Markovian zigzag}.
In other words, except when we explicitly invoke partial momentum refreshment as a theoretical tool, Hamiltonian zigzag refers to the deterministic dynamics that constitutes a proposal generation mechanism for zigzag \hmc{}.
Hamiltonian zigzag in a sense is also ``Markovian,'' although this term is rarely applied to a deterministic process, in the position-momentum space as its future trajectory depends solely on the current state.
Hamiltonian zigzag becomes non-Markovian, however, when viewed as dynamics in the position-velocity space, velocity being the time derivative of position.

We establish that Markovian zigzag is essentially Hamiltonian zigzag with ``less momentum,'' thereby providing a unified perspective to compare the two Monte Carlo paradigms.
We consider a partial momentum refreshment for Hamiltonian zigzag,  in which $\momentum_i$'s retain their signs but have their magnitudes resampled from the exponential distributions.
With this refreshment step inserted at every $\dt > 0$ time interval, Hamiltonian zigzag converges strongly to its Markovian counterpart as $\dt \to 0$.

This result has significant implications in the relative performance of the two zigzag algorithms.
Markovian zigzag to some extent avoids random-walk behavior by retaining its direction from previous moments;
the inertia induced by retaining full momentum information, however, may allow Hamiltonian zigzag to better explore the space when parameters exhibit strong dependency.
The intuition is as follows.
Along each coordinate, the partial derivative of the log-density depends on other coordinates and  its sign can flip back and forth as the dynamics evolves.
Such fluctuation leads to inconsistent guidance from the derivative, ``pushing'' the dynamics to one direction at one moment and to the opposite at the next.
However, momentum can help each coordinate of the dynamics keeps traveling in an effective 
direction without being affected by small fluctuation in the derivative.

The advantage of having full momentum is visually illustrated in Figure~\ref{fig:zigzag_trajectory_on_ar_target}.
After comparable amounts of computation,
Hamiltonian zigzag has traversed a high-density region while Markovian zigzag is still slowly diffusing away from the initial position.
This difference in behavior also manifests itself in the overall distance traveled by the dynamics (Figure~\ref{fig:sq_traveled_dist_on_AR_target}).
Our observation here suggests that, with its ability to make larger transitions under comparable computational efforts, Hamiltonian zigzag constitutes a more effective transition kernel.

\begin{figure}
	\centering
	\hspace*{.005\textwidth}
	\begin{minipage}{0.49\textwidth}
	    \includegraphics[height=.33\textheight]{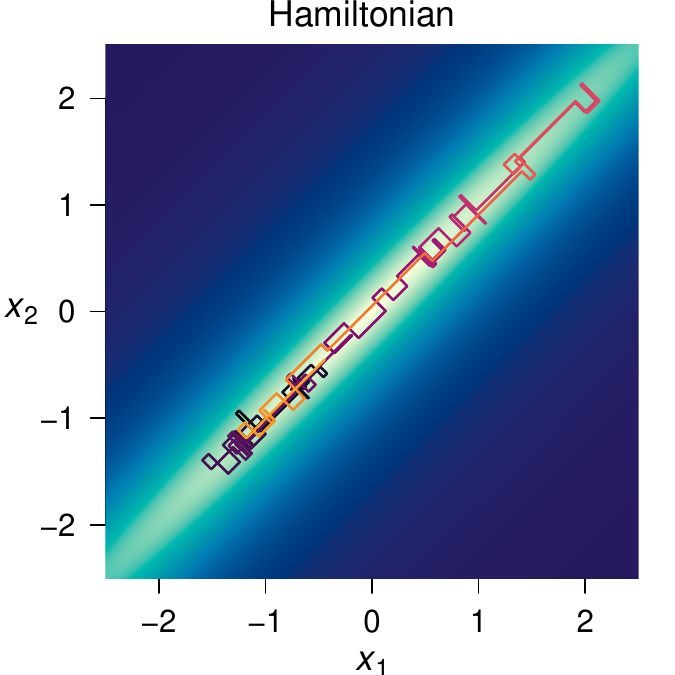}
	\end{minipage}
	\nobreak\hspace{.1em} 
	\begin{minipage}{0.47\textwidth}
		\includegraphics[height=.33\textheight, trim={.08\linewidth} 0 {.1\linewidth} 0, clip]{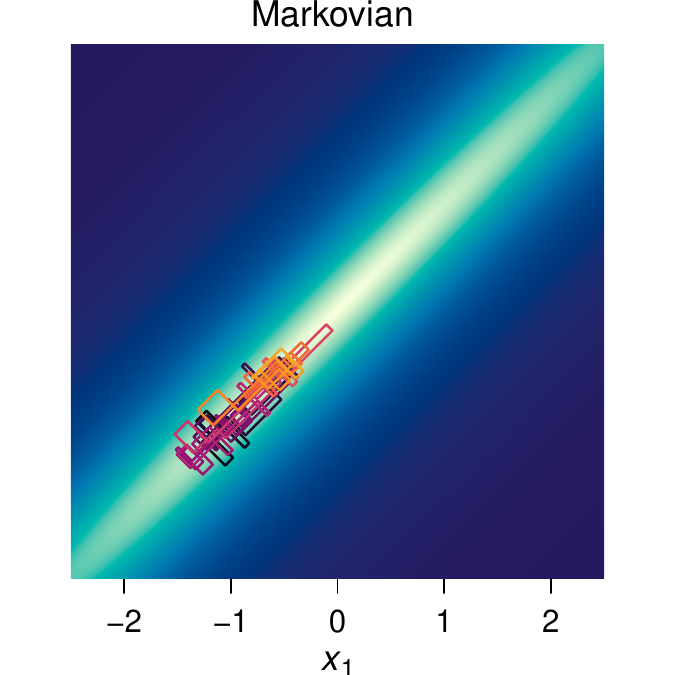}
	\end{minipage}
	\caption{%
		Trajectories of the first two position coordinates of Hamiltonian zigzag without momentum refreshment (left) and Markovian zigzag (right). 
		The target is a $2^{10} = 1{,}024$-dimensional Gaussian, corresponding to a stationary lag-one auto-regressive process with auto-correlation $0.99$ and unit marginal variances. 
		Both dynamics are simulated for $10^5$ linear segments, starting from the same position $\position_i = -1$ for all $i$ and same random velocity. 
		The line segment colors change from darkest to lightest as the dynamics evolve. 
	}
	\label{fig:zigzag_trajectory_on_ar_target}
\end{figure}

\begin{figure}
\begin{minipage}{.55\linewidth}
	\includegraphics[width=\linewidth]{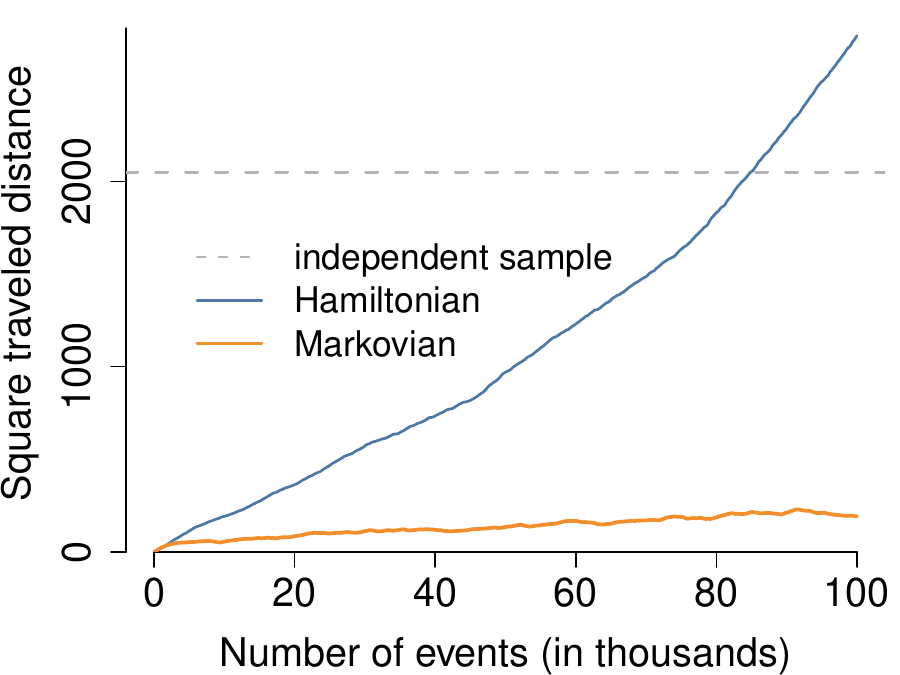}
\end{minipage}
~
\begin{minipage}{.43\linewidth}
	\caption{%
		Squared distance $\| \bposition(t) - \bposition_0 \|^2$ of the two zigzag dynamics from the initial position, plotted as function of the number of velocity change events.
		The experimental setup is identical to that of Figure~\ref{fig:zigzag_trajectory_on_ar_target}.
		The dashed line indicates the expected squared distance between the initial position and an independent sample from the target, as a benchmark of the distance traveled by an efficient transition kernel.
	}
	\label{fig:sq_traveled_dist_on_AR_target}
\end{minipage}
\end{figure}

We empirically quantify the superior performance of zigzag \hmc{} over Markovian zigzag on a range of truncated Gaussians, a special yet practically relevant class of targets on which we can efficiently simulate both zigzags. 
One of our examples arises from a Bayesian phylogenetic multivariate probit model for studying correlation structure among binary biological traits of \hiv{} viruses \citep{zhang2021phylo_multi_probit}.
This is a real-world high-dimensional problem in which the super-linear scaling property of Markovian zigzag established by \cite{bierkens2019zigzag_original} does not hold --- with the number of parameters growing proportional to that of observations, their sub-sampling and control variate techniques cannot be applied here. 
\myedit{LocalBpsForSubsampling}{%
	The same reasons hamper a use of the sub-sampling approach by \cite{bouchard2018bouncy_particle_sampler} based on the local bouncy particle sampler \citep{bardenet2017mcmc_for_tall_data}.%
}

The deterministic nature of its dynamics endows Hamiltonian zigzag with a couple of additional advantages over its Markovian counterpart as a building block for sampling algorithms for truncated Gaussians.
First, we can combine Hamiltonian zigzag with the no-U-turn algorithm to yield an effectively tuning-free sampler (Section~\ref{sec:zigzag_nuts}).
Hamiltonian zigzag is further applicable whenever a subset of parameters is conditionally distributed as truncated Gaussian; 
the split \hmc{} framework \citep{shahbaba2014split_hmc, nishimura2020discontinuous_hmc} allows us to combine Hamiltonian zigzag with other integrators to jointly update all the parameters.
The split \hmc{} extension is successfully deployed by \cite{zhang2022laplace_gause_hmc}, who additionally demonstrates zigzag \hmc{}'s advantage over the bouncy particle sampler.
These samplers based on Hamiltonian zigzag are provided as a part of the Bayesian phylogenetic software \textsc{beast} \citep{suchard2018beast} and are also available via the R package \texttt{hdtg} \citep{zhang2022hdtg}.
\myedit{ZigzagBeyondTruncatedGaussian}{%
	Finally, while Hamiltonian zigzag's use outside truncated Gaussians is beyond the scope of this article, we can in principle apply zigzag \hmc{} to any targets by approximating the dynamics through the coordinate-wise integrator of \citet{nishimura2020discontinuous_hmc} or through the mid-point integrator of \citet{chin2023bouncy_hmc} (Supplement Section~\ref{sec:midpoint_integrator}).%
}

While this work focuses on the two zigzags, their relationship as presented here seems to hint at \hmc{}'s more general connection to \pdmp{} and other non-reversible methods. 
We explore, and to some extent substantiate, this conjecture in the Discussion section.
In particular, we construct a discrete space analogue of Hamiltonian dynamics which, when combined with partial momentum refreshments, recovers the non-reversible algorithm of \cite{diaconis2000nonreversble_chain}.
Our result on the two zigzags, therefore, portends broader implications in advancing our understanding of \hmc{} and other non-reversible methods.

\FloatBarrier
\section{Zigzag Hamiltonian Monte Carlo}
\label{sec:hamiltonian_zigzag}

\subsection{Hamiltonian dynamics based on Laplace momentum}
\label{sec:laplace_momentum_based_hamiltonian_dynamics}
In order to sample from the parameter of interest $\bposition$, \hmc{} introduces an auxiliary \textit{momentum} variable $\bmomentum$ and targets the augmented distribution $\pi(\bposition, \bmomentum) := \positionMarginal(\bposition) \momentumMarginal(\bmomentum)$.
\Hmc{} explores the augmented space by simulating Hamiltonian dynamics, whose evolution is governed by the differential equation known as \textit{Hamilton's equation}:
\begin{equation}
\label{eq:hamilton}
\frac{\diff \bposition}{\diff t}
	= \nabla \kinetic(\bmomentum), \quad
\frac{\diff \bmomentum}{\diff t}
	= - \nabla \potential(\bposition),
\end{equation}
where $\potential(\bposition) = - \log \positionMarginal(\bposition)$ and $\kinetic(\bmomentum) = - \log \momentumMarginal(\bmomentum)$ are referred to as \textit{potential} and \textit{kinetic} energy.
The solution of \eqref{eq:hamilton} preserves the target $\pi(\bposition, \bmomentum)$ and, if it can be simulated exactly, can be deployed as a transition kernel.
Such exact simulation is infeasible in most applications, but a numerical approximation by a reversible integrator constitutes a valid Metropolis-Hastings proposal \citep{fang2014compressible_hmc}.
To this day, the original version of \hmc{} by \cite{duane1987hmc}, based on Gaussian-distributed momentum and the leapfrog integrator, dominates practice \citep{pymc2016, carpenter2017stan}.

For the purpose of dealing with discontinuous target densities, \cite{nishimura2020discontinuous_hmc} proposes the use of Laplace momentum $\momentumMarginal(\bmomentum) \propto \exp\!\left( - \sum_i |\momentum_i| \right)$ in \hmc{}.
The authors then proceed to develop a reversible integrator that qualitatively approximates the corresponding Hamiltonian dynamics, noting only in passing the similarity of its trajectories to those of Markovian zigzag.
Under Laplace momentum, Hamilton's equation becomes
\begin{equation}
\label{eq:hamilton_under_laplace_momentum}
\frac{\diff \bposition}{\diff t}
	= \sign(\bmomentum), \quad
\frac{\diff \bmomentum}{\diff t}
	= - \nabla U(\bposition),
\end{equation}
in which the velocity $\bvelocity := \diff \bposition / \diff t \in \{\pm 1\}^\nParam$ depends only on the sign of $\bmomentum$ and thus remains constant except when one of the $\momentum_i$'s undergoes a sign change.
This property yields a piecewise linear trajectory as follows.

We momentarily ignore a mathematical technicality that arises from the discontinuity of the $\sign$ function in \eqref{eq:hamilton_under_laplace_momentum}.
Starting from the state $(\bposition(\tau), \bmomentum(\tau))$ with $\momentum_i(\tau) \neq 0$ for all $i$ and $\bvelocity = \sign(\bmomentum)$, the dynamics according to \eqref{eq:hamilton_under_laplace_momentum} evolves as
\begin{align}
\bposition(\tau + t)
	&= \bposition(\tau) + t \bvelocity(\tau),
	\label{eq:hamiltonian_zigzag_position_dynamics}\\
\momentum_i(\tau + t)
	&= \momentum_i(\tau) - \int_{0}^{t} \partial_i \potential(\bposition(\tau + s)) \diff s
	 \ \text{ for all $i$},
	\label{eq:hamiltonian_zigzag_momentum_dynamics}
\end{align}
where \eqref{eq:hamiltonian_zigzag_position_dynamics} holds as long as the signs of $\momentum_i$, and hence $\bvelocity$, remain constant on $[\tau, \tau + t)$.
During this time, the momentum magnitude evolves as
\begin{equation}
\label{eq:hamiltonian_zigzag_momentum_magnitude_dynamics}
|\momentum_i|(\tau + t)
	= |\momentum_i|(\tau) - \int_{0}^{t} \velocity_i(\tau) \spaceBeforePartial \partial_i \potential(\bposition(\tau) + s \bvelocity(\tau)) \diff s
	 \ \text{ for all $i$}.
\end{equation}
From \eqref{eq:hamiltonian_zigzag_momentum_dynamics} and \eqref{eq:hamiltonian_zigzag_momentum_magnitude_dynamics}, we see that the sign change in $\momentum_i$ occurs at time $\tau + t_i$ where
\begin{equation}
\label{eq:hamiltonian_zigzag_event_time}
\interEventTime_i
	= \inf_{t > 0}	\left\{ |\momentum_i(\tau)| = \int_0^t \velocity_i(\tau) \spaceBeforePartial \partial_i \potential( \bposition(\tau) + s \bvelocity(\tau) ) \diff s \right\}.
\end{equation}
Let $i^* = \argmin_i \, \interEventTime_i$ denote the first coordinate to experience a sign change.
As $\momentum_{i^*}$ changes its sign at time $\tau + \interEventTime_{i^*}$, the velocity $\bvelocity = \diff \bposition / \diff t$ undergoes an instantaneous change
\begin{equation}
\label{eq:hamiltonian_zigzag_velocity_change}
\velocity_{i^*}(\tau + \interEventTime_{i^*}) = -\velocity_{i^*}(\tau), \ \
\velocity_{j}(\tau + \interEventTime_{i^*}) = \velocity_{j}(\tau)
\ \text{ for } \ j \neq i^*.
\end{equation}
Afterward, the position component proceeds along a new linear path
\begin{equation}
\label{eq:hamiltonian_zigzag_position_dynamics_after_bounce}
\bposition(\tau + t) = \bposition(\tau + \interEventTime_{i^*}) + (t - \interEventTime_{i^*}) \bvelocity(\tau + \interEventTime_{i^*})
\ \text{ for } \ t \geq t_{i^*}
\end{equation}
until the next sign change event.
All the while, the momentum component continues evolving according to \eqref{eq:hamiltonian_zigzag_momentum_dynamics}.
We summarize the properties \eqref{eq:hamiltonian_zigzag_position_dynamics}--\eqref{eq:hamiltonian_zigzag_position_dynamics_after_bounce} in the algorithmic description of Hamiltonian zigzag in Section~\ref{sec:how_to_simulate_hamiltonian_zigzag}.


Under mild conditions, which we quantify momentarily, the trajectory of $(\bposition(t), \bmomentum(t))$ as described satisfies the equation \eqref{eq:hamilton_under_laplace_momentum} except at the instantaneous moments of sign changes in $\momentum_i$'s.
For a differential equation with discontinuous right-hand side, what constitutes a solution and whether it is unique are delicate and complex questions. 
A study of existence and uniqueness typically starts by interpreting the equation as a differential inclusion problem \citep{filippov1988ode_with_discontinous_righthand}.
More precisely, given a differential equation $\diff \bm{z} / \diff t = \bm{f}(\bm{z})$ and discontinuity points of $\bm{f}$, the equality requirement is relaxed to an inclusion of the form $\diff \bm{z} / \diff t \in \bm{F}(\bm{z})$ for a suitable set $\bm{F}(\bm{z})$, such as a convex polytope whose extreme points consist of $\lim_{\bm{z}' \to \bm{z}} \bm{f}(\bm{z}')$.
The theory remains incomplete, however, despite years of research effort \citep{fetecau2003nonsmooth_mechanics, khulief2013modeling}.

As current theory falls short, we will directly establish an existence and uniqueness result for Hamilton's equation with Laplace momentum.
For continuously differentiable $\potential$, the process \eqref{eq:hamiltonian_zigzag_position_dynamics}--\eqref{eq:hamiltonian_zigzag_position_dynamics_after_bounce} defines a unique trajectory consistent with \eqref{eq:hamilton_under_laplace_momentum} as long as it stays away from the sets
\begin{equation}
\label{eq:problematic_set}
\problematicSet_i = \left\{
	(\bposition, \bmomentum) : \partial_i \potential(\bposition) = 0, \
	\momentum_i = 0 \thinnerspace
\right\}.
\end{equation}
If $\partial_{i} \potential(\bposition(\tau + t_{i})) \neq 0$ at the moment of sign change ${\momentum_i(\tau + \interEventTime_i)} = 0$,
the relation \eqref{eq:hamiltonian_zigzag_momentum_dynamics} dictates
\begin{equation*}
\lim_{t \to \interEventTime_{i}^+} \sign \big( \momentum_{i}( \tau + t ) \big)
	= - \lim_{t \to \interEventTime_{i}^-} \sign \big( \momentum_{i}(\tau + t ) \big),
\end{equation*}
where $\interEventTime_{i}^+$ and $\interEventTime_{i}^-$ indicate the right and left limit.
Hence, the velocity change and subsequent evolution according to \eqref{eq:hamiltonian_zigzag_velocity_change} and \eqref{eq:hamiltonian_zigzag_position_dynamics_after_bounce} define a unique trajectory that satisfies Equation \eqref{eq:hamilton_under_laplace_momentum} for almost every $t \geq 0$ and has its position components continuous in $t$.

\myedit{PossibilityOfWeakerAssumption}{%
	In Theorem~\ref{thm:existence_and_propeties_of_hamiltonian_zigzag} below, we make a convenient assumption to ensure that a trajectory avoids the problematic sets \eqref{eq:problematic_set} from almost every initial state. 
	While identifying more general conditions is beyond the scope of this work, we believe the dynamics to be well-defined under a much weaker assumption on $\potential(\bposition)$ satisfied by most practical situations.%
}
The theorem also establishes the time-reversibility and symplecticity of the dynamics, which together imply that Hamiltonian zigzag preserves the target distribution and thus constitutes a valid transition kernel \citep{neal2010hmc, fang2014compressible_hmc}.
\begin{theorem}
\label{thm:existence_and_propeties_of_hamiltonian_zigzag}
Suppose that $\potential(\bposition)$ is twice continuously differentiable and that the sets $\{ \bposition : \partial_i \potential(\bposition) = 0 \}$ comprise differentiable manifolds of dimension at most $\nParam - 1$.
\myedit{PartOfTheoremStatement}{%
	Then Eq~\eqref{eq:hamilton_under_laplace_momentum} defines a unique dynamics on $\mathbb{R}^{2\nParam}$ away from a set of Lebesgue measure zero.
	More precisely, there is a measure zero set $\zeroMeasureSet \subset \mathbb{R}^{2\nParam}$ such that, for all initial conditions $(\bposition(0), \bmomentum(0)) \in \mathbb{R}^{2\nParam} \setminus \zeroMeasureSet$, there exists a unique solution that satisfies Eq~\eqref{eq:hamilton_under_laplace_momentum} at almost every $t \geq 0$, that remains in $\mathbb{R}^{2\nParam} \setminus \zeroMeasureSet$ for all $t \geq 0$, and whose position component is continuous in $t$.%
}
Moreover, the dynamics is time-reversible and symplectic on $\mathbb{R}^{2\nParam} \setminus \zeroMeasureSet$.
\end{theorem}
\noindent%
\myedit{CorollaryIntro}{%
	The assumption of Theorem~\ref{thm:existence_and_propeties_of_hamiltonian_zigzag} in particular holds for strongly convex $\potential(\bposition)$ since it would imply that $\nabla \partial_i \potential = (\partial_1 \partial_i \potential, \ldots, \partial_\nParam \partial_i \potential)^\transpose \neq \bm{0}$ and, by the implicit function theorem, that $\{ \bposition : \partial_i \potential(\bposition) = 0 \}$ is a differentiable manifold of dimension $d - 1$ \citep{spivak1965calculus_on_manifold}.
	We thus have the following corollary: 
}%
\begin{corollary}
\label{cor:well_definedness_of_hamiltonian_zigzag_for_strongly_logconcave}
\myedit{corollaryStatement}{%
	For a twice continuously differentiable, strongly log-concave target $\positionMarginal(\bposition)$ and the corresponding potential energy $\potential(\bposition) = - \log \positionMarginal(\bposition)$, the Hamiltonian dynamics based on Laplace momentum is well-defined, time-reversible, and symplectic away from a set of Lebesgue measure zero.%
}%
\end{corollary}

We note that, the exact nature of Hamiltonian zigzag on smooth targets being auxiliary to their work, \cite{nishimura2020discontinuous_hmc} establish corresponding results only under a piecewise constant $\potential$ with piecewise linear discontinuity set.
Given the incomplete theory behind discontinuous differential equations and non-smooth Hamiltonian mechanics, Theorem~\ref{thm:existence_and_propeties_of_hamiltonian_zigzag} is significant on its own and of independent interest.
We defer the proof to Supplement Section~\ref{sec:proof_of_hamiltonian_zigzag_properties}, however, to keep the article's focus on the connection between Hamiltonian and Markovian zigzag and its implication on Monte Carlo simulation.
W
\myedit{CommentOnTheoryWithConstraints}{%
	e similarly defer to Supplement Section~\ref{sec:hamiltonian_zigzag_on_constrained_domain} a discussion of the theory's extension to accommodate constraints on the parameter space, which justifies our application of Hamiltonian zigzag to truncated Gaussian targets (Section~\ref{sec:simulation}).%
}

\subsection{Simulation of Hamiltonian zigzag dynamics}
\label{sec:how_to_simulate_hamiltonian_zigzag}
To summarize the above discussion and prepare for the subsequent discussion of Hamiltonian zigzag's connection to Markovian one, we describe the evolution of Hamiltonian zigzag in the space $(\bposition, \bvelocity, |\bmomentum|)$ with $\momentum_i = \velocity_i |\momentum_i|$ as follows.
We denote the $k$-th event time by $\tau^{(k)}$ for $k \geq 1$ and the corresponding state at the moment by $\left(\bposition^{(k)}, \bvelocity^{(k)}, |\bmomentum|^{(k)}\right) = \left( \bposition\!\left(\tau^{(k)}\right), \bvelocity\!\left(\tau^{(k)}\right), |\bmomentum|\!\left(\tau^{(k)}\right) \right)$.
From a given initial condition at time $\tau^{(0)}$, the position coordinate follows a piecewise linear path segmented by the event times $\tau^{(k)}$, in-between which the dynamics evolves according to
\begin{gather*}
\begin{aligned}
\bposition(\tau^{(k)} \! + t) &= \bposition^{(k)} \! + t \bvelocity^{(k)}, \ \
\bvelocity(\tau^{(k)} \! + t) = \bvelocity^{(k)}, \text{ and}\\
| \momentum_i |(\tau^{(k)} \! + t) &= |\momentum_i|^{(k)} - \int_0^t \velocity_i^{(k)} \partial_i \potential\big( \bposition^{(k)} \! + s \bvelocity^{(k)} \big) \diff s.
\end{aligned}
\end{gather*}
The $(k + 1)$-th event occurs at time
\begin{equation*}
\tau^{(k + 1)} = \tau^{(k)} + \min_i \interEventTime_i^{(k)}
	\ \text{ where } \
	\interEventTime_i^{(k)} = \inf_{t > 0} \left\{
		|\momentum_i|^{(k)} = \int_0^t \velocity_i^{(k)} \partial_i \potential\big( \bposition^{(k)} \! + s \bvelocity^{(k)} \big) \diff s
	\right\},
\end{equation*}
resulting in an instantaneous change in the $i^*$-th component of velocity for $i^*(k) = \mathrm{argmin}_i \interEventTime_i^{(k)}$:
\begin{equation*}
\velocity_{i^*}^{(k + 1)} = - \velocity_{i^*}^{(k)}
\ \text{ and } \
\velocity_{j}^{(k + 1)} = \velocity_{j}^{(k)}
\ \text{ for } \
j \neq i^*.
\end{equation*}
The position and momentum magnitude at time $\tau^{(k + 1)}$ are given by
\begin{gather*}
\begin{aligned}
\bposition^{(k + 1)} &= \bposition^{(k)} \! + \big( \tau^{(k + 1)} - \tau^{(k)} \big) \bvelocity^{(k)} \text{ and}\\
|\momentum_i |^{(k + 1)} &= |\momentum_i|^{(k)} - \int_0^{\tau^{(k + 1)} - \tau^{(k)}} \hspace*{-1ex} \velocity_i^{(k)} \partial_i \potential\big( \bposition^{(k)} \! + s \bvelocity^{(k)} \big) \diff s.
\end{aligned}
\end{gather*}
The dynamics then continues in the same manner for the next interval $\left[ \tau^{(k + 1)}, \tau^{(k + 2)} \right)$.

Algorithm~\ref{alg:hamiltonian_zigzag_simulation} summarizes the trajectory simulation process as pseudo-code.
For the moment, we do not concern ourselves with either how we would solve for $\interEventTime_i$ in Line~\ref{line:coord_wise_event_time} or how to evaluate the integrals in Line~\ref{line:integral_1} and \ref{line:integral_2}.
As we demonstrate in Section~\ref{sec:simulation}, we can exploit the analytical solutions available under (truncated) multivariate Gaussians for a highly efficient implementation (Supplement Section~\ref{sec:zigzags_on_truncated_gaussians}).

{\spacingset{1.1}
\vspace*{\baselineskip}
\hspace*{-.04\linewidth}
\begin{minipage}{.5\linewidth}
\begin{algorithm}[H]
	    \caption[caption]{Hamiltonian zigzag\\ \hspace*{2.5em}trajectory simulation for $t \in [0, \integrationTime]$}
	    \label{alg:hamiltonian_zigzag_simulation}
	    \begin{algorithmic}[1]
	        \Function{HamiltonianZigzag}{$\bposition, \hspace*{-.12em} \bmomentum, \hspace*{-.12em} \integrationTime$}
	        	\State $\tau \gets 0$
	            \State $\bvelocity \gets \textrm{sign}(\bmomentum)$
	            \While{$\tau < \integrationTime$}
	            	\For{$i = 1, \ldots, \nParam$}
	            		\State \vphantom{$\uniformRv_i \sim \unifDist(0, 1)$} 
	            		\State \begin{varwidth}[t]{\linewidth}
			            			$\interEventTime_i = \inf\limits_{t > 0} \Big\{
			            				|\momentum_i| =$ \\
			            			\hspace*{4.5em} $\int_0^t
			            				\velocity_i \spaceBeforePartial \partial_i \potential\big( \bposition \! + s \bvelocity \big) \diff s
			            			\Big\}$
	            			\end{varwidth}
	            			\label{line:coord_wise_event_time}
	            	\EndFor
	            	\State $\interEventTime^* \gets \min_i \interEventTime_i$
	            	\If{$\tau + \interEventTime^* > \integrationTime$} \\ \Comment{No further event occurred} \hspace*{.02\linewidth}
	            		\State $\bposition \gets \bposition + (\integrationTime - \tau) \bvelocity$
	            		\State $\bmomentum \gets \bmomentum - \int_0^{\integrationTime - \tau} \nabla \potential\big( \bposition \! + s \bvelocity \big) \diff s.$
	            			\label{line:integral_1}
	            		\State $\tau \gets \integrationTime$
	            	\Else
	            		\State $\bposition \gets \bposition + \interEventTime^* \bvelocity$
	            		\State $\bmomentum \gets \bmomentum - \int_0^{ \interEventTime^*} \nabla \potential\big( \bposition \! + s \bvelocity \big) \diff s.$
	            			\label{line:integral_2}
	            		\State $i^* \gets \argmin_i \interEventTime_i$
	            		\State $\velocity_{i^*} \gets - \velocity_{i^*}$
	            		\State $\tau \gets \tau + \interEventTime^*$
	            	\EndIf
	            \EndWhile
	            \State \textbf{return} $(\bposition, \bmomentum)$
	        \EndFunction
	    \end{algorithmic}
	\end{algorithm}
\end{minipage}
\nobreak\hspace{.1em} 
\begin{minipage}{.485\linewidth}
	\begin{algorithm}[H]
	    \caption[caption]{Markovian zigzag\\ \hspace*{2.5em}trajectory simulation for $t \in [0, \protect \integrationTime]$}
	    \label{alg:markovian_zigzag_simulation}
	    \begin{algorithmic}[1]
	        \Function{MarkovianZigzag}{$\bposition, \hspace*{-.12em}  \bvelocity, \hspace*{-.12em}  \integrationTime$}
	        	\State $\tau \gets 0$
	            \State
	            \While{$\tau < \integrationTime$}
	            	\For{$i = 1, \ldots, \nParam$}
	            		\State $\uniformRv_i \sim \unifDist(0, 1)$
	            			\label{line:uniform_rv_for_markovian_zigzag_simulation}
	            		\State \begin{varwidth}[t]{\linewidth}
		            			$\interEventTime_i = \inf\limits_{t > 0} \Big\{
		            				{-} \log \uniformRv_i =$ \\
		            			\hspace*{3.8em} $\int_0^t \left[
		            					\velocity_i \spaceBeforePartial \partial_i \potential\big( \bposition \! + s \bvelocity \big)
		            				\right]^+ \! \diff s
		            			\Big\}$
		            		\end{varwidth}
	            	\EndFor
	            	\State $\interEventTime^* \gets \min_i \interEventTime_i$
	            	\If{$\tau + \interEventTime^* > \integrationTime$} \\ \Comment{No further event occurred}
	            		\State $\bposition \gets \bposition + (\integrationTime - \tau) \bvelocity$
	            		\State \vphantom{$\bmomentum \gets \bmomentum - \int_0^{\integrationTime - \tau} \nabla \potential\big( \bposition \! + s \bvelocity \big) \diff s.$} 
	            		\State $\tau \gets \integrationTime$
	            	\Else
	            		\State $\bposition \gets \bposition + \interEventTime^* \bvelocity$
	            		\State
	            		\State $i^* \gets \argmin_i \interEventTime_i$
	            		\State $\velocity_{i^*} \gets - \velocity_{i^*}$
	            		\State $\tau \gets \tau + \interEventTime^*$
	            	\EndIf
	            \EndWhile
	            \State \textbf{return} $(\bposition, \bvelocity)$
	        \EndFunction
	    \end{algorithmic}
	\end{algorithm}
\end{minipage}
}

\section{Link between Hamiltonian and Markovian zigzags}

\subsection{Hamiltonian zigzag's apparent similarity to Markovian zigzag}
\label{sec:comparison_of_two_zigzags}
The Markovian zigzag process by \cite{bierkens2019zigzag_original} follows a piecewise linear trajectory similar to Hamiltonian zigzag, but without any apparent concept of momentum.
Starting from the state $(\bposition(\tau), \bvelocity(\tau))$ with $\bvelocity \in \{\pm 1\}^\nParam$, Markovian zigzag follows a linear path
\begin{equation*}
\bposition(\tau + t)
	= \bposition(\tau) + t \bvelocity(\tau), \quad
\bvelocity(\tau + t)
	= \bvelocity(\tau)
\end{equation*}
for $t \geq 0$, until the next velocity switch event $\velocity_i \gets - \velocity_i$ that occurs with Poisson rate
\begin{equation}
\label{eq:markovian_zigzag_event_rate}
\switchRate_i(\bposition, \bvelocity)
	= \left[ \velocity_i \partial_i \potential(\bposition) \right]^+
	:= \max \{ 0, \velocity_i \partial_i \potential(\bposition) \}.\footnote{%
		Markovian zigzag based on the rate \eqref{eq:markovian_zigzag_event_rate} is referred to as the \textit{canonical zigzag} by \cite{bierkens2019zigzag_original} and is the predominant version in the literature.
		More generally, however, any Poission rate satisfying
 		$\switchRate_i(\bposition, \bvelocity) - \switchRate_i(\bposition, \velocityFlipOperator_i(\bvelocity))
			= \velocity_i \partial_i \potential(\bposition)$
		can be used, where
		$\left[ \velocityFlipOperator_i(\bvelocity) \right]_i = - \velocity_i$ and $\left[ \velocityFlipOperator_i(\bvelocity) \right]_j = \velocity_j$ for $j \neq i$.
	}
\end{equation}
In particular, the next event time $\tau + \interEventTime^*$ can be simulated by setting $\interEventTime^* = \min_i \interEventTime_i$ where
\begin{equation}
\label{eq:markovian_zigzag_event_time}
\interEventTime_i
	= \inf_{t > 0}	\left\{ - \log \uniformRv_i = \int_0^t \left[ \velocity_i(\tau) \spaceBeforePartial \partial_i \potential( \bposition(\tau) + s \bvelocity(\tau) ) \right]^+ \diff s \right\}
	\ \text{ for } \
	\uniformRv_i \sim \unifDist(0, 1).
\end{equation}
At time $\tau + \interEventTime^*$, the velocity undergoes an instantaneous change
\begin{equation*}
\velocity_{i^*}(\tau + \interEventTime_{i^*}) = -\velocity_{i^*}(\tau), \ \
\velocity_{j}(\tau + \interEventTime_{i^*}) = \velocity_{j}(\tau)
\ \text{ for } \ j \neq i^* = \argmin_i \, \interEventTime_i.
\end{equation*}
The position component then proceeds along a new linear path
\begin{equation*}
\bposition(\tau + t) = \bposition(\tau + \interEventTime_{i^*}) + (t - \interEventTime_{i^*}) \bvelocity(\tau + \interEventTime_{i^*})
	\ \text{ for } \ t \geq t_{i^*}
\end{equation*}
until the next sign change event. 
The Markovian zigzag process as described has the stationary distribution $\pi(\bposition, \bvelocity) := 2^{-\nParam} \, \positionMarginal(\bposition)$, with $\bvelocity$ distributed uniformly on $ \{\pm 1\}^\nParam$.

Algorithm~\ref{alg:markovian_zigzag_simulation} describes the dynamics of Markovian zigzag in pseudo-code with empty lines inserted as appropriate to facilitate comparison with the dynamics of Hamiltonian zigzag.
The similarity between the two zigzags is striking.
The main difference lies in Lines~\ref{line:uniform_rv_for_markovian_zigzag_simulation} and \ref{line:coord_wise_event_time} of the algorithms, reflecting the formulae \eqref{eq:markovian_zigzag_event_time} and \eqref{eq:hamiltonian_zigzag_event_time} for their respective velocity switch event times.
For one thing, Markovian zigzag's event time depends on the random quantity $\uniformRv_i$ while Hamiltonian zigzag's is deterministic.
On the other hand, when combining Hamiltonian zigzag with momentum refreshment, the distributional equality $- \log \uniformRv_i \eqDistribution | \momentum_i | \sim \expDist(\textrm{scale} = 1)$ makes the quantities $- \log \uniformRv_i$ and $| \momentum_i |$ comparable in a sense and, as we will show in Section~\ref{sec:markovian_zigzag_as_limit_of_hamiltonian_zigzag}, is a key element connecting the two zigzags.

The quantities $- \log \uniformRv_i$ and $| \momentum_i |$ being comparable, the only remaining difference between \eqref{eq:markovian_zigzag_event_time} and \eqref{eq:hamiltonian_zigzag_event_time} is the presence and absence of the positive part operator $[\, \cdot \, ]^+$ in the integrands.
These presence and absence of $[\, \cdot \, ]^+$ are manifestations of the fact that Markovian zigzag is memory-less while Hamiltonian zigzag transfers energy between the potential and kinetic parts and encodes this information in momentum.
The etiology and consequence of this difference is most easily seen in the case of a one-dimensional unimodal target $\potential(\position)$, as visually illustrated in Figure~\ref{fig:markovian_and_hamiltonian_balls_comparison}.
Before a velocity switch event, a zigzag trajectory from the initial position $\position_0$ and velocity $\velocity_0$ satisfies
{\renewcommand{\arraystretch}{1.3}
\begin{align*}
\int_0^t \left[ \velocity_0 \spaceBeforePartial \partial_i \potential( \position_0 + s \velocity_0) \right]^+ \diff s
	&= \left\{
		\begin{array}{l}
		0
		\ \, \text{ for } \ 0 \leq t \leq t_{\min} \\
		\potential(\position_0 + t \velocity_0) - \potential_{\min}
		\  \, \text{ for } \ t > t_{\min},
		\end{array}
	\right.
	\\
\int_0^t \velocity_0 \spaceBeforePartial \partial_i \potential( \position_0 + s \velocity_0) \diff s
	&= \potential(\position_0 + t \velocity_0) - \potential(\position_0)
	\\
	&= \potential(\position_0 + t \velocity_0) - \potential_{\min}
		- \left[  \potential(\position_0) - \potential_{\min} \right],
	\\
	&\hspace{-4em} \ \text{ where } \
		t_{\min} := \argmin_{t \, \geq \, 0} \potential(\position_0 + t \velocity_0)
		\ \text{ and } \
		\potential_{\min}:= \potential(\position_0 + t_{\min} \velocity_0).
\end{align*}
}
The velocity switch event formulae \eqref{eq:markovian_zigzag_event_time} and \eqref{eq:hamiltonian_zigzag_event_time} therefore simplify to
{\renewcommand{\arraystretch}{1.3}
\begin{equation}
\label{eq:comparison_of_markovian_and_hamiltonian_event_time}
\arraycolsep=0pt
	\begin{array}{rrl}
	\interEventTime^\superM 
		= \inf\limits_{t \geq t_{\min}} \big\{ & \,
			- \log \uniformRv &{}
			= \potential(\position_0 + t \velocity_0) - \potential_{\min} 
		\big\}, \\
	\interEventTime^\superH 
		= \inf\limits_{t \geq t_{\min}} \big\{ & \,
			\potential(\position_0) - \potential_{\min} + |\momentum_0| &{}
			=  \potential(\position_0 + t \velocity_0) - \potential_{\min}
		\big\}.
	\end{array}
\end{equation}
}%
From \eqref{eq:comparison_of_markovian_and_hamiltonian_event_time}, we see that the Markovian event necessarily precedes the Hamiltonian one (i.e.\ $\interEventTime^\superM \leq \interEventTime^\superH$) when $|\momentum_0| = - \log u$ and hence $| \momentum(t_{\min}) | = \potential(\position_0) - \potential_{\min} + |\momentum_0| \geq - \log u$.

In higher dimensions, the same reasoning applies to the relative behavior of the two zigzags along each coordinate.
On average, the memory-less property as manifested by the presence of $[\, \cdot \, ]^+$ in \eqref{eq:markovian_zigzag_event_time} causes Markovian zigzag to experience more frequent velocity switch events and travel shorter distances along each linear segment.
In contrast, when a coordinate of Hamiltonian zigzag is going down potential energy hills (i.e. $\velocity_i \partial_i \potential(\bposition) < 0$), the decrease in potential energy causes an equivalent increase in kinetic energy and in momentum magnitude as dictated by the relation \eqref{eq:hamiltonian_zigzag_momentum_magnitude_dynamics}.
Hamiltonian zigzag can then use this stored kinetic energy to continue traveling in the same direction for longer distances.

\begin{figure}[htb]
\centering

\usetikzlibrary{calc}
\usetikzlibrary{arrows.meta}
\newcommand{\parabolaCurvature}{.6}
\newcommand{\xLeft}{-1.75}
\newcommand{\xRight}{2.5}
\newcommand{\gapBetAxisParabola}{.2}
\newcommand{\axisOverflow}{.0}
\newcommand{\ballRadius}{.2}
\newcommand{\ballCenterX}{.875 * \xLeft + .7 * \ballRadius}
\newcommand{\ballCenterY}{\parabolaCurvature * (.875 * \xLeft)^2 + \ballRadius}
\newcommand{\markovianBallX}{ - .8 * (\ballCenterX) - .7 * \ballRadius}
\newcommand{\markovianBallY}{\parabolaCurvature * (.8 * \ballCenterX)^2 + \ballRadius}
\newcommand{\hamiltonianBallY}{\markovianBallX + \ballCenterY}
\newcommand{\hamiltonianBallX}{sqrt((\hamiltonianBallY) / \parabolaCurvature) - 1.0 * \ballRadius}

\vspace*{\baselineskip}
\hspace*{-.05\linewidth}
	\begin{minipage}{.45\linewidth}
		\beginpgfgraphicnamed{zigzag_hmc_markovian_ball}
		\begin{tikzpicture}[scale=1.08, every node/.style={scale=1.08}] 
		\coordinate (ballCenter) at (\ballCenterX, {\ballCenterY});
		\draw[thick, color=jhuSpiritBlue] (0, 1.05 * \ballRadius) parabola (ballCenter);
		\draw[thick, color=jhuSpiritBlue] [-Latex] (0, 1.05 * \ballRadius) parabola ({\markovianBallX}, {\markovianBallY});
		\fill[white] (ballCenter) circle (1.4 * \ballRadius); 
		\draw[semithick, fill={rgb:black,1;white,4}] (ballCenter) circle (\ballRadius);
		\draw[semithick] (0, 0) parabola (\xLeft, - \parabolaCurvature * \xLeft^2);
		\draw[semithick]  (0, 0) parabola (\xRight, \parabolaCurvature * \xRight^2);
		\draw [semithick] [-Latex]
			(1.2 * \xLeft - \axisOverflow, - \gapBetAxisParabola) --
			(1.2 * \xRight + \axisOverflow, - \gapBetAxisParabola);
		\draw [semithick]  (\ballCenterX, - .125- \gapBetAxisParabola) -- (\ballCenterX, - \gapBetAxisParabola);
		\node [anchor=base] at (\ballCenterX, - .4 - \gapBetAxisParabola) {$\position_0$};
		\draw [semithick] (0, - .125 - \gapBetAxisParabola) -- (0, - \gapBetAxisParabola);
		\node [anchor=base] at (0, - .4 - \gapBetAxisParabola) {\hspace*{1.5ex} $\position_0 + t_{\min} \velocity_0$};
		\draw [Latex-Latex] ({\markovianBallX + 3 * \ballRadius}, 0) -- ({\markovianBallX + 3 * \ballRadius}, {\markovianBallY});
		\node[right] at ({\markovianBallX + 3 * \ballRadius}, {.4 * \markovianBallY}) {$- \log \uniformRv$};
		\end{tikzpicture}
	\endpgfgraphicnamed
	\end{minipage}
	\hspace*{-.05\linewidth}
	\begin{minipage}{.48\linewidth}
		\beginpgfgraphicnamed{zigzag_hmc_hamiltonian_ball}
		\begin{tikzpicture}[scale=1.08, every node/.style={scale=1.08}]
		\coordinate (ballCenter) at (\ballCenterX, {\ballCenterY});
		\draw[thick, color=jhuSpiritBlue] (0, 1.05 * \ballRadius) parabola (ballCenter);
		\draw[thick, color=jhuSpiritBlue] [-Latex] (0, 1.05 * \ballRadius) parabola ({\hamiltonianBallX}, {\hamiltonianBallY});
		\fill[white] (ballCenter) circle (1.4 * \ballRadius); 
		\draw[semithick, fill={rgb:black,1;white,4}] (ballCenter) circle (\ballRadius);
		\draw[semithick] (0, 0) parabola (\xLeft, - \parabolaCurvature * \xLeft^2);
		\draw[semithick]  (0, 0) parabola (\xRight, \parabolaCurvature * \xRight^2);
		\draw [semithick] [-Latex]
			(1.2 * \xLeft - \axisOverflow, - \gapBetAxisParabola) --
			(1.2 * \xRight + \axisOverflow, - \gapBetAxisParabola);
		\draw [semithick]  (\ballCenterX, - .125 - \gapBetAxisParabola) -- (\ballCenterX, - \gapBetAxisParabola);
		\node [anchor=base] at (\ballCenterX, - .4 - \gapBetAxisParabola) {$\position_0$};
		\draw [semithick] (0, - .125 - \gapBetAxisParabola) -- (0, - \gapBetAxisParabola);
		\node [anchor=base] at (0, - .4 - \gapBetAxisParabola) {\hspace*{1.5ex} $\position_0 + t_{\min} \velocity_0$};
		\newcommand{\verticalLineX}{\xLeft - .5 * \ballRadius}
		\draw[color=jhuSecondaryOrange] [|-|] (\verticalLineX, 0) -- (\verticalLineX, {\ballCenterY});
		\node[left] [jhuSecondaryOrange] at (\verticalLineX, \parabolaCurvature * \xLeft * \xLeft - \ballRadius) {$\potential(\position_0)$};
		\node[left] [jhuSecondaryOrange] at ($(\verticalLineX, .5ex) - (.75ex, 0)$) {$\potential_{\min}$};
		\draw [Latex-Latex] ({(\hamiltonianBallX) + 2 * \ballRadius}, 0) -- ({(\hamiltonianBallX) + 2 * \ballRadius}, {\markovianBallY});
		\node[right] at ({(\hamiltonianBallX) + 2 * \ballRadius}, {.4 * \markovianBallY}) {$\ |\momentum_0| \, \eqDistribution - \log \uniformRv$};
		\draw[jhuSecondaryOrange] [Latex-Latex] ({(\hamiltonianBallX) + 2 * \ballRadius}, {\markovianBallY}) -- ({(\hamiltonianBallX) + 2 * \ballRadius}, {\hamiltonianBallY});
		\node[right] at ({(\hamiltonianBallX) + 2 * \ballRadius}, {\markovianBallY) + .4 * (\hamiltonianBallY - \markovianBallY)}) {\textcolor{jhuSecondaryOrange}{$\ \potential(\position_0) - \potential_{\min}$}};
	\end{tikzpicture}
		\endpgfgraphicnamed
	\end{minipage}
	\caption{%
		Comparison of Markovian (left) and Hamiltonian (right) zigzag trajectories under the one-dimensional potential $\potential(\position)$.
		Neither zigzag is affected by velocity switch events while going down the potential energy hill, 
		during which the velocity and gradient point in the opposite directions and the relation $\velocity(t) \partial \potential(\position(t)) < 0$ holds.
		During this time, Hamiltonian zigzag stores up kinetic energy converted from potential energy, while Markovian zigzag remains memory-less.
		Once the trajectories reaches the potential energy minimum $\potential_{\min}$ at time $t_{\min} := \argmin_{t \, \geq \, 0} \potential(\position_0 + t \velocity_0)$ and start going ``uphill,'' the accumulated momentum $| \momentum(t_{\min}) | = | \momentum_0 | +  \potential(\position_0) - \potential_{\min} $ keeps Hamiltonian zigzag traveling in the same direction longer than Markovian zigzag.
		The last statement technically holds only ``on average'' due to randomness in the realized values of $| \momentum_0 | \protect \eqDistribution - \log u$.
	}
	\label{fig:markovian_and_hamiltonian_balls_comparison}
\end{figure}

\subsection{Markovian zigzag as an infinite momentum refreshment limit}
\label{sec:markovian_zigzag_as_limit_of_hamiltonian_zigzag}
We now consider a version of Hamiltonian zigzag in which we periodically refresh the momentum by resampling their magnitudes $|\momentum_i(\tau)| \sim \expDist(1)$ while keeping their signs.
This process follows a zigzag path as before, but its inter-event times are now random.
We see from the formula of \eqref{eq:hamiltonian_zigzag_event_time} that, following a momentum magnitude refreshment at time $\tau$, the velocity flip $\velocity_i \gets - \velocity_i$ occurs during the interval $[\tau, \tau + \dt]$ if and only if
\begin{equation*}
|\momentum_i(\tau)|
	\leq \max_{0 \leq t \leq \dt} \left[
		\int_0^{t} \velocity_i(\tau) \spaceBeforePartial \partial_i \potential( \bposition(\tau + s) ) \diff s
	\right].
\end{equation*}
Provided that $\partial_i \potential$ is continuous and $\partial_i \potential( \position(\tau) ) \neq 0$, the sign of $\partial_i \potential( \bposition(\tau + s) )$ stays constant on the interval $s \in [\tau, \tau + \dt]$ for sufficiently small $\dt$, so that
\begin{equation*}
\max_{0 \leq t \leq \dt} \left[
	\int_0^{t} \velocity_i(\tau) \spaceBeforePartial \partial_i \potential( \bposition(\tau + s) ) \diff s
\right] < 0
	\ \text{ or }
	= \int_0^{\dt} \big[ \velocity_i(\tau) \spaceBeforePartial \partial_i \potential( \bposition(\tau + s) ) \big]^+ \diff s.
\end{equation*}
Under these conditions, the probability of the $i$-th velocity flip is therefore
\begin{align*}
&\probability\left\{
	|\momentum_i(\tau)| \leq \max_{0 \leq t \leq \dt} \left[
		\int_0^{t} \velocity_i(\tau) \spaceBeforePartial \partial_i \potential( \bposition(\tau + s) ) \diff s
	\right]
\right\} \\
&\hspace{4em}
	= \probability\left\{
		|\momentum_i(\tau)| \leq \int_0^{\dt} \big[ \velocity_i(\tau) \spaceBeforePartial \partial_i \potential( \bposition(\tau + s) ) \big]^+ \diff s
	\right\} \\
&\hspace{4em}
	= 1 - \exp\left\{
		- \int_0^{\dt} \big[ \velocity_i(\tau) \spaceBeforePartial \partial_i \potential( \bposition(\tau + s) ) \big]^+ \diff s
	\right\}\\
&\hspace{4em}
	= \big[ \velocity_i(\tau) \spaceBeforePartial \partial_i \potential( \bposition(\tau) ) \big]^+ \dt + O(\dt^2).
	\yesnumber
	\label{eq:event_prob_for_hamiltonian_zigzag_over_small interval}
\end{align*}
Equation \eqref{eq:event_prob_for_hamiltonian_zigzag_over_small interval} shows that, immediately following the momentum magnitude refreshment, an $i$-th velocity switch event for Hamiltonian zigzag happens at a rate essentially identical to that of Markovian zigzag as given in \eqref{eq:markovian_zigzag_event_rate}.

Now consider resampling the momentum magnitude after every time interval of size $\dt$ and letting $\dt \to 0$.
Our analysis above suggests that, under this limit, the rate of coordinate-wise velocity switch events for Hamiltonian zigzag converges to $\switchRate_i(\bposition, \bvelocity) = \left[ \velocity_i \spaceBeforePartial \partial_i \potential(\bposition) \right]^+$.
That is, Hamiltonian zigzag becomes \textit{equivalent} to Markovian zigzag under this infinite momentum refreshment limit.

We now turn the above intuition into a rigorous argument.
In Theorem~\ref{thm:convergence_of_hamiltonian_zigzag_to_markovian} below, $D[0, \infty)$ denotes the space of right-continuous-with-left-limit functions from $[0, \infty)$ to $\mathbb{R}^\nParam \times \mathbb{R}^\nParam$ endowed with Skorokhod topology, the canonical space to study convergence of stochastic processes with jumps \citep{billingsley1999convergence_of_prob_measures, ethier2005markov_processes}.
In particular, the convergence $(\bposition_{\dt}, \bvelocity_{\dt}) \to (\bposition, \bvelocity)$ in this space implies the convergence of the ergodic average $\integrationTime^{-1} \int_{0}^{\integrationTime} f(\bposition_{\dt}(t)) \diff t \to \integrationTime^{-1} \int_{0}^{\integrationTime} f(\bposition(t)) \diff t$ for any continuous real-valued function $f$.

\begin{theorem}[Weak convergence]
\label{thm:convergence_of_hamiltonian_zigzag_to_markovian}
Given an initial position $\bposition(0)$ and velocity $\bvelocity(0) \in \{\pm 1\}^\nParam$, consider Hamiltonian zigzag dynamics with $\momentum_i(0) = \velocity_i(0) |\momentum_i(0)|$ with $|\momentum_i(0)| \sim \expDist(1)$ and with momentum magnitude resampling at every $\dt$ interval, i.e.\ $|\momentum_i(n \dt)| \sim \expDist(1)$ for \mbox{$n \in \mathbb{Z}^+$}.
For each $\dt$, let $(\bposition_{\dt}, \bmomentum_{\dt})$ denote the corresponding dynamics and $\bvelocity_{\dt}$ the right-continuous modification of velocity $\diff \bposition_{\dt} / \diff t = \sign(\bmomentum_{\dt})$.
Then, as $\dt \to 0$, the dynamics $(\bposition_{\dt}, \bvelocity_{\dt})$ converges weakly to the Markovian zigzag process in $D[0, \infty)$.
\end{theorem}
This weak convergence result characterizes Markovian zigzag as a special case, albeit only in the limit, of Hamiltonian zigzag. 
As such, it \emph{almost} guarantees that Hamiltonian zigzag, if combined with optimal momentum refreshment schedule, will outperform Markovian zigzag.
At a more practical level, the interpretation of Markovian zigzag as Hamiltonian zigzag with less momentum pinpoints the cause of dramatic differences in their efficiency observed in our numerical examples (Section~\ref{sec:simulation}).

En route to the weak convergence result, we in fact establish a stronger convergence in probability via explicit coupling of Hamiltonian and Markovian zigzag processes:
\begin{theorem}[Strong convergence]
\label{thm:strong_convergence_of_hamiltonian_zigzag_to_markovian}
The Hamiltonian zigzags with momentum magnitude refreshments, as described in Theorem~\ref{thm:convergence_of_hamiltonian_zigzag_to_markovian}, can be constructed so that their position-velocity components $(\bposition_{\dt}, \bvelocity_{\dt})$ converge strongly to the Markovian zigzag in $D[0, T]$ for all $T < \infty$.
More precisely, there exists a family of Hamiltonian zigzags $(\bposition_{\dt}^\superH, \bmomentum_{\dt}^\superH)$ and Markovian zigzags $(\bposition_{\dt}^\superM, \bvelocity_{\dt}^\superM)$ on the same probability space such that, for any $\epsilon > 0$ and $T > 0$,
\begin{equation}
\label{eq:conv_in_probability_definition}
\lim_{\dt \to 0} \probability\!\left\{
	\rho_\integrationTime\!\left[
		(\bposition_{\dt}^\superH, \bvelocity_{\dt}^\superH),
		(\bposition_{\dt}^\superM, \bvelocity_{\dt}^\superM)
	\right] > \epsilon
\right\} = 0,
\end{equation}
where $\bvelocity_{\dt}^\superH$ is the right-continuous modification of $\sign(\bmomentum_{\dt}^\superH)$ and $\rho_\integrationTime(\cdot, \cdot)$ is the Skorokhod metric on $[0, T]$. 
In fact, the two zigzags can be constructed so that
\begin{equation*}
\lim_{\dt \to 0} \probability\!\left\{
		(\bposition_{\dt}^\superH, \bvelocity_{\dt}^\superH) \equiv (\bposition_{\dt}^\superM, \bvelocity_{\dt}^\superM)
		\, \text{ on } \,
		[ 0, \integrationTime ]
	\right\} = 1.
\end{equation*}
\end{theorem}
In the statement above, the distributions of the Hamiltonian zigzags $(\bposition_{\dt}^\superH, \bvelocity_{\dt}^\superH)$ depend on $\dt$ due to momentum magnitude refreshments, but their Markovian counterparts all have the same distribution $(\bposition_{\dt}^\superM, \bvelocity_{\dt}^\superM) \eqDistribution (\bposition^\superM, \bvelocity^\superM)$.
By Theorem~3.1 of \cite{billingsley1999convergence_of_prob_measures}, convergence in the sense of \eqref{eq:conv_in_probability_definition} implies weak convergence $(\bposition_{\dt}^\superH, \bvelocity_{\dt}^\superH) \to (\bposition^\superM, \bvelocity^\superM)$ in $D[0, T]$ for any $T > 0$ and hence, by Theorem~16.7 of \cite{billingsley1999convergence_of_prob_measures}, in $D[0, \infty)$.
In particular, our Theorem~\ref{thm:convergence_of_hamiltonian_zigzag_to_markovian} follows from our Theorem~\ref{thm:strong_convergence_of_hamiltonian_zigzag_to_markovian}, whose proof is in Supplement Section \ref{sec:proof_of_strong_convergence_of_hamiltonian_zigzag_to_markovian}.

\section{Numerical study: two zigzags duel over truncated multivariate Gaussians}
\label{sec:simulation}
Our theoretical result of Section~\ref{sec:markovian_zigzag_as_limit_of_hamiltonian_zigzag} shows Markovian zigzag as essentially equivalent to Hamiltonian zigzag with constant refreshment of momentum magnitude.
As we heuristically argue in Section~\ref{sec:intro} and \ref{sec:comparison_of_two_zigzags}, the loss of full momentum information can make Markovian zigzag more prone to random-walk behavior in the presence of strong dependency among parameters.
We validate this intuition empirically in this section.

We have so far put aside the issue of numerically simulating zigzag trajectories in practice.
The coordinate-wise integrator of \cite{nishimura2020discontinuous_hmc} provides one way to qualitatively approximate Hamiltonian zigzag.
With a suitable modification (Supplement Section~\ref{sec:midpoint_integrator}), the mid-point integrator of \citet{chin2023bouncy_hmc} provides another option.
For exact simulations, however, both zigzags require computing the times of velocity switch events (Line~\ref{line:coord_wise_event_time} in Algorithm~\ref{alg:hamiltonian_zigzag_simulation} and \ref{alg:markovian_zigzag_simulation}).
Hamiltonian zigzag additionally requires computing the integrals of Line~\ref{line:integral_1} and \ref{line:integral_2} for updating momentum.
Being a Markovian process, Markovian zigzag allows the use of Poisson thinning in determining event times \citep{bierkens2018scaling_of_pdmp_sampler}.
This fact makes it somewhat easier to simulate Markovian zigzag, while an efficient implementation remains challenging except for a limited class of models \citep{vanetti2017piecewise_deterministic_mcmc}.

Here we focus on sampling from a truncated multivariate Gaussian, a special yet practically relevant class of targets, on which we can efficiently simulate both zigzags.
Besides simple element-wise multiplications and additions, simulating each linear segment of the zigzags only requires solving $\nParam$ quadratic equations and extracting a column of the Gaussian precision matrix $\bPhi$ (Supplement Section~\ref{sec:zigzags_on_truncated_gaussians}).
This in particular gives the zigzags a major potential advantage, depending on the structure of $\bPhi$, over other state-of-the-art algorithms for truncated Gaussians that require computationally expensive pre-processing operations involving $\bPhi$ \citep{pakman2014truncated_normal_hmc, botev2017truncated_normal}.
In fact, the numerical results of \cite{zhang2022hdtg} indicate zigzag \hmc{} as a preferred choice over the algorithms of \cite{pakman2014truncated_normal_hmc} and \cite{botev2017truncated_normal} in many high-dimensional applications.
\myedit{PointerToMoreComparisions}{%
	Supplement Section~\ref{sec:existing_sampler_from_truncated_gaussians} provides more detailed discussion of how these algorithms compare in their algorithmic complexities and complement the benchmark of \cite{zhang2022hdtg} with additional numerical results.%
}

We compare performances of the two zigzags on a range of truncated Gaussians, consisting of both synthetic and real-data posteriors.
As predicted, Hamiltonian zigzag emerges as a clear winner as dependency among parameters increases.

\subsection{Zigzag-N{\small UTS}: Hamiltonian zigzag with no-U-turn algorithm}
\label{sec:zigzag_nuts}
Given the availability of analytical solutions in simulating zigzag trajectories, Markovian zigzag is completely tuning-free in the truncated Gaussian case.
Hamiltonian zigzag requires periodic momentum refreshments $\momentum_i \sim \laplaceDist(\textrm{scale} = 1)$ for ergodicity, so the integration time $\integrationTime$ in-between refreshments remains a user-specified input.
On the other hand, being a reversible dynamics, Hamiltonian zigzag can take advantage of the no-U-turn algorithm (\nuts{}) of \cite{hoffman2014nuts} to automatically determine an effective integration time.
This way, we only need to supply a base integration time $\baseIntegrationTime$ to Hamiltonian zigzag --- the no-U-turn algorithm will then identify an appropriate integration time $\integrationTime = 2^k \baseIntegrationTime$, where $k \geq 0$ is the height of a binary trajectory tree at which the trajectory exhibits a U-turn behavior for the first time.
We provide in Supplement Section~\ref{sec:nuts_with_reversible_dynamics} the details of how to combine the no-U-turn algorithm with reversible dynamics in general.

With the automatic multiplicative adjustment of the total integration time, the combined Zigzag-\Nuts{} algorithm only requires us to set $\baseIntegrationTime$ as a reasonable underestimate of an optimal integration time.
Based on the intuition that the integration time should be proportional to a width of the target in the least constrained direction \citep{neal2010hmc}, we choose $\baseIntegrationTime$ for Zigzag-\Nuts{} as follows.
In the absence of truncation, this width of the target is proportional to $\eigenValue_{\max}^{1/2}\left( \bPhi^{-1} \right) = \eigenValue_{\min}^{-1/2}(\bPhi)$ where $\eigenValue_{\max}$ and $\eigenValue_{\min}$ denote the largest and smallest eigenvalues, both of which can be computed quickly via a small number of matrix-vector operations $\bm{w} \to \bPhi^{-1} \bm{w}$ or $\bm{w} \to \bPhi \bm{w}$ using the Lanzcos algorithm \citep{meurant2006lancoz-and-cg}. 
For the standard \hmc{} based on Gaussian momentum, an optimal integration time on multivariate Gaussian targets is $\integrationTime = \varpi \eigenValue_{\min}^{-1/2}(\bPhi) / 2$ \citep{bou2018hmc}, where $\varpi \approx 3.14$ denotes Archimedes's constant.
This suggests that $\baseIntegrationTime \approx \eigenValue_{\min}^{-1/2}(\bPhi)$ should be close to the upper end of reasonable base integration times. 
We hence propose a choice $\baseIntegrationTime = \eigenValue_{\min}^{-1/2}(\bPhi) \baseIntegrationTimeMultiplier $ for  $\baseIntegrationTimeMultiplier \leq 1$, where $\baseIntegrationTimeMultiplier$ represents a base integration time relative to the target's width $\eigenValue_{\min}^{-1/2}(\bPhi)$.

In our numerical results, we find that $\baseIntegrationTimeMultiplier = 0.1$ works well in a broad range of problems.
We observe further performance gains from a larger value, i.e.\ $\baseIntegrationTimeMultiplier > 0.1$, if the target is highly constrained.
We use $\baseIntegrationTimeMultiplier = 0.1$ in this section for simplicity's sake, but additional numerical results are available in Supplement Section~\ref{sec:base_integration_time_for_nuts}.

\subsection{Study set-up and efficiency metrics}
\label{sec:simulation_setup_and_efficiency_metric}
The existing empirical evaluations of Markovian zigzag rely on simple low-dimensional target distributions;
consequently, there is great interest in having its performance tested on more challenging higher-dimensional problems \citep{dunson2020hastings_algorithm}.
We start from where \cite{bierkens2018scaling_of_pdmp_sampler} left off --- $256$-dimensional correlated Gaussians (without truncation) --- and first test the two zigzags on synthetic truncated Gaussians of dimension up to $4{,}096$. 
We then proceed to a real-world application, comparing the performances of the two zigzags on a $11{,}235$-dimensional truncated Gaussian posterior.

We compare the two zigzags' performances using effective sample sizes (\ess{}), a well-established metric for quantifying efficiency of Markov chain Monte Carlo (\mcmc{}) algorithms \citep{geyer2011intro_to_mcmc}. 
When assessing a relative performance of two \mcmc{} algorithms, we also need to account for their per-iteration computational costs.
We therefore report \ess{} per unit time, as is commonly done in the literature, where ``time'' refers to the actual time it takes for our code to run and is not to be confused with the time scales of zigzag dynamics.
We note that relative computational speed can vary significantly from one computing environment to another due to various performance optimization strategies used by modern hardware, such as instruction-level parallelism and multi-tiered memory cache \citep{guntheroth2016optimizedCpp, nishimura2022cg_accelerated_gibbs_supp, holbrook2020massive_parallelization}.
For this reason, while \ess{} per time is arguably the most practically relevant metric, we consider an alternative platform-independent performance metric in Supplement Section~\ref{sec:ess_per_event}.

We run Zigzag-\Nuts{} for $25{,}000$ iterations on each synthetic posterior of Section~\ref{sec:synthetic_examples}.
Each iteration of Zigzag-\Nuts{} is more computationally intensive on the real-data posterior of Section~\ref{sec:phylogenetic_probit}, but also mixes more efficiently than on the hardest synthetic posterior.
We hence use a shorter chain of $1{,}500$ iterations on the real-data one.
With these chain lengths, we obtain at least $100$ \ess{} along each coordinate in all our examples.

\myedit{EssCalcExplanataion}{%
For each Markovian zigzag simulation, we collect \mcmc{} samples spaced at time intervals of size $\baseIntegrationTime$, the base integration time for Zigzag-\Nuts{}.
This way, we sample Markovian zigzag at least as frequently as Hamiltonian zigzag along their respective trajectories.
We thus ensure a fair comparison between the two zigzags and, if any, tilt the comparison in favor of Markovian zigzag.
To obtain at least $100$ \ess{} along each coordinate, we simulate Markovian zigzag for $T = 250{,}000 \times \baseIntegrationTime$ on each synthetic posterior and $T = 1{,}500 \times \baseIntegrationTime$ on the real-data posterior, generating $250{,}000$ and $1{,}500$ samples respectively.

Note that, while both zigzags can in theory utilize entire trajectories to estimate posterior quantities of interest \citep{bierkens2019zigzag_original, nishimura2020recycled_hmc}, such  approaches are often impractical in high-dimensional settings due to memory constraints.
In the $11{,}235$-dimensional example of Section~\ref{sec:phylogenetic_probit}, for example, $1{,}500$ iterations of Markovian (Hamiltonian) zigzag undergoes 
$1.6 \times 10^8$ ($2.1 \times 10^8$) 
velocity switch events.
Storing all these event locations would require $1.8$ ($2.4$) \textsc{tb} in 64-bit double precision, while providing little practical benefit because of their high auto-correlations.

We implement both zigzags (Algorithm~\ref{alg:hamiltonian_zigzag_for_truncated_normal} and \ref{alg:markovian_zigzag_for_truncated_normal} in Supplement Section~\ref{sec:zigzags_on_truncated_gaussians}) in the Java programming language as part of the Bayesian phylogenetic software \textsc{beast};
the code and instruction to reproduce the results are available at \url{https://github.com/aki-nishimura/code-for-hamiltonian-zigzag-2024}.
We run each \mcmc{} on a c5.2xlarge instance in Amazon Elastic Compute Cloud, equipped with $4$ Intel Xeon Platinum 8124M processors and $16$ \textsc{gb} of memory. 
For each target, we repeat the simulation $5$ times with different seeds and report \ess{} averaged over these $5$ independent replicates.
\Ess{}'s are computed using the R \textsc{coda} package \citep{plummer2006}.%
}

\subsection{Threshold model posteriors under correlated Gaussian priors}
\label{sec:synthetic_examples}
\cite{bierkens2018scaling_of_pdmp_sampler} consider Gaussian targets with compound symmetric covariance 
\begin{equation}
\label{eq:compound_symmetry_covariance}
\variance(\position_i) = 1, \quad
\covariance(\position_i, \position_j) = \rho \in [0, 1)
\ \text{ for } \, i \neq j.
\end{equation}
Such a distribution can be interpreted as a prior induced by a model $\position_i = \rho^{1 / 2} z + (1 - \rho)^{1 / 2} \epsilon_i$, with shared latent factor $z \sim \normalDist(0, 1)$ and individual variations $\epsilon_i \sim \normalDist(0, 1)$.
We construct truncated Gaussian posteriors by assuming a threshold model
\begin{equation*}
\observation_i = \indicator\{ \position_i > 0 \} - \indicator\{ \position_i \leq 0 \}.
\end{equation*}
An arbitrary thresholding would make it difficult to get any feel of the geometric structure behind the resulting truncated Gaussian posterior.
We hence assume a simple thresholding with $\observation_i = 1$ for all $i$, inducing posteriors constrained to the positive orthant $\{ \position_i > 0 \}$.
We investigate the effect of degree of correlation on zigzags' performance by varying the correlation coefficient, using the values $\rho = 1 - (0.1)^k$ for $k = 0, 1, 2$.


The numerical result for the compound symmetric posteriors, including the i.i.d.\ case $\rho = 0$, is summarized in Table~\ref{tab:ess_per_time_compound_symmetric}.
Since $\position_i$ are exchangeable, we calculate \ess{} only along the first coordinate.
We also calculate \ess{} along the principal eigenvector of $\bPhi^{-1}$ since \hmc{} typically struggles most in sampling from the least constrained direction \citep{neal2010hmc}.
Markovian zigzag face a similar challenge as evidenced by the visual comparison of the two zigzag samples projected onto the principal component (Figure~\ref{fig:zigzag_traceplot}).

\begin{figure}
\centering
\includegraphics[width=\linewidth]{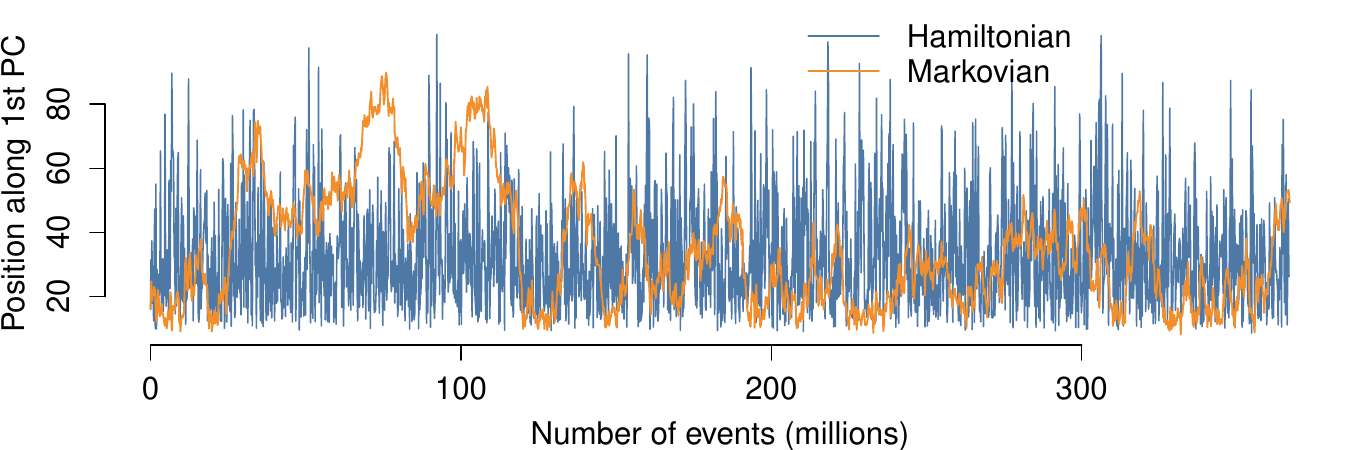}
\caption{%
	Traceplot of the zigzag samples from the $1{,}024$-dimensional compound symmetric posterior \eqref{eq:compound_symmetry_covariance} with $\rho = 0.99$, projected onto the principal component $\principalComp = (1, \ldots, 1) / \sqrt{\nParam}$ via a map $\bposition \to \langle \bposition, \principalComp \rangle$.
	The horizontal axis is scaled to represent the number of velocity switch events.
}
\label{fig:zigzag_traceplot}
\end{figure}

As predicted, Hamiltonian zigzag demonstrates increasingly superior performance over its Markovian counterpart as the correlation increases, delivering $4.5$ to $4.7$-fold gains in relative \ess{} at $\rho = 0.9$ and $40$ to $54$-fold gains at $\rho = 0.99$.
The efficiency gain is generally greater at the higher dimension $\nParam = 1{,}024$.
For the i.i.d.\ case, Hamiltonian zigzag seems to have no advantage.
This is in a sense expected since, on an i.i.d.\ target, both zigzags become equivalent to running $d$ independent one-dimensional dynamics and have no interactions among the coordinates.
In other words, Hamiltonian zigzag's additional momentum plays no role when parameters are independent. 
Also, for a univariate Gaussian target, Markovian zigzag has been shown to induce negative auto-correlations and thus achieve sampling efficiency above that of independent Monte Carlo \citep{bierkens2017limit_theorems_for_zigzag}.

\begin{table}
\begin{minipage}{.35\linewidth}
\caption{%
	\Ess{} per computing time --- relative to that of Markovian zigzag --- under the compound symmetric posteriors.
	We test the algorithms under three correlation parameter values ($\rho = 0, 0.9$, and $0.99$) and two varying dimensions ($\nParam = 256$ and $1{,}024$). 
	\Ess{}'s are calculated along the first coordinate and along the principal eigenvector of $\bPhi^{-1}$, each shown under the labels ``$\position_1$'' and ``PC.''
	\label{tab:ess_per_time_compound_symmetric}
}
\end{minipage}
~
\begin{minipage}{.6\linewidth}
\vspace*{-1.2\baselineskip}
\begin{tabular}[t]{lccccccc}
\toprule
& \multicolumn{5}{c}{Relative \Ess{} per time}\\
\cmidrule(l{3pt}r{3pt}){2-6}
\multicolumn{1}{c}{Compound symmetric} 
& \multicolumn{1}{c}{$\rho = 0$} & \multicolumn{2}{c}{$\rho = 0.9$} & \multicolumn{2}{c}{$\rho = 0.99$} \\
\cmidrule(l{3pt}r{3pt}){2-2} \cmidrule(l{3pt}r{3pt}){3-4} \cmidrule(l{3pt}r{3pt}){5-6}
\multicolumn{1}{c}{} & $x_1$  & $x_1$ & PC & $x_1$ & PC \\
\midrule
\multicolumn{1}{c}{Case: $d = 256$} & \multicolumn{5}{c}{} \\
\rowcolor{markovColor} Markovian & 1 & 1 & 1 & 1 & 1\\
\rowcolor{defaultNutsColor} Zigzag-\Nuts{}
& 0.64 & 4.5 & 4.6 & 41 & 40 \\ 
\rowcolor{hzzManualColor} Zigzag-\Hmc{} $\left( \integrationTime_\textrm{rel} = \sqrt{2} \right)$\hspace*{-1.5ex}
& 5.5 & 46 & 66 & 180 & 180 \\
\rowcolor{white} 
\multicolumn{1}{c}{Case: $d = 1{,}024$} & \multicolumn{5}{c}{} \rule{0pt}{12pt}\\
\rowcolor{defaultNutsColor} Zigzag-\Nuts{} 
& 0.57 & 4.7 & 4.5 & 54 & 54 \\
\rowcolor{hzzManualColor} Zigzag-\Hmc{} $\left( \integrationTime_\textrm{rel}  = \sqrt{2} \right)$\hspace*{-1.5ex}
& 5.6 & 56 & 85 & 300 & 300 \\
\bottomrule
\end{tabular}
\end{minipage}
\end{table}

The compound symmetric targets here have a particularly simple correlation structure;
the covariance matrix can be written as $\bPhi^{-1} = (1 - \rho) \Id + \rho \bm{1} \mathbf{1}^\transpose$ for $\mathbf{1} = (1, \ldots, 1)$, meaning that the probability is tightly concentrated along the principal component and is otherwise distributed symmetrically in all the other directions.
This simple structure in particular allows us to manually identify an effective integration time $\integrationTime$ for Hamiltonian zigzag without too much troubles.
We therefore use this synthetic example to investigate how Zigzag-\Nuts{} perform relative to manually-tuned Hamiltonian zigzag, which we denote as ``Zigzag-\Hmc{}'' in Table~\ref{tab:ess_per_time_compound_symmetric}.

For each of the three targets, we try $\integrationTime = \eigenValue_{\min}^{-1/2}(\bPhi) \integrationTime_\textrm{rel}$ with $\integrationTime_\textrm{rel} = 2^{k / 2}$ for $k = -2, -1, 0, 1, 2$.
We report the \ess{}'s based on $\integrationTime_\textrm{rel} = 2^{1/2}$ in Table~\ref{tab:ess_per_time_compound_symmetric} as we find this choice to yield the optimal \ess{} in the majority of cases. 
We see that manually-optimized Hamiltonian zigzag delivers substantial increases in \ess{} compared to Zigzag-\Nuts{}.
Such efficiency gains are also observed by the authors who proposed alternative methods for tuning \hmc{} \citep{wang2013adaptive_hmc, wu2018faster_hmc}.
The results here indicate that their tuning approaches may be worthy alternatives to Zigzag-\Nuts{} and may further reinforce Hamiltonian zigzag's advantage over Markovian zigzag.

\myedit{ZigzagUnderRotatedCS}{%
	Finally, given that the zigzags's motions are restricted to the discrete set of directions $\bvelocity \in \{\pm 1\}^\nParam$ and that the compound symmetric posteriors happen to be concentrated along one of these directions, one may wonder whether  this coincidental structure affect the above numerical results.
	To answer this question, we conduct additional simulations with rotated versions of the compound symmetric posterior, corresponding to the covariance matrices $\bPhi^{-1} = (1 - \rho) \Id + \rho \thinnerspace \principalComp \principalComp^\transpose$ with the principal components $\principalComp \in \mathbb{R}^\nParam$ drawn uniformly from the $(\nParam - 1)$-dimensional sphere.
	These additional simulations indicate that essentially the same pattern holds in the relative performance of two zigzags and that absolute \ess{} per time changes little across different rotations of the target (Supplement Section~\ref{sec:rotated_compound_symmetric}).
	The latter finding reinforces the theoretical results of \cite{bierkens2023anisotropic} who show that, under a class of anistropic Gaussian targets, any deviation from diagonal covariance results in diffusive behavior in Markovian zigzag.
	On the other hand, the inertia provided by full momentum information appears to endow Hamiltonian zigzag with fundamentally different behavior.%
}

\subsection{Posterior from phylogenetic multivariate probit model}
\label{sec:phylogenetic_probit}
We now consider a 11{,}235-dimensional target arising from the phylogenetic multivariate probit model of  \cite{zhang2021phylo_multi_probit}.
For simplicity's sake, here we describe the model with some simplifications and refer interested readers to the original work for full details.

The goal of \cite{zhang2021phylo_multi_probit} is to learn correlation structure among $\nTraits = 21$ binary biological traits across $ \nTaxa = 535$ \hiv{} viruses while accounting for their shared evolutionary history.
Their model assumes that, conditional on the bifurcating phylogenetic tree $\phylogeny$ informed by the pathogen genome sequences, latent continuous biological traits $\latentData_1, \ldots, \latentData_\nTaxa \in \mathbb{R}^\nTraits$ follows Brownian motion along the tree with an unknown $\nTraits \times \nTraits$ diffusion covariance $\traitCovariance$ (Figure~\ref{fig:phylogenetic_brownian_diffusion}).
The latent traits $\latentData = [\latentData_1, \ldots, \latentData_n]$ map to the binary observation $\bobservation \in \mathbb{R}^{\nTraits \times \nTaxa}$ via the threshold model $\observation_{ij} = \indicator\{ \position_{ij} > 0 \} - \indicator\{ \position_{ij} \leq 0 \}$.

\begin{figure}
\begin{minipage}{.6\linewidth}
\includegraphics[width=\linewidth]{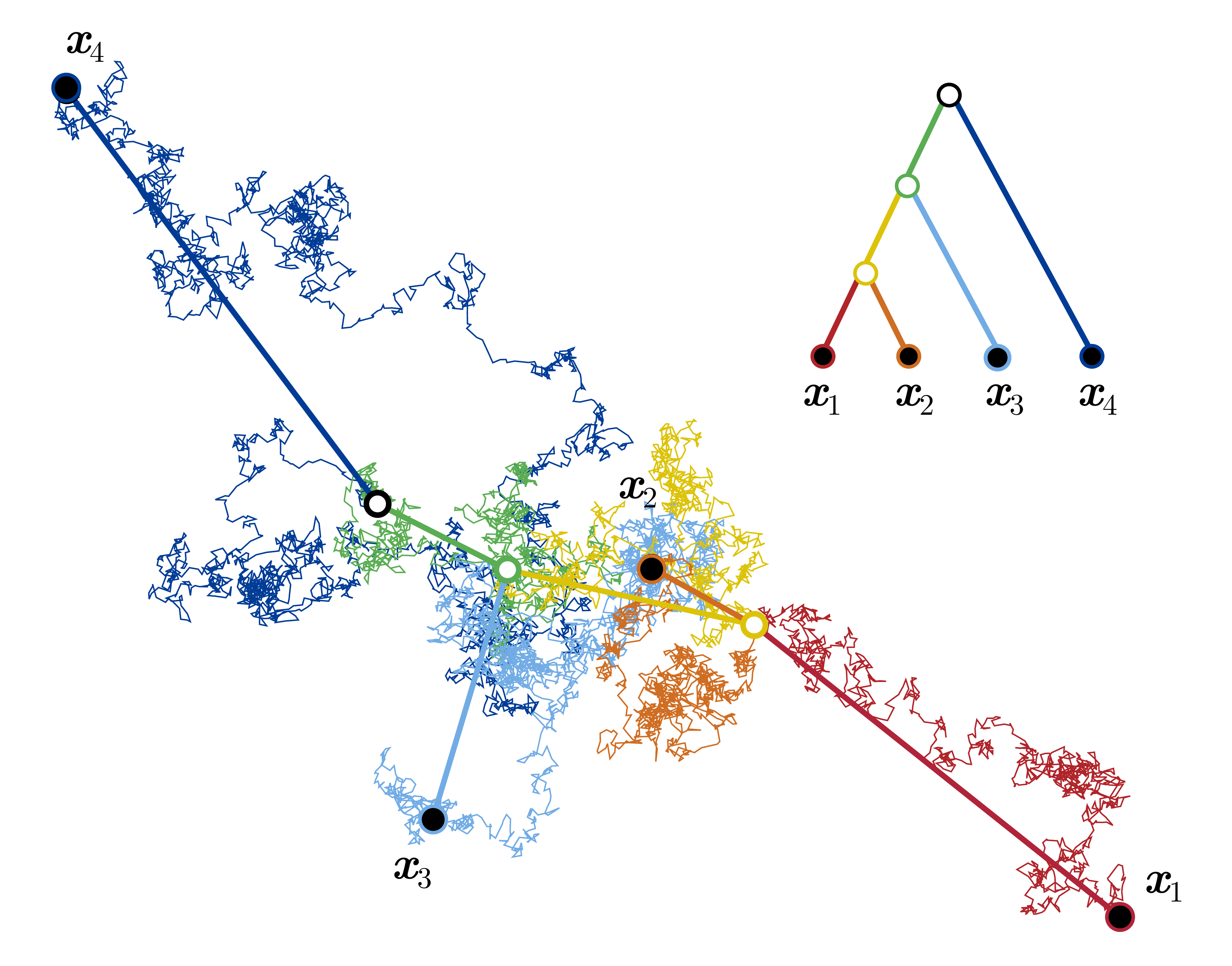}
\end{minipage}
\begin{minipage}{.38\linewidth}
\caption{%
	Example paths of latent biological traits following the phylogenetic diffusion.
	The traits of two distinct organisms evolve together until a branching event.
	Beyond that point, the traits evolve independently but with the same diffusion covariance induced by a shared bio-molecular mechanism.
}
\label{fig:phylogenetic_brownian_diffusion}
\end{minipage}
\end{figure}

The model gives rise to a posterior distribution on the joint space $(\latentData, \traitCovariance, \phylogeny)$ of latent biological traits, diffusion covariance, and phylogenetic tree.
To deal with this complex space, \cite{zhang2021phylo_multi_probit} deploys a Gibbs sampler, the computational bottleneck of which is updating $\latentData$ from its full conditional --- a $11{,}235 = 21 \times 535$ dimensional truncated multivariate Gaussian.
The $\nTraits \nTaxa \times \nTraits \nTaxa$ precision matrix $\bPhi = \bPhi(\traitCovariance, \phylogeny)$ of $\latentData \given \traitCovariance, \phylogeny$ induced by the phylogenetic Brownian diffusion model changes at each Gibbs iteration, precluding use of sampling algorithms that require expensive pre-processings of $\bPhi$.
For the purpose of our simulation, we fix $\traitCovariance$ at the highest posterior probability sample and $\phylogeny$ at the maximum clade credibility tree 
 obtained from the prior analysis by \cite{zhang2021phylo_multi_probit}.
Our target then is the distribution $\latentData \given \traitCovariance, \phylogeny \sim \normalDist(\bmu(\traitCovariance, \phylogeny), \bPhi^{-1}(\traitCovariance, \phylogeny))$ truncated to $\{ \sign(\bposition) = \bobservation \}$.
There are $404$ ($3.6\%$) missing entries in $\bobservation$ and the target remains unconstrained along these coordinates.
\myedit{PreconditioningForHivApplication}{%
	The model is parametrized in such a way that the marginal variances of $\bposition \given \traitCovariance, \phylogeny$ are comparable across the coordinates;
	the diagonal preconditioning technique for \hmc{} \citep{stan_manual, nishimura2020discontinuous_hmc}, therefore, plays little role in this application (Supplement Section~\ref{sec:preconditioning}).%
}
The target distribution parameters $\bmu$, $\bPhi$, and $\bobservation$ are available on a Zenodo repository at \url{http://doi.org/10.5281/zenodo.4679720}.

On this real-world posterior, Hamiltonian zigzag again outperforms its Markovian counterpart, delivering $6.5$-fold increase in the minimum \ess{} across the coordinates and $19$-fold increase in the \ess{} along the principal eigenvector of $\bPhi^{-1}(\traitCovariance, \phylogeny)$ (Table~\ref{tab:ess_per_time_phylo_probit}). 
There is no simple way to characterize the underlying correlation structure for this non-synthetic posterior, so we provide a histogram of the pairwise correlations in Figure~\ref{fig:pairwise_corr_hist_for_phylo_probit} as a crude descriptive summary.
We in particular find that only $0.00156\%$ of correlations have magnitudes above $0.9$. 
However, the joint structure, truncation, and high-dimensionality apparently make for a complex target, which Hamiltonian zigzag can explore more efficiently by virtue of its full momentum.
 
\begin{table}
\centering
\begin{minipage}{.55\linewidth}
\begin{tabular}[t]{lcccc}
\toprule
\multirow{2}{*}{Phylogenetic probit \hspace*{-1ex} \rule{0pt}{12pt}} 
& \multicolumn{2}{c}{Relative \Ess{} per time} \\
\cmidrule(l{3pt}r{3pt}){2-3} 
\cmidrule(l{3pt}r{3pt}){4-5}
& \hspace{2ex} min \hspace{2ex} & PC\\
\midrule
\rowcolor{markovColor} Markovian & 1 & 1\\
\rowcolor{defaultNutsColor} Zigzag-\Nuts{}
& 6.5 & 19\\
\bottomrule
\end{tabular}
\end{minipage}
~
\begin{minipage}{.4\linewidth}
\vspace*{.2\baselineskip}
\caption{%
	Relative \ess{} per computing time under the phylogenetic probit posterior ($\nParam = 11{,}235$).
	The ``min'' label indicates the minimum \ess{} across all the coordinates.
	\label{tab:ess_per_time_phylo_probit}
}
\end{minipage}
\end{table}

\begin{figure}
\centering
\hspace*{-.015\linewidth}
\begin{minipage}{.5\linewidth}
\includegraphics[width=\linewidth]{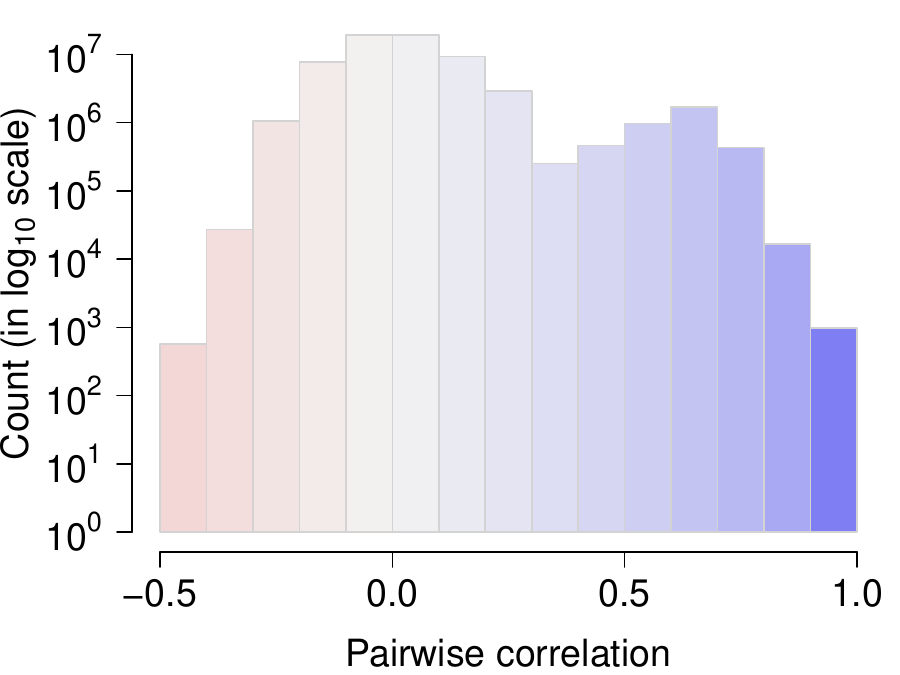}
\end{minipage}
\nobreak\hspace{.1em} 
\begin{minipage}{.49\linewidth}
\caption{%
	Histogram of the pairwise correlations, i.e. the upper off-diagonal entries of $\bPhi^{-1}(\traitCovariance, \phylogeny)$, in the phylogenetic probit posterior.
	The vertical axis is in $\log_{10}$ scale.
	No correlation falls outside the horizontal plot range $[-0.5, 1.0]$, with the smallest and largest correlation being $-0.416$ and $0.989$.
	Out of $\nTraits \nTaxa (\nTraits \nTaxa - 1) / 2 \approx 63.1 \times 10^6$ correlations,  
	$55.3 \times 10^6$ ($87.6\%$) lie within $[-0.2, 0.2]$ and $987$ ($0.00156\%$) within $[0.9, 1]$.
}
\label{fig:pairwise_corr_hist_for_phylo_probit}
\end{minipage}
\end{figure}


\FloatBarrier
\newcommand{\nState}{n}
\newcommand{\nFlow}{m}
\newcommand{\cyclicGraph}{G}
\newcommand{\discreteSpace}{D}
\newcommand{\flow}{\mathcal{F}}
\section{Discussion}
\label{sec:discussion}
In recent years, both piecewise deterministic Markov process (\pdmp{}) samplers and \hmc{} have garnered intense research interests as potential game changers for Bayesian computation.
In this article, we established that one of the most prominent \pdmp{} is actually a close cousin of \hmc{}, differing only in the amount of ``momentum'' they are born with.
This revelation provided novel insights into relative performance of the two samplers, demonstrated via the practically relevant case of truncated multivariate Gaussian targets.

The uncovered kinship between two zigzags begs the question: \emph{is there a more general relationship between \pdmp{} and \hmc{}}?
Searching for an affirmative answer to this question would require us to go beyond the current \hmc{} framework based on classical Hamiltonian dynamics.
Discontinuous changes in velocity seen in \pdmp{} cannot be imitated by smooth Hamiltonian dynamics; 
such behavior is possible under Hamiltonian zigzag only because of its non-differentiable momentum distribution.

In a sense, our result shows that \hmc{}'s momentum is really made of two components: direction and magnitude of inertia.
\Pdmp{}'s velocity on the other hand consists only of direction.
We suspect it is possible to introduce a notion of inertia magnitudes to \pdmp{} in the form of auxiliary parameters --- these parameters can interact with the main parameters so as to emulate the behavior of \hmc{}, preserving the total log-density and storing inertia obtained from potential energy downhills for later use.

While we leave more thorough investigations for future research, we now demonstrate how the above idea translates into a novel non-reversible algorithm based on Hamiltonian-like dynamics in a discrete space. 
We illustrate the approach on a cyclic graph $\cyclicGraph$ with $\nState$ vertices, i.e.\ $\nState$ discrete states placed along a circle.
For $\position \in \cyclicGraph$, we use $\position + 1$ to denote the neighbor in the clockwise direction and $\position - 1$ in the counter-clockwise direction.

\myedit{DiscreteHamiltonianDynamicsPartOne}{%
	We first augment the position space $\{ \position \in \cyclicGraph \}$ with ``direction'' $\velocity \in \{-1, +1\}$ and ``inertia'' $|\momentum| \geq 0$ ---
	there is no native notion of ``momentum $\momentum$'' in this setting, but we use the notation and terminology to draw parallels with Hamiltonian zigzag.
	On this augmented space, we define \textit{discrete Hamiltonian dynamics} with associated ``potential energy'' $\potential(\position)$, whose trajectory $\{ \position(t), \velocity(t), |\momentum|(t) \}$ evolves as follows.
	At each time $t \in \mathbb{Z}$, the next state is given by%
}
\begin{equation*}
\begin{gathered}
\position(t + 1) = \position(t) + v(t), \ \
\velocity(t + 1) = \velocity(t), \ \ \\
|\momentum|(t + 1) =  |\momentum|(t) - \potential(\position(t) + \velocity(t)) + \potential(\position(t))
\end{gathered}
\end{equation*}
if $|\momentum|(t) \geq \potential(\position(t) + \velocity(t)) - \potential(\position(t))$; otherwise,
\begin{equation*}
\position(t + 1) = \position(t), \ \
\velocity(t + 1) = - \velocity(t), \ \
|\momentum|(t + 1) =  |\momentum|(t).
\end{equation*}
\myedit{DiscreteHamiltonianDynamicsPartTwo}{%
	Figure~\ref{fig:hamiltonian_dynamics_on_graph} visually illustrates behavior of this discrete dynamics, which has much conceptual similarity to one-dimensional Hamiltonian zigzag as studies in Section~\ref{sec:comparison_of_two_zigzags} and illustrated in Figure~\ref{fig:markovian_and_hamiltonian_balls_comparison}. 
	In particular, the dynamics by design conserves the sum $\potential(\position) + | \momentum | $, with the ``inertia'' $|\momentum|$ playing the role analogous to that of the kinetic energy $\kinetic(\momentum)$ in the continuous case.
	It is straightforward to verify that the dynamics is reversible and volume-preserving, admits $\pi(\position, \velocity, |\momentum|) = \frac{1}{2} \positionMarginal(\position) \exp(-|\momentum|)$ for $\positionMarginal(\position) \propto \exp\!\left(- \potential(\position) \right)$ as an invariant distribution, and thus constitutes a valid transition kernel \citep{fang2014compressible_hmc, nishimura2020discontinuous_hmc}.%
}

\begin{figure}
\centering
\begin{minipage}{.39\linewidth}
\includegraphics[page=1, width=\linewidth]{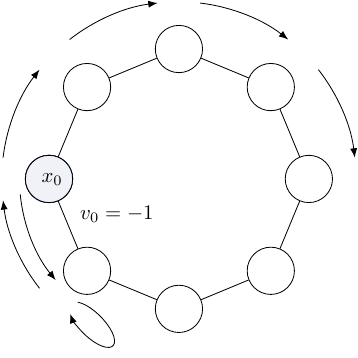}
\end{minipage}
~~
\begin{minipage}{.57\linewidth}
\includegraphics[page=2, width=\linewidth]{Figure/hamiltonian_dynamics_on_graph}
\end{minipage}

\caption{%
	Example trajectory of discrete Hamiltonian dynamics on a cyclic graph $\cyclicGraph$ along with associated potential energy $\{ \potential(\position) \}_{\position \in \cyclicGraph}$.
	The trajectory initially proceeds counterclockwise as directed by the velocity $\velocity_0 = -1$. 
	It does not have sufficient inertia to reach $\position_0 - 2$, however, and hence reverses its course during the second step. 
	The conservation of $\potential(\position) + | \momentum |$ implies that the trajectory then continues clockwise at least to $\position_0 + 4$. 
}
\label{fig:hamiltonian_dynamics_on_graph}
\end{figure}

\myedit{DiscreteHamiltonianDynamicsPartThree}{%
	Now consider adding a ``partial momentum refreshment'' to the above dynamics, drawing $|\momentum|(t) \sim \expDist(1)$ and flipping the sign of $\velocity(t)$ with probability $r \in [0, 1]$ at the beginning of each time step.
	Remarkably, the position-velocity marginal of this dynamics coincides with the general form of non-reversible algorithms presented in Section~5 of \cite{diaconis2000nonreversble_chain}.
	Under a ``full momentum refreshment'' of drawing $|\momentum|(t) \sim \expDist(1)$ and $\velocity(t) \sim \unifDist(\{-1, +1\})$, the dynamics reduces to a Metropolis algorithm with symmetric proposals to the neighbors.
	We can analogously construct Hamiltonian-like dynamics on a more general discrete space, adding partial momentum refreshment to which recovers the fiber algorithm of \cite{diaconis2000nonreversble_chain} and its generalization by \cite{herschlag2020nonreversible_districting_maps}.%
}

Despite using \hmc{} as an inspiration, \cite{diaconis2000nonreversble_chain} and subsequent work have failed to recognize the actual \hmc{} analogues lurking behind their non-reversible methods.
This is understandable --- the notion of momentum from smooth Hamiltonian dynamics does not easily transfer to discrete spaces or non-differentiable paths of \pdmp{}.
On the other hand, having established the explicit link between the two zigzags and having identified the dual role of momentum in providing both direction and inertia, we only need a bit of outside-the-box thinking to extend the idea to discrete spaces.

A unified framework for the \hmc{} and other non-reversible paradigms would present a number of opportunities. 
For example, the introduction of momentum to \cite{herschlag2020nonreversible_districting_maps}'s flow-based algorithm could improve its performance by reducing random walk behavior along each flow.
There could also be cross-fertilization of ideas between continuous- and discrete-space methods.
The elaborate procedures developed in the \pdmp{} literature for choosing the next direction at an event time \citep{fearnhead2018piecewise_deterministic_process, wu2020coordinate_sampler}, for example, is likely applicable to the flow-based algorithm in discrete spaces.
All in all, our revelation of the two zigzags' kinship --- and the insights generated from it --- will likely catalyze further novel developments in Monte Carlo methods.

{\spacingset{1.4}
\bibliographystyle{agsm}
\bibliography{zigzag_hmc}
}
\newpage

%

\renewcommand{\thesection}{S\arabic{section}}
\renewcommand{\thetable}{S\arabic{table}}
\renewcommand{\thefigure}{S\arabic{figure}}
\renewcommand{\theequation}{S\arabic{equation}}
\renewcommand{\thealgorithm}{S\arabic{algorithm}}
\renewcommand{\thepage}{S\arabic{page}}

\setcounter{section}{0}
\setcounter{figure}{0}
\setcounter{table}{0}
\setcounter{equation}{0}
\setcounter{algorithm}{0}
\setcounter{page}{1}

{
  \bigskip
  \bigskip
  \bigskip
  \begin{center}
    {\LARGE\bf Supplement to ``\titleString''}
  \end{center}
}

\section{Proof of Theorem~\ref{thm:existence_and_propeties_of_hamiltonian_zigzag}}
\label{sec:proof_of_hamiltonian_zigzag_properties}

In the proofs below, we use the following standard terminologies.
A \textit{solution operator} $\{ \solutionOp_t \}_{t \geq 0}$ of a differential equation is a family of maps such that $\solutionOp_0$ acts as the identity and $( \bposition_t, \bmomentum_t) = \solutionOp_t(\bposition, \bmomentum)$ for $t \geq 0$ constitutes a solution of the given equation.
The dynamics associated with a differential equation is \textit{time-reversible} if $\momentumFlipOp \circ \solutionOp_t = (\momentumFlipOp \circ \solutionOp_t)^{-1}$ where $\momentumFlipOp$ is the momentum flip operator $\momentumFlipOp(\bposition, \bmomentum) = (\bposition, - \bmomentum)$.
In other words, the forward dynamic followed by the momentum flip can be undone by the exact same operation.
Also, the dynamics is \textit{symplectic} if
\begin{equation}
\label{eq:symplecticity_definition}
\frac{
	\partial (\bposition_t, \bmomentum_t)
}{
	\partial (\bposition, \bmomentum)
}^\transpose
\begin{bmatrix}
	\bm{0} & \ \Id  \, \\
	- \Id  & \ \bm{0} \,
\end{bmatrix}
\frac{
	\partial (\bposition_t, \bmomentum_t)
}{
	\partial (\bposition, \bmomentum)
}
	=
	\begin{bmatrix}
		\bm{0} & \ \Id  \, \\
		- \Id  & \ \bm{0} \,
	\end{bmatrix}
\end{equation}
for all $t \geq 0$, where $\Id$ denote the $\nParam \times \nParam$ identity matrix.
The relation \eqref{eq:symplecticity_definition} can be equivalently expressed as \protect\citepSupp{deGosson2011symplectic_methods}:
\begin{equation}
\label{eq:symplecticity_definition_simplified}
\positionPosition^\transpose \momentumPosition, \
\positionMomentum^\transpose \momentumMomentum
\ \text{ are symmetric; } \ \
\positionPosition^\transpose \momentumMomentum - \momentumPosition^\transpose \positionMomentum = \Id.
\end{equation}
\myedit{LocallyDefinedProperties}{%
	Note that both time-reversibility and symplecticity are locally defined properties and hence make sense for dynamics defined on a subset of $\mathbb{R}^{2\nParam}$.%
}%

\begin{proof}[Proof of existence and uniqueness]
As explained in Section~\ref{sec:laplace_momentum_based_hamiltonian_dynamics}, the solution is completely characterized by \eqref{eq:hamiltonian_zigzag_position_dynamics} -- \eqref{eq:hamiltonian_zigzag_position_dynamics_after_bounce} as long as the trajectory stays away from the problematic set $\problematicSet = \bigcup_i \problematicSet_i$ for $\problematicSet_i$ as in \eqref{eq:problematic_set}.
The conclusions therefore follow once we show the existence of the measure zero set $\zeroMeasureSet$, starting away from which  the trajectory described in Section~\ref{sec:laplace_momentum_based_hamiltonian_dynamics} avoids $\problematicSet$ for all $t \geq 0$.

We construct $\zeroMeasureSet$ as a union of lower dimensional manifolds by backtracking states that may end up in $\problematicSet$ under the dynamics of \eqref{eq:hamilton_under_laplace_momentum}.
For brevity, we omit the adverb ``at most'' when referring to manifold dimensions and take the dimension of $\{ \partial_i \potential(\bposition) = 0 \}$ to be exactly $d - 1$, but the arguments equally hold when some of the manifolds are of lower dimensions.
We first observe that, being a transversal intersection of two $(\nParam - 1)$-dimensional manifolds $\{ (\bposition, \bmomentum) : \partial_i \potential(\bposition) = 0 \}$ and $\{ (\bposition, \bmomentum) : \momentum_i = 0 \}$, the set $\problematicSet_i$  is a $(\nParam - 2)$-dimensional manifold.
We then consider a dynamics similar to \eqref{eq:hamilton_under_laplace_momentum} but with a fixed velocity $\bvelocity \in \{ \pm 1 \}^\nParam$:
\begin{equation}
\label{eq:fixed_velocity_dynamics}
\frac{\diff \bposition}{\diff t}
	= \bvelocity, \quad
\frac{\diff \bmomentum}{\diff t}
	= - \nabla U(\bposition).
\end{equation}
Let $\solutionOp_{\bvelocity, -t}$ for $t \geq 0$ denote the solution operator of \eqref{eq:fixed_velocity_dynamics} backward in time i.e.\
\begin{equation*}
\solutionOp_{\bvelocity, -t}(\bposition, \bmomentum)
	= \left(
		\bposition - t \bvelocity, \ \
		\bmomentum + \int_{0}^{t} \nabla \potential(\bposition - s \bvelocity) \diff s
	\right).
\end{equation*}
By our construction, the set $\zeroMeasureSet^{(1)} = \bigcup_{i, \, \bvelocity \, \in \, \{ \pm 1 \}^\nParam} \zeroMeasureSet_{i, \bvelocity}^{(1)}$ for
\begin{equation}
\label{eq:problematic_set_under_backward_dynamics}
\zeroMeasureSet_{i, \bvelocity}^{(1)} = \textstyle \bigcup_{t \, \geq \, 0} \solutionOp_{\bvelocity, -t}(\problematicSet_i)
\end{equation}
contains all the initial states from which the dynamics of \eqref{eq:hamilton_under_laplace_momentum} reaches $\problematicSet$ at the first velocity switch event.
In other words, for all initial conditions $(\bposition(0), \bmomentum(0)) \not\in \zeroMeasureSet^{(1)}$, the dynamics of \eqref{eq:hamilton_under_laplace_momentum} circumvents the problematic set $\problematicSet$ at least until the second velocity switch event.
Moreover, each $\zeroMeasureSet_{i, \bvelocity}^{(1)}$ has measure zero because it is an image of the $(\nParam - 1)$-dimensional manifold $[0, \infty) \times \problematicSet_i$ under the continuously differentiable map $(t, \bposition, \bmomentum) \to \solutionOp_{\bvelocity, -t}(\bposition, \bmomentum)$.

We now address the question of whether the dynamics continues to avoid the problematic sets after the first event.
Even if a trajectory starts away from the set $\zeroMeasureSet^{(1)}$, it could happen that $(\bposition(\tau), \bmomentum(\tau)) \in \zeroMeasureSet_{i, \bvelocity}^{(1)}$ at some time point $\tau$, in which case the trajectory could encounter $\problematicSet_i$ at the next velocity switch event.
Such a trajectory, having started outside $\zeroMeasureSet_{i, \bvelocity}^{(1)}$, can only enter $\zeroMeasureSet_{i, \bvelocity}^{(1)}$ at the moment of velocity switch 
and hence necessarily passes through the set
\begin{equation*}
\textstyle \bigcup_{j} \left[
	\zeroMeasureSet_{i, \bvelocity}^{(1)} \cap
	\{ (\bposition, \bmomentum) : \momentum_j = 0 \}
\right].
\end{equation*}
Consequently, the set $\zeroMeasureSet^{(2)} = \bigcup_{i, \, \bm{w} \, \in \, \{ \pm 1 \}^\nParam} \zeroMeasureSet_{i, \bm{w}}^{(2)}$ for
\begin{equation*}
\zeroMeasureSet_{i, \bm{w}}^{(2)}
	= \textstyle \bigcup_{t \, \geq \, 0} \solutionOp_{\bm{w}, -t}\left(
		\bigcup_{j, \bvelocity} \left[
			\zeroMeasureSet_{i, \bvelocity}^{(1)} \cap
			\{ (\bposition, \bmomentum) : \momentum_j = 0 \}
		\right]
	\right)
\end{equation*}
contains all the initial states from which the dynamics of \eqref{eq:hamilton_under_laplace_momentum} reaches $\problematicSet$ at the second velocity switch event.
We will now show that each $\zeroMeasureSet_{i, \bvelocity}^{(1)} \cap \{ (\bposition, \bmomentum) : \momentum_j = 0 \}$ is contained in a $(\nParam - 2)$-dimensional manifold.
Following the same reasoning as for $\zeroMeasureSet_{i, \bvelocity}^{(1)}$, we can then conclude that each $\zeroMeasureSet_{i, \bm{w}}^{(2)}$ has measure zero.

To this end, as we already know $\problematicSet_j$ is a $(\nParam - 2)$-dimensional manifold, it suffices to show that the intersection of $\zeroMeasureSet_{i, \bvelocity}^{(1)} \cap \{ \momentum_j = 0 \}$ with $\problematicSet_j^\complement$ belongs to a $(\nParam - 2)$-dimensional manifold.
This follows from the fact that the two manifolds $\zeroMeasureSet_{i, \bvelocity}^{(1)}$ and $\{ \momentum_j = 0 \}$ intersect transversally away from $\problematicSet_j$.
To show this, we study the tangent spaces at $(\bposition, \bmomentum) \in \zeroMeasureSet_{i, \bvelocity}^{(1)} \cap \{ \momentum_j = 0 \} \cap \problematicSet_j^\complement$.
The tangent space of $\{ \momentum_j = 0 \}$ spans all but the direction of the standard basis $\bm{e}_{\nParam + j}$.
On the other hand, by our construction \eqref{eq:problematic_set_under_backward_dynamics}, the tangent space of $\zeroMeasureSet_{i, \bvelocity}^{(1)}$ at $(\bposition, \bmomentum)$ includes the time derivative $\partial_t \solutionOp_{\bvelocity, - t}(\bposition, \bmomentum)$ evaluated at $t = 0$, whose $(\nParam + j)$-th component is given by
\begin{equation*}
\left[ \partial_t \solutionOp_{\bvelocity, -0}(\bposition, \bmomentum) \right]_{\nParam + j}
	= \big[
		(- \bvelocity, \nabla \potential(\bposition))
	\big]_{\nParam + j}
	= \partial_j \potential(\bposition)
	\neq 0,
\end{equation*}
where the last inequality holds because $(\bposition, \bmomentum) \not\in \problematicSet_j$.
This proves that $\zeroMeasureSet_{i, \bvelocity}^{(1)} \cap \{ \momentum_j = 0 \}$ is contained in a $(\nParam - 2)$-dimensional manifold and hence that each $\zeroMeasureSet_{i, \bm{w}}^{(2)}$ has measure zero.

We have thus far shown that, when starting from the initial states outside the sets of measure zero $\zeroMeasureSet^{(1)}$ and $\zeroMeasureSet^{(2)}$, the dynamics of \eqref{eq:hamilton_under_laplace_momentum} avoids the problematic set $\problematicSet$ at least until the second velocity switch event.
We can continue the construction in the same manner and, for $k \geq 2$, quantify a set $\zeroMeasureSet^{(k + 1)}$ of measure zero that include all the initial states from which the dynamics may encounter $\problematicSet$ at the $(k + 1)$-th velocity switch event.
We obtain the set in the theorem statement by taking the union $\zeroMeasureSet = \bigcup_i \zeroMeasureSet^{(i)}$.
\end{proof}

\begin{proof}[Proof of time-reversibility and symplecticity]
Since compositions of time-reversible and symplectic maps are again time-reversible and symplectic, by dividing up the interval $[0, t]$ into smaller pieces if necessary, we can without loss of generality assume that there is at most one velocity switch event in the interval $[0, t]$.

For time-reversibility, we need to show that Hamiltonian zigzag satisfies, in the language of Algorithm~\ref{alg:hamiltonian_zigzag_simulation},
\begin{equation*}
\begin{aligned}
(\bposition, - \bmomentum)
	&= \textproc{HamiltonianZigzag}(\bposition^*, - \bmomentum^*, t) \\
	&\hspace*{-1em}\text{ where }
		(\bposition^*, \bmomentum^*) = \textproc{HamiltonianZigzag}(\bposition, \bmomentum, t).
\end{aligned}
\end{equation*}
Verifying this is straightforward, so we focus on proving symplecticity.

We establish the symplecticity by directly calculating the derivatives of the solution $(\bposition_t, \bmomentum_t)$ with respect to the initial condition $(\bposition, \bmomentum) \not\in \zeroMeasureSet$.
In case there is no velocity switch event, we have
\begin{equation*}
\bposition_t = \bposition + t \bvelocity, \quad
\bmomentum_t = \bmomentum - \int_0^t \nabla \potential(\bposition + s \bvelocity) \diff s.
\end{equation*}
Differentiating $(\bposition_t, \bmomentum_t)$ with respect to $(\bposition, \bmomentum)$, we obtain
\begin{equation}
\label{eq:symplecticity_proof_derivatives_under_no_event}
\begin{aligned}
\positionPosition = \Id, \quad
\positionMomentum = \bm{0}, \quad
\momentumPosition = \int_{0}^{t} \nabla^2 \potential(\bposition + s \bvelocity) \diff s, \quad
\momentumMomentum = \Id,
\end{aligned}
\end{equation}
where $\nabla^2 \potential$ denotes the symmetric Hessian matrix.
The symplecticity condition \eqref{eq:symplecticity_definition_simplified} is easily verified from \eqref{eq:symplecticity_proof_derivatives_under_no_event}.

We now deal with the case in which there is one velocity switch event on $[0, t]$.
In this case, we have
\begin{equation}
\label{eq:symplectic_proof_solution_map}
\begin{aligned}
\bposition_t
	&= \bposition+ t^* \bvelocity + (t - t^*) \bvelocity^*
	\ \, \text{ for } \ \,
	\bvelocity = \sign(\bmomentum), \\
\bmomentum_t
	&= \bmomentum
	+ \int_{0}^{t^*} \nabla \potential(\bposition + s \bvelocity) \diff s
	+ \int_{t^*}^{t} \nabla \potential\big( \bposition + t^* \bvelocity + (s - t^*) \bvelocity^* \big) \diff s,
\end{aligned}
\end{equation}
where $t^*$ and $\bvelocity^*$ satisfy the following relations for some index $\spEventIndex$:
\begin{gather}
|\momentum_{\spEventIndex}|
	= \int_0^{t^*} \velocity_{\spEventIndex} \partial_{\spEventIndex} \potential( \bposition + s \bvelocity) \diff s,
	\label{eq:symplecticity_proof_event_time} \\
\velocity_{\spEventIndex}^* = - \velocity_{\spEventIndex}
	\ \, \text{ and } \ \,
	\velocity_j^* = \velocity_j \text{ for all $j \neq \spEventIndex$}.
	\nonumber
\end{gather}
To simplify the calculations to follow, we re-express \eqref{eq:symplectic_proof_solution_map} as
\begin{equation}
\label{eq:symplecticity_proof_solution_map_simplified}
\begin{aligned}
\bposition_s
	&= \bposition + s \bvelocity^* + 2 t^* \velocity_{\spEventIndex} \bm{e}_{\spEventIndex}
	\ \, \text{ for } \ s \in (t^*, t],
	 \\
\bmomentum_t
	&= \bmomentum
	+ \int_{0}^{t^*} \nabla \potential(\bposition + s \bvelocity) \diff s
	+ \int_{t^*}^{t} \nabla \potential\big( \bposition_s \big) \diff s.
\end{aligned}
\end{equation}
We differentiate \eqref{eq:symplecticity_proof_solution_map_simplified} using Lemma~\ref{lem:differentiability_of_event_time} below and the chain rule.
This yields
\begin{align}
\label{eq:symplecticity_proof_position_position_simplified}
\positionPosition
	&= \positionPosition[s]
	= \Id + 2 \velocity_{\spEventIndex} \bm{e}_{\spEventIndex} \frac{\partial t^*}{\partial \bposition}
	\ \ \text{ for } \ s \in [t^*, t],
	\\
\positionMomentum
	&= 2 \velocity_{\spEventIndex} \bm{e}_{\spEventIndex} \frac{\partial t^*}{\partial \bmomentum}
	= \frac{2 \velocity_{\spEventIndex}}{\partial_{\spEventIndex} \potential(\bposition^*)}
		\bm{e}_{\spEventIndex} \bm{e}_{\spEventIndex}^\transpose,
	\label{eq:symplecticity_proof_position_momentum_simplified} \\
\momentumPosition
	&= \nabla \potential(\bposition^*) \frac{\partial t^*}{\partial \bposition}
	- \nabla \potential(\bposition^*) \frac{\partial t^*}{\partial \bposition}
	+ \int_{0}^{t^*} \nabla^2 \potential(\bposition + s \bvelocity) \diff s
	+ \int_{t^*}^{t} \nabla^2 \potential (\bposition_s) \, \positionPosition[s] \diff s
	\nonumber \\
&= \int_{0}^{t^*} \nabla^2 \potential(\bposition + s \bvelocity) \diff s
	+ \left( \int_{t^*}^{t}
			\nabla^2 \potential(\bposition_s)
		\diff s \right)
		\positionPosition,
	\label{eq:symplecticity_proof_momentum_position_simplified} \\
\momentumMomentum
	&= \Id
	+ \nabla \potential(\bposition^*) \frac{\partial t^*}{\partial \bmomentum}
	- \nabla \potential(\bposition^*) \frac{\partial t^*}{\partial \bmomentum}
	+ \int_{t^*}^{t}
			\nabla^2 \potential\big( \bposition_s \big)
			\positionMomentum[s]
		\diff s
	\nonumber \\
	&= \Id
		+ \frac{2 \velocity_{\spEventIndex}}{\partial_{\spEventIndex} \potential(\bposition^*)}
			\left( \int_{t^*}^{t}
				\nabla^2 \potential\big( \bposition_s \big)
			\diff s \right)
			\bm{e}_{\spEventIndex} \bm{e}_{\spEventIndex}^\transpose.
		\label{eq:symplecticity_proof_momentum_momentum_simplified}
\end{align}

We now check the three conditions \eqref{eq:symplecticity_definition_simplified} for symplecticity using the relations \eqref{eq:symplecticity_proof_position_position_simplified} -- \eqref{eq:symplecticity_proof_momentum_momentum_simplified} and \eqref{eq:symplecticity_proof_event_time_derivative}.
The symmetry of $(\partial \bposition_t / \partial \bmomentum)^\transpose (\partial \bmomentum_t / \partial \bmomentum)$ is trivial.
The symmetry of $(\partial \bposition_t / \partial \bposition)^\transpose (\partial \bmomentum_t / \partial \bposition)$ follows
by noting that
\begin{align*}
\positionPosition^\transpose \momentumPosition
	&= \positionPosition^\transpose
		\left( \int_{0}^{t^*} \nabla^2 \potential(\bposition + s \bvelocity) \diff s \right)
		+ \positionPosition^\transpose
			\left( \int_{t^*}^{t}
				\nabla^2 \potential\big( \bposition_s \big)
			\diff s \right)
			\positionPosition \\
	&= \int_{0}^{t^*} \! \nabla^2 \potential(\bposition + s \bvelocity) \diff s
	+ 2 \velocity_{\spEventIndex} \partial_{\spEventIndex} \potential(\bposition^*)
		\frac{\partial t^*}{\partial \bposition}^\transpose \! \frac{\partial t^*}{\partial \bposition}
	+ \positionPosition^\transpose \! \!
		\left( \int_{t^*}^{t}
			\nabla^2 \potential\big( \bposition_s \big)
		\diff s \right)
		\positionPosition,
\end{align*}
where each term is clearly symmetric in the last expression.
Finally, we observe that
\begin{equation}
\label{eq:symplecticity_proof_intermediate_1}
\begin{aligned}
\positionPosition^\transpose \momentumMomentum
	&= \positionPosition^\transpose
	\left[
		\Id + \frac{2 \velocity_{\spEventIndex}}{\partial_{\spEventIndex} \potential(\bposition^*)}
					\left( \int_{t^*}^{t}
						\nabla^2 \potential(\bposition_s)
					\diff s \right)
					\bm{e}_{\spEventIndex} \bm{e}_{\spEventIndex}^\transpose
	\right] \\
	&= \Id
	+ 2 \velocity_{\spEventIndex} \frac{\partial t^*}{\partial \bposition}^\transpose \bm{e}_{\spEventIndex}^\transpose
	+ \frac{2 \velocity_{\spEventIndex}}{\partial_{\spEventIndex} \potential(\bposition^*)}
		\positionPosition^\transpose \left( \int_{t^*}^{t}
			\nabla^2 \potential(\bposition_s)
		\diff s \right)
		\bm{e}_{\spEventIndex} \bm{e}_{\spEventIndex}^\transpose
\end{aligned}
\end{equation}
and that
\begin{equation}
\label{eq:symplecticity_proof_intermediate_2}
\begin{aligned}
\momentumPosition^\transpose \positionMomentum
	&= \left[
		\int_{0}^{t^*} \nabla^2 \potential(\bposition + s \bvelocity) \diff s
		+ \left( \int_{t^*}^{t}
				\nabla^2 \potential\big( \bposition_s \big)
			\diff s \right)
			\positionPosition
	\right]^\transpose
	\frac{2 \velocity_{\spEventIndex}}{\partial_{\spEventIndex} \potential(\bposition^*)}
			\bm{e}_{\spEventIndex} \bm{e}_{\spEventIndex}^\transpose \\
	&= 2 \velocity_{\spEventIndex} \frac{\partial t^*}{\partial \bposition}^\transpose \bm{e}_{\spEventIndex}^\transpose
		+ \frac{2 \velocity_{\spEventIndex}}{\partial_{\spEventIndex} \potential(\bposition^*)}
			\positionPosition^\transpose \left( \int_{t^*}^{t}
				\nabla^2 \potential(\bposition_s)
			\diff s \right)
			\bm{e}_{\spEventIndex} \bm{e}_{\spEventIndex}^\transpose.
\end{aligned}
\end{equation}
Subtracting \eqref{eq:symplecticity_proof_intermediate_2} from \eqref{eq:symplecticity_proof_intermediate_1} yields the identity, completing the proof.
\end{proof}

\begin{lemma}
\label{lem:differentiability_of_event_time}
If $\partial_{\spEventIndex} \potential(\bposition + t^* \sign(\bmomentum)) \neq 0$, the velocity switch even time $t^*$ as in \eqref{eq:symplecticity_proof_event_time} is a continuously differentiable function of $\bposition$ and $\bmomentum$ with the derivatives
\begin{equation}
\label{eq:symplecticity_proof_event_time_derivative}
\frac{\partial t^*}{\partial \bposition}
	= - \frac{1}{\partial_{\spEventIndex} \potential(\bposition^*)}
		\bm{e}_{\spEventIndex}^\transpose
		\int_0^{t^*} \nabla^2 \potential( \bposition + s \bvelocity) \diff s, \quad
\frac{\partial t^*}{\partial \bmomentum}
	= \frac{1}{\partial_{\spEventIndex} \potential(\bposition^*)} \bm{e}_{\spEventIndex}^\transpose,
\end{equation}
where $\bposition^* = \bposition + t^* \bvelocity$, $\nabla^2 \potential$ is the Hessian, and the derivatives are expressed as row vectors.
\end{lemma}

\begin{proof}
Consider a function $f: \mathbb{R}^{2\nParam + 1} \to \mathbb{R}$ defined as
\begin{equation*}
f(t, \bposition, \bmomentum)
	= | \momentum_{\spEventIndex} | - \int_0^t \sign(\momentum_{\spEventIndex}) \spaceBeforePartial \partial_{\spEventIndex} \potential( \bposition + s \, \sign(\bmomentum)) \diff s.
\end{equation*}
Its derivative with respect to $t$ is given by
\begin{equation*}
\partial_t f(t, \bposition, \bmomentum)
	= - \, \sign(\momentum_{\spEventIndex}) \spaceBeforePartial \partial_{\spEventIndex} \potential( \bposition + t \, \sign(\bmomentum)).
\end{equation*}
Therefore, if $\partial_{\spEventIndex} \potential(\bposition + t^* \sign(\bmomentum)) \neq 0$, we have $\partial_t f(t^*, \bposition, \bmomentum) \neq 0$ and, by the implicit function theorem \protect\citepSupp{spivak1965calculus_on_manifold}, $t^*$ as defined implicitly via \eqref{eq:symplecticity_proof_event_time} is a continuously differentiable function of $\bposition$ and $\bmomentum$.
We now implicitly differentiate \eqref{eq:symplecticity_proof_event_time} in $\bposition$ and $\bmomentum$, which gives us
\begin{equation*}
\begin{aligned}
\bm{0}^\transpose
	&= \velocity_{\spEventIndex} \partial_{\spEventIndex} \potential(\bposition + t^* \bvelocity) \frac{\partial t^*}{\partial \bposition}
	+ \int_0^{t^*} \velocity_{\spEventIndex} \frac{\partial}{\partial \bposition} \partial_{\spEventIndex} \potential( \bposition + s \bvelocity) \diff s, \\
\sign(\momentum_{\spEventIndex}) \bm{e}_{\spEventIndex}^\transpose
	&= \velocity_{\spEventIndex} \partial_{\spEventIndex} \potential(\bposition + t^* \bvelocity) \frac{\partial t^*}{\partial \bmomentum}.
\end{aligned}
\end{equation*}
Rearranging terms, we obtain
\begin{equation*}
\frac{\partial t^*}{\partial \bposition}
	= - \frac{1}{\partial_{\spEventIndex} \potential(\bposition + t^* \bvelocity)}
	\int_0^{t^*} \frac{\partial}{\partial \bposition} \partial_{\spEventIndex} \potential( \bposition + s \bvelocity) \diff s, \quad
\frac{\partial t^*}{\partial \bmomentum}
	= \frac{1}{\partial_{\spEventIndex} \potential(\bposition + t^* \bvelocity)} \bm{e}_{\spEventIndex}^\transpose.
	\qedhere
\end{equation*}
\end{proof}

\section{Hamiltonian zigzag on constrained domain}
\label{sec:hamiltonian_zigzag_on_constrained_domain} 
The theory of Hamiltonian dynamics based on Laplace momentum in Section~\ref{sec:hamiltonian_zigzag} can be extended to define Hamiltonian zigzag on a constrained domain.
Our presentation of this extension follows the discontinuous Hamiltonian dynamics framework of \protect\citetSupp{nishimura2020discontinuous_hmc}, which in particular extends the technique for handling parameter constraints under Gaussian momentum \protect\citepSupp{neal2010hmc} to that under more general momentum distributions.

Essentially, Hamiltonian zigzag on a constrained domain can be derived as a limit of unconstrained dynamics by approximating the boundary with a steep change in the potential energy.
To describe the idea more precisely, we denote by $\Upsilon$ the constrained domain with piecewise smooth boundary $\partial \Upsilon$ and by $\potential(\bposition)$ the potential energy defined on $\Upsilon$.
Also, for each $\bposition \in \partial \Upsilon$, we denote by $\bm{\nu}(\bposition)$ the unit vector orthogonal to the boundary and pointing outside the constraint.
Using the notation of \protect\citetSupp{nishimura2020discontinuous_hmc}, we now consider a sequence of smooth potential energies $\{ \potential_\delta \}_{\delta > 0}$ defined on the unconstrained space such that  $\potential_\delta \equiv \potential$ on $\Upsilon$ and $\potential_\delta(\bposition) \to \infty$ as $\delta \to 0$ for all $\bposition \in \mathbb{R}^\nParam \setminus \Upsilon$.
The sequence $\{ \potential_\delta \}_{\delta > 0}$ defines the corresponding sequence of smooth Hamiltonian dynamics whose limit in turn defines the constrained dynamics.

The approximating sequence has a property $\| \nabla \potential_\delta(\bposition) \| \to \infty$ and $\nabla \potential_\delta(\bposition) / \| \nabla \potential_\delta(\bposition) \| \linebreak \to \bm{\nu}(\bposition)$ as $\delta \to 0$ for each $\bposition \in \partial \Upsilon$.
The effect of the constraint can thus be thought of as the potential energy $\potential(\bposition)$ having a gradient of infinite size in the direction of $\bm{\nu}(\bposition)$ on each boundary point $\bposition \in \partial \Upsilon$.
This implies that, when encountering the boundary at $\bposition \in \partial \Upsilon$, the constrained Hamiltonian zigzag undergoes an instantaneous change in momentum 
$$\bmomentum^+ = \bmomentum^- - \gamma \bm{\nu} \ \text{ for } \, \gamma \geq 0$$ 
where $\bm{\nu} = \bm{\nu}(\bposition)$ and $\bmomentum^+$ and $\bmomentum^-$ denote the momentum at infinitesimal moments before and after the event.
And the scalar $\gamma$ is determined through the energy conservation condition 
$$\sum_i | \momentum_i^+ | 
	= \sum_i | \momentum_i^- - \gamma \nu_i | 
	= \sum_i | \momentum_i^- |,$$ 
which admits a unique solution by convexity of the map $\bmomentum \to \sum_i | \momentum_i |$ .

The scalar $\gamma$ is implicitly defined in the above construction, but can be explicitly characterized in terms of the velocity $\bvelocity$ at the moment of hitting the boundary and of the orthogonal vector $\bm{\nu}$.
For constraints of the type $\left\{ \bposition : \sign(\bposition) = \bobservation \in \{\pm 1\}^\nParam \right\}$ as seen in our examples of Section~\ref{sec:simulation} and as common in other applications, the constrained dynamics's behavior at boundary $\{ \position_i = 0 \}$ amounts to a sign flip $\momentum_i^+ = - \momentum_i^-$ in the $i$-th momentum component, causing a trajectory to ``bounce'' against the boundary.
In other words, given the boundary vector $\bm{\nu} = \velocity_i \bm{e}_i$ for $\velocity_i  = \sign(\momentum_i)$ and the $i$-th basis vector $\bm{e}_i$, the conservation of energy dictates $\gamma = - 2 \velocity_i \momentum_i^-$. 

Under a more general constraint, we can solve the energy conservation condition $\sum_i | \momentum_i^- - \gamma \nu_i | = \sum_i | \momentum_i^- |$ for $\gamma$ as follows.
Denoting $\bvelocity^- = \sign(\bmomentum^-)$ and $\bvelocity^+ = \sign(\bmomentum^+)$, we first observe that we necessarily have $\langle \bvelocity^-, \bm{\nu} \rangle > 0$ and $\langle \bvelocity^+, \bm{\nu} \rangle < 0$ for the trajectory to have been moving toward and away from the boundary before and after the bounce event. 
In fact, given the constrained dynamics's origin as the limit of smooth ones, we can imagine the potential energy pushing the trajectory in the direction $-\bm{\nu}$ until it causes enough change in the momentum for the velocity to point inward.
The amount of such ``push'' required to change the $i$-th velocity component is given as a solution of $\momentum_i^- - \gamma \nu_i = 0$:
\begin{equation*}
\gamma_i = 
	\left\{
	\begin{array}{ll}
	\momentum_i^- / \nu_i &\text{ if } \sign(\momentum_i^-) = \sign(\nu_i) \\
	\infty &\text{ otherwise,}
	\end{array}
	\right.
\end{equation*}
where we use a placeholder value $\gamma_i = \infty$ to indicate non-existence of a positive-valued solution. 
From these $\gamma_i$'s, we can then determine in which order the velocity components would change their signs.
We denote this order of the sign changes by $i^{(k)}$, defined as
\begin{equation*}
i^{(1)} = \argmin_{i} \,  \gamma_i
\ \text{ and } \
i^{(k + 1)} = \argmin_{  i \thinnerspace \notin \thinnerspace  \{  i^{(1)}, \thinnerspace \ldots, \thinnerspace i^{(k)} \} } \gamma_i
\ \text{ for } \, k \geq 1.
\end{equation*} 
Technically, the above formula does not uniquely define $i^{(k + 1)}$ once $\min_{  i \thinnerspace \notin \thinnerspace  \{  i^{(1)}, \thinnerspace \ldots, \thinnerspace i^{(k)} \} } \gamma_i = \infty$, but this indeterminacy in $i^{(k + 1)}$ has no bearing on the rest of our discussion.
We correspondingly denote $\gamma^{(k)} = \gamma_{i^{(k)}}$ and define a sequence of velocity vectors $\bvelocity^{(k)}$ as
\begin{equation*}
\velocity^{(k)}_i =
	\left\{
	\begin{array}{cl}
	- \velocity_i^- &\text{ if } i \in \left\{ i^{(1)}, \ldots, i^{(k)} \right\} \\
	\velocity_i^- &\text{ otherwise}.
	\end{array}
	\right.
\end{equation*}

Having thus defined $\gamma^{(k)}$ and $\bvelocity^{(k)}$, we are now ready to find the explicit formula for the solution of $\sum_i | \momentum_i^- - \gamma \nu_i | = \sum_i | \momentum_i^- |$.
To this end, we first observe that, for $\gamma \in \big( \gamma^{(k)}, \gamma^{(k + 1)} \big)$,
\begin{equation*}
\frac{\diff}{\diff \gamma} \sum_i | \momentum_i^- - \gamma \nu_i | 
	= \frac{\diff}{\diff \gamma} \sum_i \velocity^{(k)}_i (\momentum_i^- - \gamma \nu_i)
	= - \langle \bvelocity^{(k)}, \bm{\nu} \rangle.
\end{equation*}
Moreover, our construction guarantees $- \langle \bvelocity^{(k)}, \bm{\nu} \rangle < - \langle \bvelocity^{(k + 1)}, \bm{\nu} \rangle$.
Therefore, the solution of $\sum_i | \momentum_i^- - \gamma \nu_i | = \sum_i | \momentum_i^- |$ must exist in the interval $\big(\gamma^{(k^*)}, \gamma^{(k^* + 1)} \big)$ where 
\begin{equation}
\label{eq:index_of_last_velocity_component_to_change_sign}
k^* = \max\!\left\{ k : \sum_i | \momentum_i^- - \gamma^{(k)} \nu_i | < \sum_i | \momentum_i^- | \right\}.
\end{equation}
And, for $\gamma \in \big(\gamma^{(k^*)}, \gamma^{(k^* + 1)} \big)$, we have
\begin{equation*}
\begin{aligned}
\sum_i | \momentum_i^- - \gamma \nu_i |
	&= \sum_i \velocity^{(k^*)}_i (\momentum_i^- - \gamma \nu_i)
	&= \left\langle \bvelocity^{(k^*)}, \bmomentum^{-} \right\rangle - \gamma \left\langle \bvelocity^{(k^*)}, \bm{\nu} \right\rangle.
\end{aligned}
\end{equation*}
The solution of $\sum_i | p_i^- - \gamma \nu_i | = \sum_i | p_i^- |$, therefore, is given by
\begin{equation*}
\gamma
	= \frac{
		\sum_i | \momentum_i | -  \left\langle \bmomentum^{-}, \bvelocity^{(k^*)} \right\rangle
	}{
		- \left\langle \bvelocity^{(k^*)}, \bm{\nu} \right\rangle
	}
	= \frac{
			2 \left( \sum_{ \, i \, \in \, \{  i^{(1)}, \thinnerspace \ldots,\thinnerspace i^{(k^*)} \} } | \momentum_i | \right)
		}{
			- \left\langle \bvelocity^{(k^*)}, \bm{\nu} \right\rangle
		}.
\end{equation*}

The theorem below establishes that Hamiltonian zigzag on a constrained domain retains the key properties of the unconstrained dynamics.
We assume the domain under linear constraints for convenience, but believe the conclusion to hold under more general piecewise smooth constraints. 
\begin{theorem}
\label{thm:symplecticity_under_linear_constraints}
Under the assumption of Theorem~\ref{thm:existence_and_propeties_of_hamiltonian_zigzag}, consider the constrained version of Hamiltonian zigzag as constructed above on a domain
\begin{equation*}
\Upsilon = \left\{ 
	\bposition : 
		\langle \bposition, \bm{\nu}_\ell \rangle \leq c_\ell, \
		\| \bm{\nu}_\ell \| = 1
		\, \text{ for } \thinnerspace
		\ell = 1, \ldots, \nLinearConstraint
\right\}.
\end{equation*}
This constrained dynamics is well-defined in the sense of Theorem~\ref{thm:existence_and_propeties_of_hamiltonian_zigzag}, away from a set of Lebesgue measure zero.
Moreover, it is time-reversible and symplectic.
\end{theorem}

\begin{proof}
The description in Supplement Section~\ref{sec:hamiltonian_zigzag_on_constrained_domain} uniquely defines the dynamics' behavior at a boundary for any trajectory whose initial position $\bposition(0)$ is within the domain's interior $\Upsilon \setminus \partial \Upsilon$ and is away from the zero-measure set $\zeroMeasureSet$ as constructed in the proof of Theorem~\ref{thm:existence_and_propeties_of_hamiltonian_zigzag}.
The dynamics is thus well-defined.
Also, constructed as a pointwise limit of a sequence of time-reversible dynamics, constrained Hamiltonian zigzag is automatically time-reversible.

To establish symplecticity, recall first that a composition of symplectic maps is again symplectic.
Since we have already established symplecticity for Hamiltonian zigzag away from the constraints, we only need to establish these properties for the portion of the dynamics in which a trajectory goes through a boundary event.
Without loss of generality, assume that a trajectory starting from $(\bposition, \bmomentum)$ encounters a single boundary event in an interval $[0, t]$ and that the corresponding boundary is given by $\{ \bposition : \langle \bposition, \bm{\nu} \rangle = c \}$.
Let $\velocityFlipIndexSet = \{  i^{(1)}, \thinnerspace \ldots, \thinnerspace i^{(k^*)} \}$ indicate the velocity components that undergo sign changes during this boundary event, where the index $k^*$ is defined as in Equation~\eqref{eq:index_of_last_velocity_component_to_change_sign}.
Then the trajectory's state at time $t$ is given by
\begin{equation}
\label{eq:constrained_dynamics_symplecticity_proof_solution_map} 
\begin{aligned}
\bposition_t
	&= \bposition+ t^* \bvelocity + (t - t^*) \bvelocity^*, \\
\bmomentum_t
	&= \bmomentum
	- \int_{0}^{t^*} \nabla \potential(\bposition + s \bvelocity) \diff s
	- \gamma \bm{\nu} 
	- \int_{t^*}^{t} \nabla \potential\big( \bposition + t^* \bvelocity + (s - t^*) \bvelocity^* \big) \diff s,
\end{aligned}
\end{equation}
where $t^*$, $\bvelocity^*$, and $\gamma$ satisfy the following relations:
\begin{gather*}
t^* 
= \frac{c - \langle \bposition, \bm{\nu} \rangle}{\langle \bvelocity, \bm{\nu} \rangle};
\ \,
\velocity_{\spEventIndex}^* = - \velocity_{\spEventIndex}
	\, \text{ if } i \in \velocityFlipIndexSet
	\, \text{ and } \,
	\velocity_j^* = \velocity_j \text{ otherwise}; \\
\gamma 
	= \frac{
		2 \sum_{ \, i \, \in \, \velocityFlipIndexSet} | \momentum_i | 
	}{
		- \left\langle \bvelocity^*, \bm{\nu} \right\rangle
	}
	= \frac{
		\langle \bmomentum^*, \bvelocity - \bvelocity^* \rangle
	}{
		- \left\langle \bvelocity^*, \bm{\nu} \right\rangle
	} 
	\ \text{ with } \,
	\bmomentum^* 
		= \bmomentum - \int_{0}^{t^*} \nabla \potential(\bposition + s \bvelocity) \diff s.
\end{gather*}
In preparation for analyzing the derivatives of \eqref{eq:constrained_dynamics_symplecticity_proof_solution_map}, we note that
\begin{gather*}
\frac{\partial t^*}{\partial \bposition}
	= - \frac{1}{\langle \bvelocity, \bm{\nu} \rangle} \bm{\nu}^\transpose, \ \
\frac{\partial \gamma}{\partial \bposition}
	= - \frac{1}{\langle \bvelocity^*, \bm{\nu} \rangle}  (\bvelocity - \bvelocity^*)^\transpose 
		\frac{\partial \bmomentum^* }{\partial \bposition}, \ \
\frac{\partial \gamma}{\partial \bmomentum}
	= - \frac{1}{\langle \bvelocity^*, \bm{\nu} \rangle}  (\bvelocity - \bvelocity^*)^\transpose, \\
\frac{\partial \bmomentum^* }{\partial \bposition}
	= - \nabla \potential(\bposition^*) \frac{\partial t^*}{\partial \bposition} - \int_{0}^{t^*} \nabla^2 \potential(\bposition + s \bvelocity) \diff s 
	= \frac{1}{\langle \bvelocity, \bm{\nu} \rangle} \nabla \potential(\bposition^*) \bm{\nu}^\transpose - \int_{0}^{t^*} \nabla^2 \potential(\bposition + s \bvelocity) \diff s.
\end{gather*}

We now differentiate \eqref{eq:constrained_dynamics_symplecticity_proof_solution_map} using the chain rule and obtain, after following calculations similar to those in the proof of Theorem~\ref{thm:existence_and_propeties_of_hamiltonian_zigzag}, 
\begin{align}
\positionPosition
	&=  \Id - \frac{1}{\langle \bvelocity, \bm{\nu} \rangle} (\bvelocity - \bvelocity^*) \bm{\nu}^\transpose, 
	\label{eq:constrained_dynamics_symplecticity_proof_first_partial_deriv} \\
\positionMomentum
	&= \bm{0}, \\
\momentumPosition
	&= 
	\int_{0}^{t^*} \nabla^2 \potential(\bposition + s \bvelocity) \diff s
	- \frac{1}{\left\langle \bvelocity^*, \bm{\nu} \right\rangle} 
		\bm{\nu} (\bvelocity - \bvelocity^*)^\transpose \frac{\partial \bmomentum^*}{\partial \bposition}
	+ \left( \int_{t^*}^{t}
			\nabla^2 \potential(\bposition_s)
		\diff s \right)
		\positionPosition, \\
\momentumMomentum
	&= \Id + \frac{1}{\langle \bvelocity^*, \bm{\nu} \rangle} \bm{\nu} (\bvelocity - \bvelocity^*)^\transpose.
	\label{eq:constrained_dynamics_symplecticity_proof_last_partial_deriv} 
\end{align}
We check the three conditions \eqref{eq:symplecticity_definition_simplified} for symplecticity from Equation~\eqref{eq:constrained_dynamics_symplecticity_proof_first_partial_deriv} -- \eqref{eq:constrained_dynamics_symplecticity_proof_last_partial_deriv}.
We have $(\partial \bposition_t / \partial \bmomentum)^\transpose (\partial \bmomentum_t / \partial \bmomentum) = \bm{0}$, which is trivially symmetric.
It is also straightforward to check
\begin{equation*}
\positionPosition^\transpose \momentumMomentum - \momentumPosition^\transpose \positionMomentum 
	= \left[ \Id - \frac{1}{\langle \bvelocity, \bm{\nu} \rangle} (\bvelocity - \bvelocity^*) \bm{\nu}^\transpose \right]
		\left[ \Id + \frac{1}{\langle \bvelocity^*, \bm{\nu} \rangle} \bm{\nu} (\bvelocity - \bvelocity^*)^\transpose \right] \\
	= \Id.
\end{equation*}
It remains to show that $(\partial \bposition_t / \partial \bposition)^\transpose (\partial \bmomentum_t / \partial \bposition)$ is symmetric. 
We start by observing that
\begin{align*}
&\positionPosition^\intercal \momentumPosition 
	- 
	\positionPosition^\intercal 
	\left( \int_{t^*}^{t}
		\nabla^2 \potential(\bposition_s)
	\diff s \right)
	\positionPosition \\
&\hspace*{1.5em}=
	\left( \Id - \frac{1}{\langle \bvelocity, \bm{\nu} \rangle} \bm{\nu} (\bvelocity - \bvelocity^*)^\transpose \right)
	\left(
		\int_{0}^{t^*} \nabla^2 \potential(\bposition + s \bvelocity) \diff s
		- \frac{1}{\left\langle \bvelocity^*, \bm{\nu} \right\rangle} 
			\bm{\nu} (\bvelocity - \bvelocity^*)^\transpose \frac{\partial \bmomentum^*}{\partial \bposition}
	\right),
	\yesnumber
	\label{eq:constrained_dynamics_symmetry_proof_intermediate_expression_2}
\end{align*}
where we have subtracted a symmetric matrix from $(\partial \bposition_t / \partial \bposition)^\transpose (\partial \bmomentum_t / \partial \bposition)$ on the left hand side.
Expanding \eqref{eq:constrained_dynamics_symmetry_proof_intermediate_expression_2}, we have
\begin{equation}
\label{eq:constrained_dynamics_symmetry_proof_intermediate_expression_1}
\begin{aligned}
&\positionPosition^\intercal \momentumPosition 
	- 
	\positionPosition^\intercal 
	\left( \int_{t^*}^{t}
		\nabla^2 \potential(\bposition_s)
	\diff s \right)
	\positionPosition \\
&\hspace*{1.5em}=
	\int_{0}^{t^*} \nabla^2 \potential(\bposition + s \bvelocity) \diff s 
	- \frac{1}{\langle \bvelocity, \bm{\nu} \rangle} \bm{\nu} (\bvelocity - \bvelocity^*)^\transpose \int_{0}^{t^*} \nabla^2 \potential(\bposition + s \bvelocity) \diff s \\
&\hspace*{3.5em}
	- \frac{1}{\left\langle \bvelocity^*, \bm{\nu} \right\rangle} 
		\bm{\nu} (\bvelocity - \bvelocity^*)^\transpose \frac{\partial \bmomentum^*}{\partial \bposition}
	+ \frac{1}{\langle \bvelocity, \bm{\nu} \rangle \left\langle \bvelocity^*, \bm{\nu} \right\rangle} \bm{\nu} (\bvelocity - \bvelocity^*)^\transpose
		\bm{\nu} (\bvelocity - \bvelocity^*)^\transpose \frac{\partial \bmomentum^*}{\partial \bposition}.
\end{aligned}
\end{equation}
We note that the last term in the above equation can be simplified as
\begin{align*}
\frac{1}{\langle \bvelocity, \bm{\nu} \rangle \left\langle \bvelocity^*, \bm{\nu} \right\rangle} 
	\bm{\nu} (\bvelocity - \bvelocity^*)^\transpose
	\bm{\nu} (\bvelocity - \bvelocity^*)^\transpose \frac{\partial \bmomentum^*}{\partial \bposition}
	&= \frac{%
			\langle \bvelocity - \bvelocity^*, \bm{\nu} \rangle
		}{%
			\langle \bvelocity, \bm{\nu} \rangle \left\langle \bvelocity^*, \bm{\nu} \right\rangle
		} 
		\bm{\nu} (\bvelocity - \bvelocity^*)^\transpose 
		\frac{\partial \bmomentum^*}{\partial \bposition} \\
	&= \left(
			\frac{1}{\langle \bvelocity^*, \bm{\nu} \rangle} - \frac{1}{\langle \bvelocity, \bm{\nu} \rangle}
		\right)
		\bm{\nu} (\bvelocity - \bvelocity^*)^\transpose 
		\frac{\partial \bmomentum^*}{\partial \bposition}.
\end{align*}
Plugging in the simplified expression back into Equation~\eqref{eq:constrained_dynamics_symmetry_proof_intermediate_expression_1}, we obtain
\begin{align*}
&\positionPosition^\intercal \momentumPosition 
	- 
	\positionPosition^\intercal 
	\left( \int_{t^*}^{t}
		\nabla^2 \potential(\bposition_s)
	\diff s \right)
	\positionPosition 
	- \int_{0}^{t^*} \nabla^2 \potential(\bposition + s \bvelocity) \diff s \\
&\hspace*{1.5em}=
	- \frac{1}{\langle \bvelocity, \bm{\nu} \rangle} \bm{\nu} (\bvelocity - \bvelocity^*)^\transpose \int_{0}^{t^*} \nabla^2 \potential(\bposition + s \bvelocity) \diff s
	- \frac{1}{\left\langle \bvelocity, \bm{\nu} \right\rangle} 
		\bm{\nu} (\bvelocity - \bvelocity^*)^\transpose \frac{\partial \bmomentum^*}{\partial \bposition} \\
&\hspace*{1.5em}=
	- \frac{1}{\langle \bvelocity, \bm{\nu} \rangle} \bm{\nu} (\bvelocity - \bvelocity^*)^\transpose \int_{0}^{t^*} \nabla^2 \potential(\bposition + s \bvelocity) \diff s \\
&\hspace*{4em}
	- \frac{1}{\left\langle \bvelocity, \bm{\nu} \right\rangle} 
		\bm{\nu} (\bvelocity - \bvelocity^*)^\transpose 
		\left(
			\frac{1}{\langle \bvelocity^*, \bm{\nu} \rangle} \nabla \potential(\bposition^*) \bm{\nu}^\transpose
			- \int_{0}^{t^*} \nabla^2 \potential(\bposition + s \bvelocity) \diff s
		\right)\\
&\hspace*{1.5em}=
	\frac{%
			\left\langle \bvelocity - \bvelocity^*, \nabla \potential(\bposition^*) \right\rangle
		}{%
			\left\langle \bvelocity, \bm{\nu} \right\rangle \left\langle \bvelocity^*, \bm{\nu} \right\rangle
		}%
		\bm{\nu} \bm{\nu}^\transpose.
\end{align*}
The last equality makes it clear that $(\partial \bposition_t / \partial \bposition)^\transpose (\partial \bmomentum_t / \partial \bposition)$ can be expressed as the sum of the three symmetric matrices and is thus symmetric itself.
This completes the proof of symplecticity and of the theorem.
\end{proof}

\section{Proof of Theorem~\ref{thm:strong_convergence_of_hamiltonian_zigzag_to_markovian}}
\label{sec:proof_of_strong_convergence_of_hamiltonian_zigzag_to_markovian}

In the proof of Theorem~\ref{thm:strong_convergence_of_hamiltonian_zigzag_to_markovian} below, we couple the two zigzags by using the same random variable $\uniformRv_{n, i} \sim \unifDist(0, 1)$ \nolinebreak for the momentum refreshment $| \momentum_i^\superH(n\dt) | = - \log(\uniformRv_{n, i})$ and simulation of the velocity switch event time \eqref{eq:markovian_zigzag_event_time} for $(\bposition_{\dt}^\superM, \bvelocity_{\dt}^\superM)$ on $[n \dt, (n + 1) \dt]$.
The critical property of our coupling is visually illustrated in Figure~\ref{fig:coupled_zigzag_trajectories}.
Our construction guarantees that, given $(\bposition_{\dt}^\superH, \bvelocity_{\dt}^\superH)(n \dt) = (\bposition_{\dt}^\superM, \bvelocity_{\dt}^\superM)(n \dt)$, the probability of the two zigzag trajectories diverging on the next time interval $[n \dt, (n + 1) \dt]$ is at most $O(\dt^2)$.
Therefore, the total probability of them diverging during any of the the $\lceil \integrationTime / \dt \rceil$ intervals is $O(\dt)$.
In other words, with probability $1 - O(\dt)$, the coupled zigzags $(\bposition_{\dt}^\superH, \bmomentum_{\dt}^\superH)$ and $(\bposition_{\dt}^\superM, \bvelocity_{\dt}^\superM)$ stay together during the entire $[0, T]$ interval and hence have the Skorokhod distance of zero.

\begin{figure}
\begin{minipage}{.35\linewidth}
\includegraphics[page=1, height=.2\textheight]{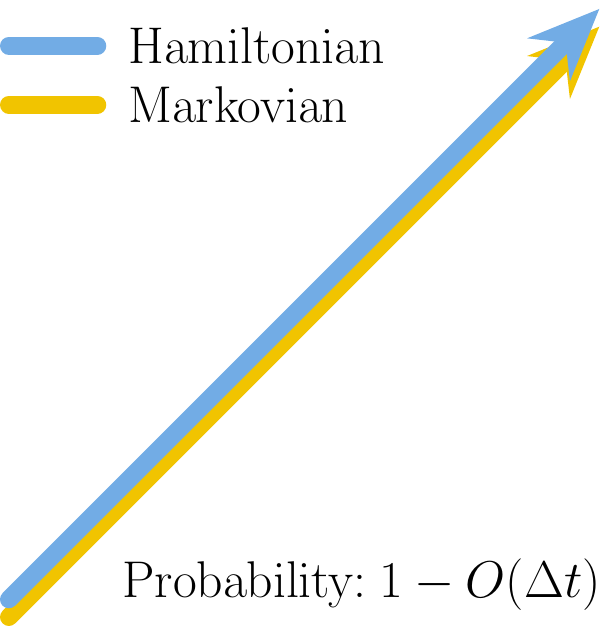}
\vspace*{\baselineskip}

\includegraphics[page=4, height=.2\textheight]{Figure/coupled_zigzag_paths}
\end{minipage}
~
\begin{minipage}{.2\linewidth}
\includegraphics[page=2, height=.2\textheight]{Figure/coupled_zigzag_paths}
\vspace*{\baselineskip}

\includegraphics[page=3, height=.2\textheight]{Figure/coupled_zigzag_paths}
\end{minipage}
\begin{minipage}{.38\linewidth}
\caption{%
	Potential trajectories of the coupled Hamiltonian (blue) and Markovian (yellow) zigzag on an interval $[n \dt, (n + 1) \dt]$ conditional on them being at the same state at time $n \dt$.
	With probability \mbox{$1 - O(\dt)$} or $O(\dt)$, the two zigzags trace the same exact path with zero or one velocity switch event (top left and top right).
	The paths diverge with probability at most $O(\dt^2)$ (bottom left and bottom right).
}
\label{fig:coupled_zigzag_trajectories}
\end{minipage}
\end{figure}

\begin{proof}
Let $\{ \buniformRv_n = (\uniformRv_{n, 1}, \ldots, \uniformRv_{n, \nParam}) \}$ and $\{ \baltUniformRv_t = (\altUniformRv_{t, 1}, \ldots, \altUniformRv_{t, \nParam}) \}$ for $n \in \mathbb{Z}^+ \cup \{0\}$ and $t \in (0, \infty)$ be a collection of independent $\unifDist(0, 1)$ random variables.
We will construct Hamiltonian zigzag out of $\buniformRv_n$ and additionally require $\baltUniformRv_t$ for Markovian zigzag.
We begin with the construction of Hamiltonian zigzag with momentum magnitude refreshment.
Given the initial state
$\bposition_{\dt}^\superH(0) = \bposition(0)$ and $\bvelocity_{\dt}^\superH(0) = \bvelocity(0)$, set $|\bmomentum_{\dt}^\superH(0)| = - \log(\buniformRv_0)$ and $\bmomentum_{\dt}^\superH(0) = |\bmomentum_{\dt}^\superH(0)| \bvelocity_{\dt}^\superH(0)$.
For $t \in [0, \dt]$, let the process evolve according to the differential equation \eqref{eq:hamilton_under_laplace_momentum}.
At time $t = \dt$, resample the momentum magnitude by setting $|\bmomentum^\superH(\dt)| \gets - \log(\buniformRv_1)$ and $\bmomentum^\superH(\dt) \gets |\bmomentum^\superH(\dt)| \bvelocity_{\dt}^\superH(\dt)$, after which we let the process evolve again for $[\dt, 2 \dt]$ according to \eqref{eq:hamilton_under_laplace_momentum}.
We define the rest of dynamics inductively on the intervals $[n \dt, (n + 1) \dt]$ until $(n + 1) \dt > T$, assigning
\begin{equation*}
|\bmomentum_{\dt}^\superH(n \dt)| = - \log(\buniformRv_{n}).
\end{equation*}

We now turn to the construction of Markovian zigzag, first focusing on the interval $[0, \dt]$.
Given the initial state $\bposition_{\dt}^\superM(0) = \bposition(0)$ and $\bvelocity_{\dt}^\superM(0) = \bvelocity(0)$, the process travels along a linear path with constant velocity
\begin{equation*}
\bposition_{\dt}^\superM(t) = \bposition(0) + t \bvelocity(0), \
\bvelocity_{\dt}^\superM(t) = \bvelocity(0)
\end{equation*}
until time $\dt$ or the first event time $\tau_0^{\lowerscript{(1)}} = \min_i \interEventTime_{0, i}^{\lowerscript{(1)}}$ where
\begin{equation*}
\interEventTime_{0, i}^{\lowerscript{(1)}} = \inf_{t > 0} \left\{
	|\momentum_{\dt, i}^\superH(0)|
		= - \log \uniformRv_{0, i}
		\leq \int_0^t \left[
		\velocity_i(0) \spaceBeforePartial \partial_i \potential( \bposition(0) + s \bvelocity(0) )
	\right]^+ \diff s
\right\}.
\end{equation*}
If $\tau_0^{\lowerscript{(1)}} \leq \dt$, the process undergoes an instantaneous velocity change at $t = \tau_0^{\lowerscript{(1)}}$:
\begin{equation}
\label{eq:velocity_change}
\velocity_{\dt, i^*}^\superM(\tau_0^{\lowerscript{(1)}})
	= - \velocity_{i^*}(0)
	\ \text{ for } \
	i^* = \argmin_i \interEventTime_{0, i}^{\lowerscript{(1)}}.
\end{equation}
The process then continues for $t \geq \tau_0^{\lowerscript{(1)}}$ along a new linear path
\begin{equation}
\label{eq:markovian_zigzag_after_first_event}
\bposition_{\dt}^\superM(t)
	= \bposition_{\dt}^\superM \big( \tau_{0}^{\lowerscript{(1)}} \big)
	+ \big( t - \tau_{0}^{\lowerscript{(1)}} \big) \bvelocity_{\dt}^\superM \big( \tau_{0}^{\lowerscript{(1)}} \big), \
\bvelocity_{\dt}^\superM(t) = \bvelocity_{\dt}^\superM\big( \tau_{0}^{\lowerscript{(1)}} \big)
\end{equation}
until time $\dt$ or the next event time $\tau_0^{\lowerscript{(2)}} = \tau_0^{\lowerscript{(1)}} + \min_i \interEventTime_{0, i}^{\lowerscript{(2)}}$ where
\begin{equation}
\label{eq:markovian_event_time_after_first}
\interEventTime_{0, i}^{\lowerscript{(2)}} = \inf_{t > 0} \left\{
	- \log \altUniformRv_{\raisebox{2pt}{\scriptsize $\tau_{0}^{(1)}$}\!, i} \leq \int_0^t \left[
		\velocity_{\dt, i}^\superM \big( \tau_{0}^{\lowerscript{(1)}} \big)  \spaceBeforePartial \partial_i \potential \big(
			\bposition_{\dt}^\superM \big( \tau_{0}^{\lowerscript{(1)}} \big) + s \bvelocity_{\dt}^\superM \big( \tau_{0}^{\lowerscript{(1)}} \big)
		\big)
	\right]^+ \diff s
\right\},
\end{equation}
undergoing instantaneous velocity change 
at $t = \tau_0^{\lowerscript{(2)}}$ if $\tau_0^{\lowerscript{(2)}} \leq \dt$:
\begin{equation*}
\velocity_{\dt, i^*}^\superM \big( \tau_{0}^{\lowerscript{(2)}} \big)
	= - \velocity_{i^*} \big( \tau_{0}^{\lowerscript{(1)}} \big)
	\ \text{ for } \
	i^* = \argmin_i \interEventTime_{0, i}^{\lowerscript{(2)}}.
\end{equation*}
Continuing this construction until $\tau_0^{(k)} > \dt$ defines the Markovian zigzag process on $[0, \dt]$ with the given initial condition \protect\citepSupp{davis1993markov_models, bierkens2019zigzag_original}.

Our construction of Markovian zigzag on $[0, \dt]$ so far follows the standard algorithm for simulating piecewise deterministic Markov processes, except for the fact that we couple the Markovian zigzag to its Hamiltonian counter-part via the random variable $\buniformRv_0 = \exp(- |\bmomentum_\dt^\superH(0)|)$.
In case $\tau_0^{\lowerscript{(1)}} \leq \dt$, however, the evolution of $(\bposition_{\dt}^\superM(t), \bvelocity_{\dt}^\superM(t))$ after $t = \tau_0^{\lowerscript{(1)}}$ would have little relation to that of $(\bposition_{\dt}^\superH(t), \bvelocity_{\dt}^\superH(t))$.
To ensure that the two processes remain closely interrelated, we ``re-couple'' them at every $\dt$ interval through the shared random variable $\buniformRv_{n}$.
More precisely, for each $n \in \mathbb{Z}^+$, we define the first potential event time on $[n \dt, (n + 1) \dt]$ as $\tau_n^{\lowerscript{(1)}} = n \dt + \min_i \interEventTime_{n, i}^{\lowerscript{(1)}}$ where
\begin{equation*}
\begin{aligned}
\interEventTime_{n, i}^{\lowerscript{(1)}} = \inf_{t > 0}
	\Big\{
		|\momentum_{\dt, i}^\superH(n \dt)|
		&= - \log \uniformRv_{n, i} \\
		&\leq \int_0^t \left[
			\velocity_{\dt, i}^\superM(n \dt)
			\spaceBeforePartial \partial_i \potential \big(
				\bposition_{\dt}^\superM(n \dt) + s \bvelocity_{\dt}^\superM(n \dt)
			\big)
		\right]^+ \diff s
	\Big\}.
\end{aligned}
\end{equation*}
If $\interEventTime_{n, i}^{\lowerscript{(1)}} < \dt$, the process continues for the rest of $t \in [n \dt, (n + 1) \dt]$ in the same manner as it did on $[0, \dt]$, according to the appropriate analogues of \eqref{eq:velocity_change}, \eqref{eq:markovian_zigzag_after_first_event}, and \eqref{eq:markovian_event_time_after_first}.
The process is thus defined inductively until $(n + 1) \dt > \integrationTime$.
The described procedure defines the same transition kernel on the intervals $[n \dt, (n + 1) \dt]$ as on $[0, \dt]$.
By the Markov property and Chapman-Kolmogorov equation \protect\citepSupp{davis1993markov_models}, therefore, the process $(\bposition_{\dt}^\superM, \bvelocity_{\dt}^\superM)$ defines Markovian zigzag on $[0, \integrationTime]$.

We now show that the coupled Hamiltonian and Markovian zigzag process follow the exact same path with high probability.
Observe that
{\renewcommand{\arraystretch}{1.25}
\begin{align}
&\probability \! \left\{
		(\bposition_{\dt}^\superH, \bvelocity_{\dt}^\superH) \equiv (\bposition_{\dt}^\superM, \bvelocity_{\dt}^\superM)
		\ \text{ on } \
		\big[ 0, \lceil \integrationTime / \dt \rceil \dt \big]
	\right\} \nonumber \\
&\hspace{.75em}
	= \prod_{n = 0}^{\lceil \integrationTime / \dt \rceil} \probability \!\left\{
			\begin{array}{r}
			(\bposition_{\dt}^\superH, \bvelocity_{\dt}^\superH) \equiv (\bposition_{\dt}^\superM, \bvelocity_{\dt}^\superM) \hspace{1em} \\
			\text{ on } \ \left[ n \dt, (n + 1) \dt \right]
			\end{array}
			\ \bigg| \
			(\bposition_{\dt}^\superH, \bvelocity_{\dt}^\superH) \equiv (\bposition_{\dt}^\superM, \bvelocity_{\dt}^\superM)
			\ \text{ on } \ \left[ 0, n \dt \right]
		\right\} \nonumber \\
&\hspace{.75em}
	= \prod_{n = 0}^{\lceil \integrationTime / \dt \rceil} \probability \!\left\{
			\begin{array}{r}
			(\bposition_{\dt}^\superH, \bvelocity_{\dt}^\superH) \equiv (\bposition_{\dt}^\superM, \bvelocity_{\dt}^\superM) \hspace{1em} \\
			\text{ on } \ \left[ n \dt, (n + 1) \dt \right]
			\end{array}
			\ \bigg| \
			(\bposition_{\dt}^\superH, \bvelocity_{\dt}^\superH)(n \dt) = (\bposition_{\dt}^\superM, \bvelocity_{\dt}^\superM)(n \dt)
		\right\}, \label{eq:prob_of_exact_coupling_as_product}
\end{align}
}%
where the second equality follows from the conditional independence of the processes on $[0, n \dt]$ and $[n \dt, \infty)$ by construction.
Applying Lemma~\ref{LEM:COUPLING_ON_EACH_INTERVAL} to each term in \eqref{eq:prob_of_exact_coupling_as_product}, we have
\begin{equation}
\label{eq:exact_coupling_prob_lower_bd}
\begin{aligned}
\probability \! \left\{
		(\bposition_{\dt}^\superH, \bvelocity_{\dt}^\superH) \equiv (\bposition_{\dt}^\superM, \bvelocity_{\dt}^\superM)
		\ \text{ on } \
		\big[ 0, \lceil \integrationTime / \dt \rceil \dt \big]
	\right\}
&\geq \left(
		1 - C_\integrationTime(\bposition(0)) \dt^2
	\right)^{\lceil \integrationTime / \dt \rceil} \\
&\geq
	1 - C_\integrationTime(\bposition(0)) (\integrationTime + \dt) \dt, 
\end{aligned}
\end{equation}
where we can take
$
C_\integrationTime(\bposition)
	= \textstyle \frac{1}{2} \gradBound_\integrationTime(\bposition)^2 \nParam^2
		+ 2 \hessianBound_\integrationTime(\bposition) \nParam^2
$
in terms of $\gradBound_{\integrationTime}(\bposition)$ and $\hessianBound_\integrationTime(\bposition)$ as in \eqref{eq:grad_and_hessian_bound_definition}
since $\bposition_{\dt}^\superH(t)$ and $\bposition_{\dt}^\superM(t)$ stay within the $L^\infty$-ball $\ball_\integrationTime$ for all $t \in [0, \integrationTime]$.
We conclude from \eqref{eq:exact_coupling_prob_lower_bd} that, as $\dt \to 0$,
\begin{equation*}
\probability\!\left\{
		\rho_\integrationTime\!\left[
			(\bposition_{\dt}^\superH, \bvelocity_{\dt}^\superH),
			(\bposition_{\dt}^\superM, \bvelocity_{\dt}^\superM)
		\right] = 0
	\right\}
	=
\probability\!\left\{
		(\bposition_{\dt}^\superH, \bvelocity_{\dt}^\superH) \equiv (\bposition_{\dt}^\superM, \bvelocity_{\dt}^\superM)
		\ \text{ on } \
		[ 0, \integrationTime ]
	\right\} \to 1. \qedhere
\end{equation*}
\end{proof}

Lemma~\ref{LEM:COUPLING_ON_EACH_INTERVAL} below is used in our proof of Theorem~\ref{thm:convergence_of_hamiltonian_zigzag_to_markovian}:
\begin{lemma}
\label{LEM:COUPLING_ON_EACH_INTERVAL}
For fixed $\bposition_n$ and $\bvelocity_n$, we have
{\renewcommand{\arraystretch}{1.25}
\begin{equation*}
\begin{aligned}
\probability \!\left\{
		\begin{array}{r}
		(\bposition_{\dt}^\superH, \bvelocity_{\dt}^\superH) \equiv (\bposition_{\dt}^\superM, \bvelocity_{\dt}^\superM) \hspace{2em} \\
		\text{ on } \ \left[ n \dt, (n + 1) \dt \right]
		\end{array}
		\ \Bigg| \
		(\bposition_{\dt}^\superH, \bvelocity_{\dt}^\superH)(n \dt)
			= (\bposition_{\dt}^\superM, \bvelocity_{\dt}^\superM)(n \dt)
			= (\bposition_n, \bvelocity_n)
	\right\} \\
	\geq 1 - \left(
			\textstyle \frac{1}{2} \gradBound_\dt(\bposition_n)^2 \nParam^2
				+ 2 \hessianBound_\dt(\bposition_n)\thinnerspace \nParam^2
	\right) \dt^2,
\end{aligned}
\end{equation*}
}%
where
\begin{equation}
\label{eq:grad_and_hessian_bound_definition}
\begin{alignedat}{3}
\gradBound_\dt(\bposition)
	&= \max\big\{ \hspace{1.5ex} | \partial_i \potential(\bposition') | &&: 1 \leq i \leq \nParam, \, \bposition' \in \ball_\dt(\bposition) \big\} \\
\hessianBound_\dt(\bposition)
	&= \max\big\{ \| \nabla \partial_i \potential(\bposition') \|_2 &&: 1 \leq i \leq \nParam, \, \bposition' \in \ball_\dt(\bposition) \big\}
\end{alignedat}
\end{equation}
with $\| \cdot \|_2$ denoting the $\ell^2$-norm and $\ball_r(\bposition) = \{ \bposition' : \| \bposition' - \bposition \|_\infty < r \}$ an $L^\infty$-ball. 
\end{lemma}

The intuition behind our proof below of Lemma~\ref{LEM:COUPLING_ON_EACH_INTERVAL} is as follows.
On the interval $[n \dt, (n + 1) \dt]$, two or more velocity switch events occur only with the negligible probability of $O(\dt^2)$, allowing us to focus on the case of zero or one event.
When starting from the same position and velocity, our coupling of $(\bposition_{\dt}^\superH, \bvelocity_{\dt}^\superH)$ and $(\bposition_{\dt}^\superM, \bvelocity_{\dt}^\superM)$ implies that, if Markovian zigzag does not experience any event on $[n \dt, (n + 1) \dt]$, neither does Hamiltonian zigzag.
Otherwise, when Markovian zigzag undergoes its first event, with high probability Hamiltonian zigzag experiences the same velocity switch at the exact same moment.

In bounding the probabilities of various events in the proof, we use the fact that zigzag velocities only take values in $\{ \pm 1\}^\nParam$ and hence, regardless of their actual trajectories, $\bposition_{\dt}^\superH(t)$ and $\bposition_{\dt}^\superM(t)$ stay within the $L^\infty$-ball $\ball_\dt$ around the initial position.

\begin{proof}[Proof of Lemma~\ref{LEM:COUPLING_ON_EACH_INTERVAL}]
We assume without loss of generality $n = 0$ as the proof is essentially identical for all $n$.
Throughout the proof, all the discussed events are conditioned on $(\bposition_{\dt}^\superH, \bvelocity_{\dt}^\superH)(0) = (\bposition_{\dt}^\superM, \bvelocity_{\dt}^\superM)(0) = (\bposition_0, \bvelocity_0)$, but we do not explicitly denote the conditioning to avoid notational clutter.
For $i = 1, \ldots, \nParam$, define indicator variables
\begin{equation*}
\eventIndicator_i
	= \mathds{1}\left\{
		|\momentum_{\dt, i}^\superH(0)|
			= - \log \uniformRv_{0, i}
			\leq \int_0^\dt \big[
			\velocity_{0, i} \spaceBeforePartial \partial_i \potential( \bposition_0 + s \bvelocity_0 )
		\big]^+ \diff s
	\right\},
\end{equation*}
where $I_i = 1$ indicates the possibility of the velocity sign change along the $i$-th coordinate during the interval $[0, \dt]$.
We partition the probability space into three types of events based on how many $\eventIndicator_i$'s equal $1$:
\begin{equation}
\label{eq:partition_via_number_of_events}
\big\{ \eventIndicator_i = 0 \, \text{ for all $i$} \big\}
	\hspace{1ex} \text{ or } \hspace{1ex}
	{\textstyle \bigcup\limits_{i = 1}^\nParam} \big\{ \eventIndicator_i = 1, \,  \eventIndicator_j = 0 \text{ for } j \neq i \big\}
	\hspace{1ex} \text{ or } \hspace{1ex}
	\mathsmaller{\bigcup}\limits_{1 \leq i < j \leq \nParam} \big\{ \eventIndicator_i = \eventIndicator_j = 1 \big\}.
\end{equation}
Below we analyze implications of each event type on the trajectory of $(\bposition_{\dt}^\superH, \bvelocity_{\dt}^\superH)$ and $(\bposition_{\dt}^\superM, \bvelocity_{\dt}^\superM)$ on $[0, \dt]$.
We first establish that
\begin{equation}
\label{eq:no_event_implication}
\big\{ \eventIndicator_i = 0 \, \text{ for all $i$} \big\}
	\subset \left\{
			(\bposition_{\dt}^\superH, \bvelocity_{\dt}^\superH) \equiv (\bposition_{\dt}^\superM, \bvelocity_{\dt}^\superM)
			\ \text{ on } \ \left[ 0, \dt \right]
		\right\},
\end{equation}
and that
\begin{equation}
\label{eq:more_than_two_events_implication}
\probability \big\{ \eventIndicator_i = \eventIndicator_j = 1 \big\}
	\leq \gradBound_\dt(\bposition_0)^2 \dt^2
	\ \text{ for } \,
	i < j.
\end{equation}
We then show that, if $| \partial_i \potential(\bposition_0) | > \hessianBound_\dt(\bposition_0) \sqrt{\nParam} \dt$,
\begin{multline}
\label{eq:one_event_with_non_neligible_gradient_implication}
\probability\big(
\big\{ \eventIndicator_i = 1, \,  \eventIndicator_j = 0 \text{ for } j \neq i \big\}
	\, \big\backslash \left\{
			(\bposition_{\dt}^\superH, \bvelocity_{\dt}^\superH) \equiv (\bposition_{\dt}^\superM, \bvelocity_{\dt}^\superM)
			\text{ on } \left[ 0, \dt \right]
		\right\}
\big) \\
	\leq \hessianBound_\dt(\bposition_0)(\nParam - 1) \dt^2 \hspace*{1.5em}
\end{multline}
and otherwise
\begin{equation}
\label{eq:one_event_with_neligible_gradient_implication}
\probability \big\{ \eventIndicator_i = 1, \,  \eventIndicator_j = 0 \text{ for } j \neq i \big\}
	\leq 2 \hessianBound_\dt(\bposition_0) \sqrt{\nParam} \dt^2.
\end{equation}
Once we establish the results \eqref{eq:no_event_implication} -- \eqref{eq:one_event_with_neligible_gradient_implication}, combined with the partition \eqref{eq:partition_via_number_of_events} of probability space, they imply that
\begin{align*}
&1 - \probability\left\{
		(\bposition_{\dt}^\superH, \bvelocity_{\dt}^\superH) \equiv (\bposition_{\dt}^\superM, \bvelocity_{\dt}^\superM)
		\ \text{ on } \ \left[ 0, \dt \right]
	\right\} \\
&\hspace{2em} \leq
	\textstyle \frac{1}{2} \gradBound_\dt(\bposition_0)^2 \nParam (\nParam - 1) \dt^2
	+ \max\!\big\{ 2 \hessianBound_\dt(\bposition_0) \nParam^{3/2}, \hessianBound_\dt(\bposition_0) \thinnerspace \nParam (\nParam - 1) \big\}  \dt^2 \\
&\hspace{2em} \leq
	\textstyle \frac{1}{2} \gradBound_\dt(\bposition_0)^2 \nParam^2 \dt^2
	+ 2 \hessianBound_\dt(\bposition_0)\thinnerspace \nParam^2 \dt^2,
\end{align*}
completing the proof.

We first deal with the event type $\{ \eventIndicator_i = 0 \, \text{ for all $i$} \}$.
Note that
\begin{equation*}
\{ \eventIndicator_i = 0 \}
	= \left\{
		\int_0^\dt
			\velocity_{0, i} \spaceBeforePartial \partial_i \potential( \bposition_0 + s \bvelocity_0 )
		\diff s
			\leq \int_0^\dt \big[
				\velocity_{0, i} \spaceBeforePartial \partial_i \potential( \bposition_0 + s \bvelocity_0 )
			\big]^+ \diff s
			< - \log \uniformRv_{0, i}
		\right\}.
\end{equation*}
In other words,  $\eventIndicator_i = 0$ implies the lack of a velocity change along the $i$-th coordinate  for both Hamiltonian and Markovian zigzag.
If $\eventIndicator_i = 0$ for all $i$, therefore, the two zigzags thus stay along the line $\bposition(t) = \bposition_0 + t \bvelocity_0$ with the constant velocity $\bvelocity(t) = \bvelocity_0$ on $[0, \dt]$.
This establishes \eqref{eq:no_event_implication}.

To study implications of the event type $\{ \eventIndicator_i = \eventIndicator_j = 1 \}$ for $i < j$, note that
\begin{equation*}
\int_0^\dt \big[
	\velocity_{0, i} \spaceBeforePartial \partial_i \potential( \bposition_0 + s \bvelocity_0 )
\big]^+ \diff s
	\leq \max_{\ball_\dt(\bposition_0)} | \spaceBeforePartial \partial_i \potential | \, \dt.
\end{equation*}
It follows that
\begin{equation*}
\begin{aligned}
\probability\{ \eventIndicator_i = \eventIndicator_j = 1 \}
	&\leq \probability\Big\{
			- \log \uniformRv_{0, i} \leq \gradBound_\dt(\bposition_0) \dt, \
			- \log \uniformRv_{0, j} \leq \gradBound_\dt(\bposition_0) \dt
		\Big\} \\
	&= \left[
		 1 - \exp \!\big( - \gradBound_\dt(\bposition_0) \dt \big)
	\right]^2,
\end{aligned}
\end{equation*}
from which \eqref{eq:more_than_two_events_implication} follows by observing that $\exp(-\epsilon) \geq 1 - \epsilon$ and hence $1 - \exp(-\epsilon) \leq \epsilon$.

We now turn to the event type $\{ \eventIndicator_i = 1, \,  \eventIndicator_j = 0 \text{ for } j \neq i \}$.
In case $| \partial_i \potential(\bposition_0) | \leq \hessianBound_\dt(\bposition_0) \sqrt{\nParam} \dt$, we have $| \partial_i \potential(\bposition_0 + s \bvelocity_0) | \leq 2 \hessianBound_\dt(\bposition_0) \sqrt{\nParam} \dt$ for $s \in [0, \dt]$ and obtain \eqref{eq:one_event_with_neligible_gradient_implication} by observing that
\begin{align*}
\probability\{ \eventIndicator_i = 1, \,  \eventIndicator_j = 0 \text{ for } j \neq i \}
	&\leq \probability\{ \eventIndicator_i = 1 \} \\
	&= \probability \!\left\{
			- \log \uniformRv_{0, i}
				\leq \int_0^\dt \big[
				\velocity_{0, i} \spaceBeforePartial \partial_i \potential( \bposition_0 + s \bvelocity_0 )
			\big]^+ \diff s
		\right\} \\
	&\leq \probability \!\left\{
			- \log \uniformRv_{0, i}
				\leq  2 \hessianBound_\dt(\bposition_0) \sqrt{\nParam} \dt^2
		\right\} \\
	&\leq 2 \hessianBound_\dt(\bposition_0) \sqrt{\nParam} \dt^2.
\end{align*}%

Finally, we consider the case $| \partial_i \potential(\bposition_0) | > \hessianBound_\dt(\bposition_0) \sqrt{\nParam} \dt$ under the event type $\{ \eventIndicator_i = 1, \,  \eventIndicator_j = 0 \text{ for } j \neq i \}$.
We first focus on behavior of the two zigzags along the $i$-th coordinate since, within the event $\{ \eventIndicator_j = 0 \text{ for }  j \neq i \}$, the first velocity changes on $[0, \dt]$ can only occur in the $i$-th component.
The condition $| \partial_i \potential(\bposition_0) | > \hessianBound_\dt(\bposition_0) \sqrt{\nParam} \dt$ guarantees that the sign of $\partial_i \potential$ is constant on $\ball_\dt(\bposition_0)$ and along both zigzags' trajectories on $[0, \dt]$.
In particular, we have either $\velocity_{0, i} \spaceBeforePartial \partial_i \potential(\bposition_0 + s \bvelocity_0) < 0$ for all $s \in [0, \dt]$ or $\velocity_{0, i} \spaceBeforePartial \partial_i \potential(\bposition_0 + s \bvelocity_0) > 0$ for all $s \in [0, \dt]$.
The former would imply, combined with the condition $\eventIndicator_i = 1$, an impossibility $
- \log \uniformRv_{0, i} 
	\leq \int_0^\dt \left[
		\velocity_{0, i} \spaceBeforePartial \partial_i \potential( \bposition_0 + s \bvelocity_0 )
	\right]^+ \diff s = 0
$.
Therefore, we must instead have $\velocity_{0, i} \spaceBeforePartial \partial_i \potential(\bposition_0 + s \bvelocity_0) > 0$ for all $s \in [0, \dt]$.
This implies that%
\begin{equation*}
\int_0^\dt
	\velocity_{0, i} \spaceBeforePartial \partial_i \potential( \bposition_0 + s \bvelocity_0 )
\diff s
	= \int_0^\dt \big[
		\velocity_{0, i} \spaceBeforePartial \partial_i \potential( \bposition_0 + s \bvelocity_0 )
	\big]^+ \diff s
\end{equation*}
and that both zigzags undergo the first velocity change $\velocity_{\dt, i}^\superH\big( \tau_{0}^{\lowerscript{(1)}} \big) = \velocity_{\dt, i}^\superM\big( \tau_{0}^{\lowerscript{(1)}} \big) = - \velocity_{0, i}$ at the same exact moment $\tau_{0}^{\lowerscript{(1)}}$.
With this new velocity, we have $- \velocity_{0, i} \spaceBeforePartial \partial_i \potential < 0$ on $\ball_\dt(\bposition_0)$ and hence there is no further change in $\velocity_{\dt, i}^\superH(t)$ or $ \velocity_{\dt, i}^\superM(t)$ during $t \in \big[ \tau_{0}^{\lowerscript{(1)}}, \tau_{0}^{\lowerscript{(1)}} + \dt \big]$.


We have so far shown that, under the assumed event type and condition on $| \partial_i \potential(\bposition_0) |$, 1) the two zigzags stay together on $[0, \dt]$ up to the first velocity change and 2) an additional velocity change, which could cause their paths to diverge, cannot occur along the $i$-th coordinate.
In particular, the two zigzag paths on $[0, \dt]$ would diverge only if there were a sign change in the $j$-th velocity component for some $j \neq i$. 
We now show that such events occur with probability at most $O(\dt^2)$ within the event $\{ \eventIndicator_j = 0 \text{ for }  j \neq i \}$, thereby establishing \eqref{eq:one_event_with_non_neligible_gradient_implication}.

We start by observing that the difference in integrals will be small on $[0, \dt]$ regardless of realized zigzag trajectories:
\begin{equation}
\label{eq:bound_on_integral_along_hamiltonian_traj}
\int_0^\dt \velocity_{\dt, j}^\superH(s) \spaceBeforePartial \partial_j \potential \big( \bposition_\dt^\superH(s) \big) \diff s
	\leq \int_0^\dt \velocity_{0, j} \spaceBeforePartial \partial_j \potential( \bposition_0 + s \bvelocity_0 ) \diff s
		+ \hessianBound_\dt(\bposition_0) \dt^2
\end{equation}
as $|\partial_j \potential(\bposition) - \partial_j \potential(\bposition')| \leq \hessianBound_\dt(\bposition_0) \dt$ for any $\bposition, \bposition' \in \ball_\dt(\bposition_0)$.
Similarly,
\begin{equation}
\label{eq:bound_on_integral_along_markovian_traj}
\int_0^\dt \left[
	\velocity_{\dt, j}^\superM(s) \spaceBeforePartial \partial_j \potential \big( \bposition_\dt^\superM(s) \big)
\right]^+ \diff s
	\leq \int_0^\dt \big[
			\velocity_{0, j} \spaceBeforePartial \partial_j \potential( \bposition_0 + s \bvelocity_0 )
		\big]^+ \diff s
		+ \hessianBound_\dt(\bposition_0) \dt^2.
\end{equation}
We then observe that
{\renewcommand{\arraystretch}{1.25}
\begin{align}
&\big\{
	\velocity_{\dt, j}^\superH \not\equiv \velocity_{\dt, j}^\superM \, \text{ on } \ [0, \dt]
\big\} \nonumber \\
&\hspace{2em}\mathop{\mathlarger{\mathlarger{\subset}}}
\left\{
	\begin{array}{l}
	- \log \uniformRv_{0, j}
		\leq \displaystyle \int_0^\dt
				\velocity_{\dt, j}^\superH(s) \spaceBeforePartial \partial_j \potential \!\left( \bposition_\dt^\superH(s) \right)
			\diff s \\
	\hspace{2em} \text{ or } \
	- \log \uniformRv_{0, j}
		\leq \displaystyle \int_0^\dt \left[
			\velocity_{\dt, j}^\superM(s) \spaceBeforePartial \partial_j \potential \!\left( \bposition_\dt^\superM(s) \right)
		\right]^+ \diff s
	\end{array}
\right\} \nonumber \\
&\hspace{2em}\mathop{\mathlarger{\mathlarger{\subset}}}
\left\{
	- \log \uniformRv_{0, j}
		\leq
			\int_0^\dt \big[
				\velocity_{0, j} \spaceBeforePartial \partial_j \potential( \bposition_0 + s \bvelocity_0 )
			\big]^+ \diff s
			+ \hessianBound_\dt(\bposition_0) \dt^2
\right\}, \label{eq:velocity_change_sufficient_condition}
\end{align}
}%
where the latter inclusion follows from \eqref{eq:bound_on_integral_along_hamiltonian_traj} and \eqref{eq:bound_on_integral_along_markovian_traj}.
The inclusion \eqref{eq:velocity_change_sufficient_condition} implies that
\begin{align*}
&\probability\big\{
	\eventIndicator_j = 0, \ \
	\velocity_{\dt, j}^\superH \not\equiv \velocity_{\dt, j}^\superM \, \text{ on } \ [0, \dt]
\big\} \\
&\hspace{2em}\leq \probability\left\{
		\begin{array}{l}
		\displaystyle \int_0^\dt \big[
				\velocity_{0, i} \spaceBeforePartial \partial_i \potential( \bposition_0 + s \bvelocity_0 )
			\big]^+ \diff s
		< - \log \uniformRv_{0, j} \\
		\hspace{4em} \leq
			\displaystyle \int_0^\dt \big[
				\velocity_{0, j} \spaceBeforePartial \partial_j \potential( \bposition_0 + s \bvelocity_0 )
			\big]^+ \diff s
			+ \hessianBound_\dt(\bposition_0) \dt^2
		\end{array}
	\right\} \\
&\hspace{2em} \leq
	\hessianBound_\dt(\bposition_0) \dt^2.
\end{align*}
We can now conclude \eqref{eq:one_event_with_non_neligible_gradient_implication} as follows:
\begin{align*}
&\probability\left\{
	\eventIndicator_i = 1, \ \,  
	\eventIndicator_j = 0 \, \text{ for all } \, j \neq i, \ \,
	(\bposition_{\dt}^\superH, \bvelocity_{\dt}^\superH) \not\equiv (\bposition_{\dt}^\superM, \bvelocity_{\dt}^\superM)
	\text{ on } \left[ 0, \dt \right]
\right\} \\
&\hspace{2em}=
	\probability\left\{
		\eventIndicator_i = 1, \ \,  
		\eventIndicator_j = 0 \, \text{ for all } \, j \neq i, \ \,
		\velocity_{\dt, j}^\superH \not\equiv \velocity_{\dt, j}^\superM \, \text{ on } \, [0, \dt] \, \text{ for some } \, j \neq i
	\right\} \\
&\hspace{2em}\leq
	\probability\left\{
		\eventIndicator_j = 0 \, \text{ for all } \, j \neq i, \ \,
		\velocity_{\dt, j}^\superH \not\equiv \velocity_{\dt, j}^\superM \, \text{ on } \, [0, \dt] \, \text{ for some } \, j \neq i
	\right\} \\
&\hspace{2em}\leq
	\textstyle \sum_{j \neq i}
	\probability\left\{
		\eventIndicator_j = 0, \ \, 
		\velocity_{\dt, j}^\superH \not\equiv \velocity_{\dt, j}^\superM \, \text{ on } \, [0, \dt]
	\right\} \\
&\hspace{2em}\leq
	(\nParam - 1) \hessianBound_\dt(\bposition_0) \dt^2. \qedhere
\end{align*} 
\end{proof}

\section{Mid-point integrator for Hamiltonian zigzag}
\label{sec:midpoint_integrator}

In terms of a practical use of Hamiltonian zigzag and performance comparison against Markovian zigzag, this article focuses exclusively on truncated Gaussian targets where both dynamics allow exact simulations. 
It is possible, however, to apply zigzag \hmc{} to more general target distributions by numerically approximating Hamiltonian zigzag through a reversible integrator. 
Numerical approximation of discontinuous dynamics remains a frontier of research in computational mathematics \protect\citepSupp{fetecau2003nonsmooth_mechanics}, but \protect\citetSupp{nishimura2020discontinuous_hmc} presents one method to qualitatively approximate Hamiltonian zigzag by solving the differential equation \eqref{eq:hamilton_under_laplace_momentum} one coordinate at a time.
This approach can be efficient when the model parameters have a conditional independence structure, which can be exploited to reduce computational costs during the coordinate-wise updates.
Here we present alternative approach that forgoes the coordinate-wise solution and may prove more efficient when no such conditional independence exists.

Our integrator, described in Algorithm~\ref{alg:midpoint_integrator}, is inspired by the mid-point approximation of Hamiltonian-like discontinuous dynamics by \protect\citetSupp{chin2023bouncy_hmc} and is a deterministic analogue of the numerical approximation scheme for Markovian zigzag proposed by \protect\citetSupp{bertazzi2023pdmp_splitting}. 
The idea is to approximate the required integral for updating momentum via the mid-point method
$\int_0^\dt \nabla \potential(\bposition(\tau + s) \diff s) \approx \dt \thinnerspace \nabla \potential\!\left( \bposition\!\left(\tau + \frac{\dt}{2} \right) \right)$. 
The integrator can be easily verified to be volume-preserving and reversible and hence generates a valid Metropolis proposal \protect\citepSupp{fang2014compressible_hmc}. 
Zigzag \hmc{} based on the mid-point integrator is implemented in the \texttt{laplace} branch of \url{https://github.com/aki-nishimura/pynuts}.

\myeditNoPrint{MidpointAlgCounter}{\arabic{algorithm}}
{\spacingset{1.1}
	\setlength{\intextsep}{\baselineskip}
	\algrenewcommand\algorithmicindent{1.5em}
	\begin{algorithm}[H]
		\caption{Numerical approximation of Hamiltonian zigzag with stepsize $\dt$}
		\label{alg:midpoint_integrator}
		\begin{algorithmic}[1]
			\Function{MidpointIntegrator}{$\bposition, \bmomentum, \dt$}
			\State $\bposition \gets \bposition + \frac{\dt}{2} \sign(\bmomentum)$
			\State $\bmomentum' \gets \bmomentum - \dt \thinnerspace \nabla \potential(\bposition)$
			\For{$i = 1, \dots, \nParam$}
			\If{$\sign(\momentum_i) = \sign(\momentum_i')$}
			\State $\momentum_i \gets \momentum_i'$ 
			\Else
			\State $\momentum_i \gets - \momentum_i$ 
			\EndIf
			\EndFor
			\State $\bposition \gets \bposition + \frac{\dt}{2} \sign(\bmomentum)$
			\EndFunction
		\end{algorithmic}
	\end{algorithm}
}

\section{Efficient zigzags on truncated Gaussians}
\label{sec:zigzags_on_truncated_gaussians}
Here we describe how to simulate Hamiltonian and Markovian zigzag on a target distribution $\bposition \sim \normalDist(\bmu, \bPhi^{-1})$ truncated to $\left\{ \sign(\bposition) = \bobservation\in \{\pm 1\}^\nParam \right\}$.
On the support of the target, we have $\nabla \potential(\bposition) = \bPhi (\bposition - \bmu)$ and hence
\begin{alignat*}{2}
\int_0^t \nabla \potential\left( \bposition + s \bvelocity \right) \diff s
	&= \int_0^t \left( \bPhiPosition + s \bPhiVelocity \right) \diff s &\\
	&= t \bPhiPosition + \frac{t^2}{2} \bPhiVelocity
	&\ \text{ where } \
	\bPhiPosition = \bPhi (\bposition - \bmu)
	\ \text{ and } \
	\bPhiVelocity = \bPhi \bvelocity.
	\yesnumber
	\label{eq:line_integral_for_truncated_normal}
\end{alignat*}
In particular, the formula for the velocity switch event of Line~\ref{line:coord_wise_event_time} in Algorithm~\ref{alg:hamiltonian_zigzag_simulation} becomes
\begin{equation}
\label{eq:hamiltonian_event_time_for_truncated_normal}
t_i
	= \inf_{t > 0} \left\{
		|\momentum_i| = \velocity_i \int_0^t \left( \varphi_{\bposition, i} + s \varphi_{\bvelocity, i} \right) \diff s
	\right\}
	= \inf_{t > 0} \left\{
		\momentum_i
			= t \varphi_{\bposition, i} + \frac{t^2}{2} \varphi_{\bvelocity, i}
	\right\}.
\end{equation}
Determining $t_i$ thus amounts to solving a quadratic equation with constraint.

We can handle the truncation via the technique of \protect\citeSupp{neal2010hmc}.
Essentially, when the trajectory reaches a boundary $\{x_i = 0\}$, Hamiltonian zigzag bounces off against it, causing an instantaneous momentum and velocity flip $\momentum_i \gets - \momentum_i$ and $\velocity_i \gets - \velocity_i$.
The trajectory then continues with these new momentum and velocity.
In an algorithm implementation, this amounts to checking for velocity switch events caused by domain boundaries in addition to the native ones caused by the gradient $\nabla \potential$ guiding the dynamics.

Algorithm~\ref{alg:hamiltonian_zigzag_for_truncated_normal} provides pseudo code for simulating Hamiltonian zigzag on a multivariate Gaussian truncated to $\left\{ \sign(\bposition) = \bobservation\in \{\pm 1\}^\nParam \right\}$.
Line~\ref{eq:hamiltonian_boundary_event_determination_start} -- \ref{eq:hamiltonian_boundary_event_determination_end} deal with the truncation, while the rest is essentially a target-specific version of Algorithm~\ref{alg:hamiltonian_zigzag_simulation} with the analytical calculations \eqref{eq:line_integral_for_truncated_normal} and \eqref{eq:hamiltonian_event_time_for_truncated_normal} incorporated.
The function \textproc{minPositiveRoot}$(a, b, c)$, vectorized in the pseudo code, returns the (smaller) positive root of the quadratic equation $a t^2 + b t + c$ if it exists and $\infty$ otherwise.
After updating $\bposition$ and $\bvelocity$ in Line~\ref{line:position_update} and \ref{line:velocity_update}, we can exploit the relations of Line~\ref{line:precision_position_update} and \ref{line:precision_velocity_update} to avoid recomputing the matrix-vector products $\bPhi \bposition$ and $\bPhi \bvelocity$ from scratch.
Note that, while the velocity update $\bvelocity_{\rm new} = \bvelocity_{\rm old} - 2 \velocity_{{\rm old}, i^*} \bm{e}_{i^*}$ corresponds to $\bPhi  \bvelocity_{\rm new} = \bPhi \bvelocity_{\rm old} - 2 \velocity_{{\rm old}, i^*} \bPhi \bm{e}_{i^*}$, the sign in front of $\velocity_{i^*}$ in Line~\ref{line:precision_velocity_update} appears flipped as the velocity has already been updated.

{\spacingset{1.1}
\setlength{\intextsep}{\baselineskip}
\hspace*{-.04\linewidth}
\begin{minipage}{.5\linewidth}
	\begin{algorithm}[H]
	    \caption{Hamiltonian zigzag on \\truncated Gaussian simulated for $t \in [0, \integrationTime]$}
	    \label{alg:hamiltonian_zigzag_for_truncated_normal}
	    \begin{algorithmic}[1]
	        \Function{HamiltonianZigzag}{$\bposition, \hspace*{-.12em} \bmomentum, \hspace*{-.12em} \integrationTime$}
	        	\State $\tau \gets 0$
	            \State $\bvelocity \gets \textrm{sign}(\bmomentum)$
	            \State $\bPhiPosition \gets \bPhi \bposition$ \label{line:Phi_x}
	            \State $\bPhiVelocity \gets \bPhi \bvelocity$ \label{line:Phi_v}
	            \While{$\tau < \integrationTime$}
            		\State \vphantom{$\bm{\uniformRv} \iidSim \unifDist(0, 1)$} 
            		\State \vphantom{$\bm{\interEventTime}^\dagger \gets \Call{firstPositiveTime}{\bPhiVelocity, \bPhiPosition}$}
            		\State $\bm{c} \gets - \bmomentum$
            		\State\begin{varwidth}[t]{\linewidth}
            			$\bm{\interEventTime}_{\rm grad} \gets$\\	
            				\hspace*{1em} $\textproc{minPositiveRoot} (\frac{1}{2} \bPhiVelocity, \bPhiPosition, \bm{c})$
            			\end{varwidth}
	            	\State $\interEventTime_{\rm grad}^* \gets \min_i \interEventTime_{{\rm grad}, i}$
	            	\State $\bm{\interEventTime}_{\rm bdry} = - \bposition / \bvelocity$
	            		\label{eq:hamiltonian_boundary_event_determination_start}
	            	\State $\bm{\interEventTime}_{\rm bdry}[ \bm{\interEventTime}_{\rm bdry} < 0 ] \gets \infty$
	            	\State $\interEventTime_{\rm bdry}^* \gets \min_i \interEventTime_{{\rm bdry}, i}$
	            	\State $\interEventTime^* \gets \min\{ \interEventTime_{\rm grad}^*, \interEventTime_{\rm bdry}^* \}$
	            		\label{eq:hamiltonian_boundary_event_determination_end}
	            	\If{$\tau + \interEventTime^* > \integrationTime$} 
	            		\State $\bposition \gets \bposition + (\integrationTime - \tau) \bvelocity$
	            		\State $\bmomentum \gets \bmomentum - (\integrationTime - \tau) \bPhiPosition - (\integrationTime - \tau)^2 \bPhiVelocity / 2$
	            		\State $\tau \gets \integrationTime$
	            	\Else
	            		\State $\bposition \gets \bposition + \interEventTime^* \bvelocity$
	            			\label{line:position_update}
	            		\State $\bmomentum \gets \bmomentum - t^* \bPhiPosition - t^{*2} \bPhiVelocity / 2$
	            		\State $i^* \gets \argmin_i \interEventTime_i$
	            		\State $\velocity_{i^*} \gets - \velocity_{i^*}$
	            			\label{line:velocity_update}
	            		\State $\bPhiPosition \gets \bPhiPosition + \interEventTime^* \bPhiVelocity$
	            			\label{line:precision_position_update}
	            		\State $\bPhiVelocity \gets \bPhiVelocity + 2 \velocity_{i^*} \bPhi \bm{e}_{i^*}$
	            			\label{line:precision_velocity_update}
	            		\State $\tau \gets \tau + \interEventTime^*$
	            	\EndIf
	            \EndWhile
	            \State \textbf{return} $(\bposition, \bmomentum)$
	        \EndFunction
	    \end{algorithmic}
	\end{algorithm}
\end{minipage}
\nobreak\hspace{.05em} 
\begin{minipage}{.485\linewidth}
	\begin{algorithm}[H]
	    \caption{Markovian zigzag on \\truncated Gaussian simulated for $t \in [0, \integrationTime]$}
	    \label{alg:markovian_zigzag_for_truncated_normal}
	    \begin{algorithmic}[1]
	         \Function{MarkovianZigzag}{$\bposition, \hspace*{-.12em}  \bvelocity, \hspace*{-.12em}  \integrationTime$}
	        	\State $\tau \gets 0$
	            \State
	            \State $\bPhiPosition \gets \bPhi \bposition$
	            \State $\bPhiVelocity \gets \bPhi \bvelocity$
	            \While{$\tau < \integrationTime$}
            		\State $\bm{\uniformRv} \iidSim \unifDist(0, 1)$ \label{line::uniform_rv_for_truncated_gaussian_markovian_zigzag}
            		\State $\bm{\interEventTime}^\dagger \gets \Call{firstPositiveTime}{\bPhiVelocity, \bPhiPosition}$%
					\State $\bm{c} \gets \log \bm{\uniformRv}
									+ \bm{\interEventTime}^\dagger \bvelocity \thinnerspace \bPhiVelocity
									+ \frac{1}{2} \bm{\interEventTime}^{\dagger 2} \bvelocity \thinnerspace \bPhiPosition$%
							\label{line:elemwise_prod_used}
					\State \begin{varwidth}[t]{\linewidth}
			            			$\bm{\interEventTime}_{\rm grad} \gets$ \\
			            			\hspace*{.5em} $\textproc{minPositiveRoot} \raisebox{-.2ex}{$\scriptstyle \geq \bm{\interEventTime}^\dagger$}(\frac{1}{2} \bPhiVelocity, \bPhiPosition, \bm{c})$
			            		\end{varwidth}
	            	\State $\interEventTime_{\rm grad}^* \gets \min_i \interEventTime_{{\rm grad}, i}$
	            	\State $\bm{\interEventTime}_{\rm bdry} = - \bposition / \bvelocity$
	            		\label{eq:markovian_boundary_event_determination_start}
	            	\State $\bm{\interEventTime}_{\rm bdry}[ \bm{\interEventTime}_{\rm bdry} < 0 ] \gets \infty$
	            	\State $\interEventTime_{\rm bdry}^* \gets \min_i \interEventTime_{{\rm bdry}, i}$
	            	\State $\interEventTime^* \gets \min\{ \interEventTime_{\rm grad}^*, \interEventTime_{\rm bdry}^* \}$
	            		\label{eq:markovian_boundary_event_determination_end}
	            	\If{$\tau + \interEventTime^* > \integrationTime$}
	            		\State $\bposition \gets \bposition + (\integrationTime - \tau) \bvelocity$
	            		\State \vphantom{$\bmomentum \gets \bmomentum - \int_0^{\integrationTime - \tau} \nabla \potential\big( \bposition \! + s \bvelocity \big) \diff s.$} 
	            		\State $\tau \gets \integrationTime$
	            	\Else
	            		\State $\bposition \gets \bposition + \interEventTime^* \bvelocity$
	            		\State
	            		\State $i^* \gets \argmin_i \interEventTime_i$
	            		\State $\velocity_{i^*} \gets - \velocity_{i^*}$
	            		\State $\bPhiPosition \gets \bPhiPosition + \interEventTime^* \bPhiVelocity$
	            		\State $\bPhiVelocity \gets \bPhiVelocity + 2 \velocity_{i^*} \bPhi \bm{e}_{i^*}$
	            		\State $\tau \gets \tau + \interEventTime^*$
	            	\EndIf
	            \EndWhile
	            \State \textbf{return} $(\bposition, \bvelocity)$
	        \EndFunction
	    \end{algorithmic}
	\end{algorithm}
\end{minipage}
}

For Markovian zigzag, we can take advantage of the following analytical expressions.
Observe that
\begin{align*}
\int_0^t \left[ \velocity_i \spaceBeforePartial \partial_i \potential \!\left( \bposition + s \bvelocity \right) \right]^+ \diff s
	&= \int_{\interEventTime_i^\dagger}^t \left(
			\velocity_i \varphi_{\bposition, i} + s \velocity_i \varphi_{\bvelocity, i}
		\right) \diff s \\
	&= t \velocity_i \varphi_{\bposition, i}
		+ \frac{t^2}{2} \velocity_i \varphi_{\bvelocity, i}
		- \interEventTime_i^{\dagger} \velocity_i \varphi_{\bposition, i}
		- \frac{\interEventTime_i^{\dagger 2}}{2} \velocity_i \varphi_{\bvelocity, i}
		\yesnumber
		\label{eq:markovian_zigzag_quadratic_equation}
\end{align*}
provided $t \geq \interEventTime_i^\dagger$, where $\interEventTime_i^\dagger$ denotes the time, if it exists, at which the linear function $s \to \velocity_i \varphi_{\bposition, i} + s \velocity_i \varphi_{\bvelocity, i}$ on $s \geq 0$ attains a positive value for the first time.
More precisely, we define
\begin{equation}
\label{eq:first_positive_time}
\begin{aligned}
\interEventTime_i^\dagger
	:= \left\{
		\begin{array}{ll}
		- \varphi_{\bvelocity, i} / \varphi_{\bposition, i}
		&\text{ if } \
			\velocity_i \varphi_{\bposition, i} < 0 
			\ \text{ and } \
			\velocity_i \varphi_{\bvelocity, i} \geq 0, \\
		\infty
		&\text{ if } \
			\velocity_i \varphi_{\bposition, i} < 0 
			\ \text{ and } \
			\velocity_i \varphi_{\bvelocity, i} < 0, \\
		0
		&\text{ otherwise}.
		\end{array}
	\right.
\end{aligned}
\end{equation}
Using \eqref{eq:markovian_zigzag_quadratic_equation}, we can express the velocity switch event formula of Line~\ref{line:coord_wise_event_time} in Algorithm~\ref{alg:markovian_zigzag_simulation} as
\begin{equation*}
t_i
	= \inf_{t > \interEventTime_i^\dagger} \left\{
		- \log \uniformRv_i
			= t \velocity_i \varphi_{\bposition, i}
			+ \frac{t^2}{2} \velocity_i \varphi_{\bvelocity, i}
			- \interEventTime_i^{\dagger} \velocity_i \varphi_{\bposition, i}
			- \frac{\interEventTime_i^{\dagger 2}}{2} \velocity_i \varphi_{\bvelocity, i}
	\right\},
\end{equation*}
if the infimum exists and $t_i = \infty$ otherwise.

Algorithm~\ref{alg:markovian_zigzag_for_truncated_normal} provides pseudo code for simulating Markovian zigzag on a multivariate Gaussian truncated to $\left\{ \sign(\bposition) = \bobservation\in \{\pm 1\}^\nParam \right\}$.
The $\textproc{firstPositiveTime}(\bPhiVelocity, \bPhiPosition)$ function computes $\interEventTime_i^\dagger$'s as in \eqref{eq:first_positive_time}.
The expressions $\bm{\interEventTime}^\dagger \bvelocity \thinnerspace \bPhiVelocity$ and $\bm{\interEventTime}^{\dagger 2} \bvelocity \thinnerspace \bPhiPosition$ in Line~\ref{line:elemwise_prod_used} denote element-wise multiplied vectors.
The $\textproc{minPositiveRoot} \raisebox{-.15ex}{$\scriptstyle \geq \interEventTime_i^\dagger$}(a, b, c)$
function returns the (smaller) positive root as in in Algorithm~\ref{alg:hamiltonian_zigzag_for_truncated_normal} but also requires the root to be greater than $\interEventTime_i^\dagger$.
The input $c_i$ can be infinity here, in which case we require the function to return $\infty$.

\section{Empirical validation of zigzag {\large HMC}'s ergodicity}
\label{sec:empirical_valication_of_hzz_ergodicity}

\myedit{EmpiricalValicationOfErgodicity}{%
	Establishing rigorous ergodicity results for an \mcmc{} algorithm is often challenging \protect\citepSupp{jones2001honest_exploration}.
	For algorithms with deterministic components, such as \hmc{} and \pdmp{}, even a property as basic as irreducibility poses serious barriers to the analysis \protect\citepSupp{durmus2020irreducibility_of_hmc, bierkens2019ergodicity};
	in fact, it was years after \hmc{}'s rise to the prominence and its wide adoption before irreducibility was established under the sufficiently general condition. 
	We therefore leave theoretical investigation of zigzag \hmc{}s ergodicity to future work and, instead, empirically validate its ergodicity in the truncated Gaussian case.
	
	To this end, we take a low-dimensional truncated Gaussian with randomly chosen parameters, run a long chain of zigzag-\Nuts{} on the target, and compare its output to ``ground truth'' reference samples.
	The low-dimensionality allows us to obtain the ground truth by a simple rejection algorithm, which first generate a sample from multivariate Gaussian with the given parameters and then check whether the sample satisfies the constraint. 
	We construct a target of dimension $\nParam = 16$ by truncating a multivariate Gaussian  $\normalDist\!\left( \bmu, \bm{\Sigma} \right)$ to the positive orthant $\{ x_i \geq 0 \}$. 
	The mean parameter is generated as $\mu_i \sim \normalDist(0.25, 0.1^2)$ and covariance as $\bm{\Sigma} \sim \operatorname{Wishart}(\textrm{df} = 2 \nParam, \textrm{mean} = \bm{\Sigma}_0)$, with $\bm{\Sigma}_0 = (1 - \rho) \Id + \rho \bm{1} \mathbf{1}^\transpose$ for $\rho = 0.9$ having the compound symmetric structure of~\eqref{eq:compound_symmetry_covariance}.
	We obtain $10^6$ zigzag-\Nuts{} samples from the corresponding number of iterations with base integration time chosen as described in Section~\ref{sec:zigzag_nuts}. 
	We obtain $3{,}487{,}760$ reference samples after acceptance-rejection of $10^7$~proposals made from the untruncated multivariate Gaussian.
	
	Figure~\ref{fig:hzz_vs_rejection_mean_and_cov_comparison} and \ref{fig:hzz_vs_rejection_marginal_density_comparison} show results of the empirical validation.
	The two sets of Monte Carlo samples yields essentially identical estimates of the target's means and pairwise covariances (Figure~\ref{fig:hzz_vs_rejection_mean_and_cov_comparison}).
	Similarly, the univariate and bivariate marginals along the two randomly chosen coordinates are visually indistinguishable between the two sets of samples (Figure~\ref{fig:hzz_vs_rejection_marginal_density_comparison}).
	We also find the two algorithms to yield visually indistinguishable samples under other target parameter settings.
	Overall, our empirical validation leaves little doubt about zigzag \hmc{}'s ergodicity, at least on truncated Gaussians.%
}

\begin{figure}
\centering
\includegraphics[height=.4\linewidth]{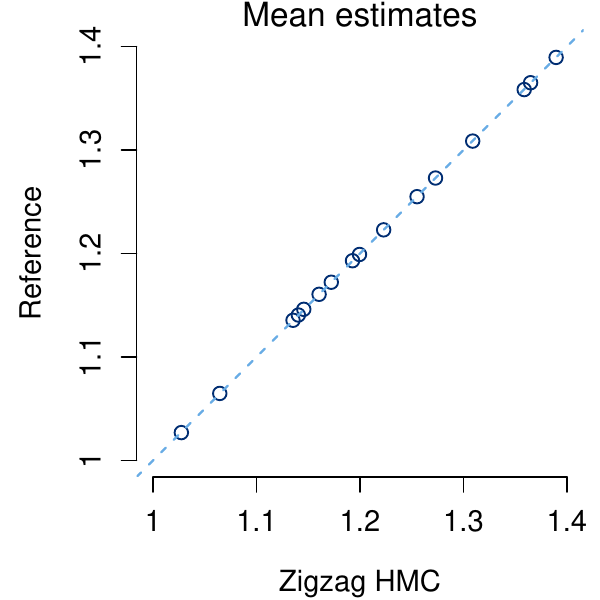}
\includegraphics[height=.4\linewidth]{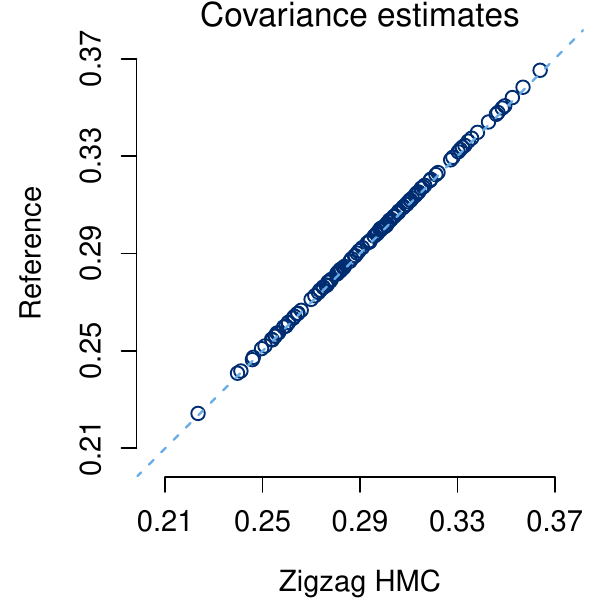}
\caption{%
	Comparison of zigzag-\Nuts{} against reference samples in terms of their estimates of the target's means and pairwise covariances.
}
\label{fig:hzz_vs_rejection_mean_and_cov_comparison}
\end{figure}

\begin{figure}
	\begin{subfigure}{\linewidth}
	\centering
	\includegraphics[height=.33\linewidth, trim=0 {.04\linewidth} 0 {.1\linewidth}, clip]{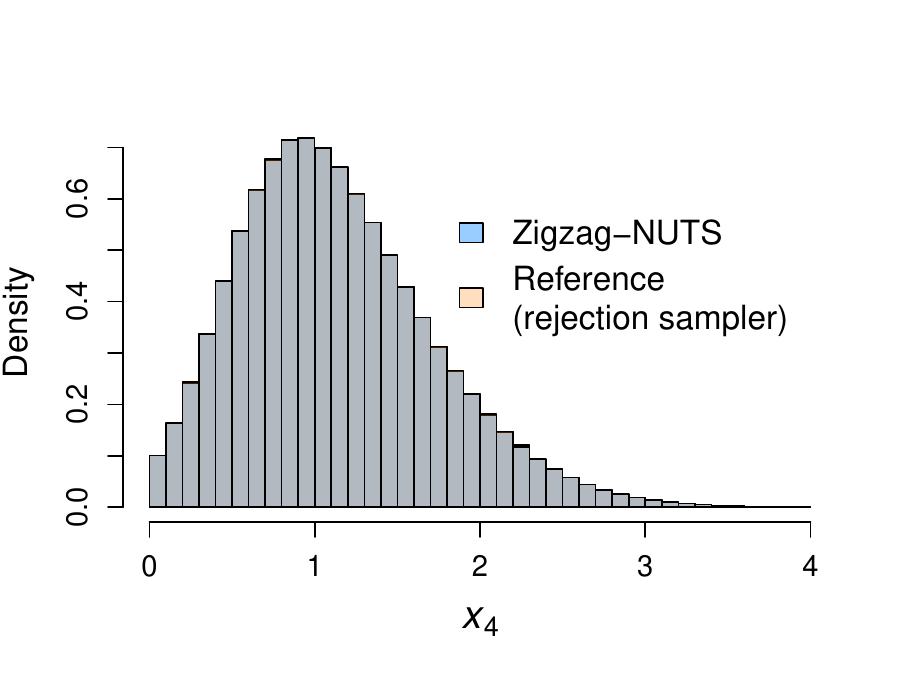}
	\hspace*{-.05\linewidth}
	\includegraphics[height=.33\linewidth, trim={.13\linewidth} {.04\linewidth} 0 {.1\linewidth}, clip]{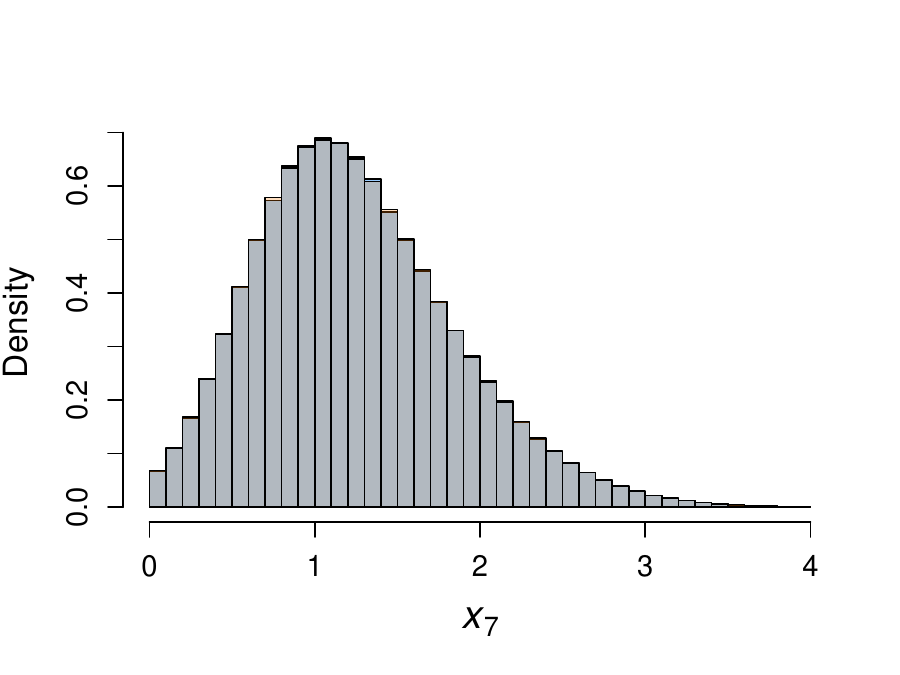}
	\subcaption{%
		Comparison in terms of univariate marginals.
	}
	\label{fig:hzz_vs_rejection_marginal_density_comparison_univariate}
	\end{subfigure}

	\vspace*{.5\baselineskip}
	\begin{subfigure}{\linewidth}
	\centering
	\hspace*{-.04\linewidth}
	\includegraphics[height=.43\linewidth, trim=0 {.07\linewidth} 0 {.05\linewidth}, clip]{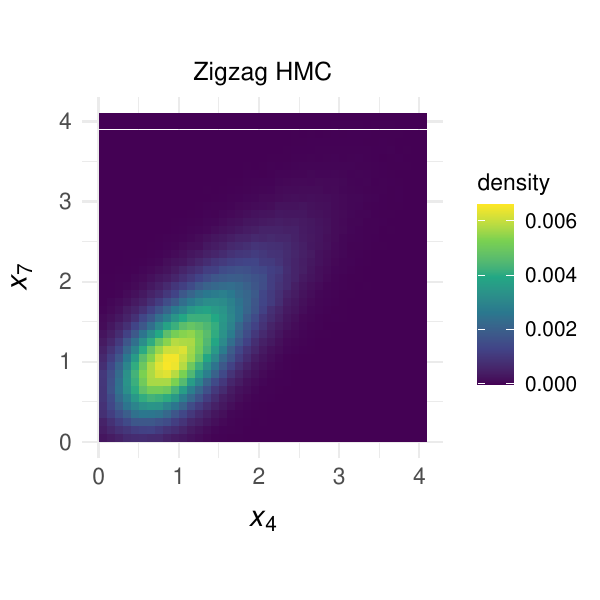}
	\includegraphics[height=.43\linewidth, trim={.085\linewidth} {.07\linewidth} {.15\linewidth} {.05\linewidth}, clip]{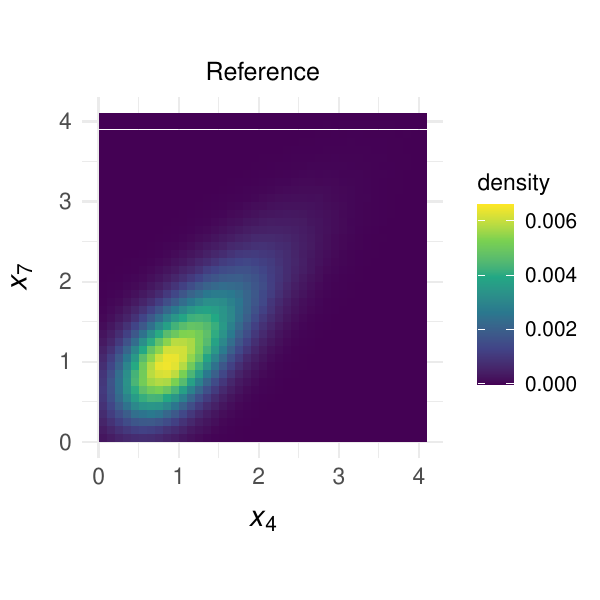}
	\subcaption{%
		Comparison in terms of bivariate marginals.
	}
	\end{subfigure}
\caption{%
	Comparison of zigzag-\Nuts{} against reference samples.
	In the upper subfigure, the two individual colors are hardly visible as the two histograms overlap almost perfectly.
}
\label{fig:hzz_vs_rejection_marginal_density_comparison}
\end{figure}

\FloatBarrier

\newcommand{\nEvent}{n}
\section{Comparison of zigzag {\large HMC}/{\large NUTS}  with existing samplers for truncated Gaussians}
\label{sec:existing_sampler_from_truncated_gaussians} 

\myedit{ComparisonAgainstOtherSamplersIntro}{%
	While this article's main focus is on the connection between the two zigzags and their relative performance, here we study how zigzag \hmc{}/\nuts{} compares against two other state-of-the-art samplers for truncated Gaussians: the algorithms of \protect\citetSupp{pakman2014truncated_normal_hmc} and of \protect\citetSupp{botev2017truncated_normal}.
	In practice, a truncated Gaussian often appears as a conditional distribution within a Gibbs sampler for a more complex joint target distribution \protect\citepSupp{chib1998multivariate_probit, zhang2021phylo_multi_probit}.
	The underlying mean and covariance parameters change from one Gibbs iteration to another in such applications.
	In comparing efficiencies of truncated Gaussian samplers, therefore, it is of interest to consider a cost of pre-processing, which must be repeated when the parameters change, as well as a cost of subsequent computations.
	The discussions and numerical results below complement those in \protect\citetSupp{zhang2022hdtg} and \protect\citetSupp{zhang2022laplace_gause_hmc}.%
}

\subsection{Computational complexity}
\label{sec:existing_sampler_complexity} 

\myedit{ComparisonAgainstOtherSamplersComplexity}{%
	The algorithm of \protect\citetSupp{botev2017truncated_normal} generates independent samples, but both zigzag \hmc{} and \protect\citetSupp{pakman2014truncated_normal_hmc}'s are \mcmc{} algorithms.
	An \mcmc{} algorithm's efficiency is determined by both its per-iteration cost and mixing rate.
	Obtaining quantitative results on an algorithm's mixing rate is challenging, however, even on a restricted family of targets;
	see, for example, \protect\citetSupp{chen2020mixing_of_hmc, lu2020L2_convergence_rate_for_PDMPs} and \protect\citetSupp{bierkens2018scaling_of_pdmp_sampler} for latest developments in quantifying mixing rates of \hmc{} and \pdmp{} algorithms.
	Given this difficulty, here we restrict ourselves to characterizing the per-iteration costs of the three algorithms being considered. 
	Overall sampling efficiencies, which account for mixing rates, are numerically studied in next section.
	
	In considering the per-iteration cost of zigzag \hmc{}, we continue our approach in Section~\ref{sec:simulation} and focus on Zigzag-\Nuts{}, the version of zigzag \hmc{} with the choice of integration time automated through the no-U-turn algorithm.
	As explained in Section~\ref{sec:zigzag_nuts}, the use of the no-U-turn algorithm makes Zigzag-\Nuts{} mostly tuning-free.
	The subsequent simulation of Hamiltonian zigzag requires the one-time $O(\nParam^2)$ cost of computing $\bPhi \bposition$ and $\bPhi \bvelocity$ (Line~\ref{line:Phi_x} and \ref{line:Phi_v} of Algorithm~\ref{alg:hamiltonian_zigzag_for_truncated_normal}), where $\nParam$ as before denotes the number of parameters. 
	Afterward, all the calculations involved are $O(\nParam)$ per velocity switch event.
	Denoting the number of velocity switch events for time duration $\integrationTime$ by $\nEvent_\integrationTime$, the computational cost of each Zigzag-\Nuts{} iteration can therefore be expressed as $O(\nParam \nEvent_\integrationTime)$. 
	With the no-U-turn algorithm, the integration time $\integrationTime$ is chosen dynamically and can vary from one iteration to another.
	The corresponding number of velocity switch events $\nEvent_\integrationTime$ is difficult to quantify but the result of \protect\citetSupp{bierkens2018scaling_of_pdmp_sampler} suggests $\nEvent_\integrationTime = O(\nParam \thinnerspace \integrationTime)$ for Markovian zigzag.
	Assuming that a similar result holds for Hamiltonian zigzag, we expect an $O(\nParam^2 \integrationTime)$ cost for each Zigzag-\Nuts{} iteration.
	
	The algorithm of \protect\citetSupp{pakman2014truncated_normal_hmc} is a version of \hmc{} based on the standard Hamiltonian dynamics with Gaussian momentum. 
	This Gaussian momentum dynamics is simulated exactly as independent harmonic oscillations along the principal components of the covariance matrix, with truncation causing elastic bounces against the parameter space boundaries.
	The simulation requires a one-time calculation of the Cholesky factor $\bm{L}$ of the covariance (or precision) matrix, which incurs an $O(\nParam^3)$ cost. 
	Subsequently, each calculation of next bounce time requires a couple of matrix-vector multiplications by $\bm{V}^\transpose \bm{L}^\transpose$, 
	where the columns of $\bm{V}$ consists of vectors orthogonal to the boundaries, and incurs an $O(\nParam^2)$ cost under the orthant constraint. 
	The dynamics is simulated for duration $\integrationTime$, with a value of $\integrationTime \in [\pi / 8, \pi / 2]$ typically recommended for efficient mixing \protect\citepSupp{pakman2014truncated_normal_hmc, zhang2022hdtg}.
	Denoting the resulting number of boundary events by $\nEvent_{\textrm{bd}}$, the total cost for each iteration of the algorithm is $O(\nParam^3 + \nParam^2 \nEvent_{\textrm{bd}})$.
	For a multivariate Gaussian target with independent coordinates and orthant constraint, the expected number of boundary events is given by $\nParam \thinspace \integrationTime / \pi$ at stationarity.
	We therefore take $\nEvent_{\textrm{bd}} \approx \nParam \thinspace \integrationTime / \pi$ as a generally reasonable approximation, which we in fact observe to be the case when applying the algorithm to the phylogenetic posterior of Section~\ref{sec:phylogenetic_probit} in the numerical result below.
	With this approximation, the cost of each iteration can be expressed as $O(\nParam^3 + \nParam^3 \integrationTime / \pi)$.
	
	The algorithm of \protect\citetSupp{botev2017truncated_normal} is a rejection sampler based on a minimax tilted proposal density. 
	Finding the minimax tilting requires a constrained convex optimization, each iteration of which requires $O(\nParam^3)$ operations.
	The subsequent rejection sampling requires only $O(\nParam^2)$ operations per attempt; 
	however, it suffers from the usual curse of dimensionality in rejection sampling, with the acceptance rate in general decaying exponentially in the number of parameters.
	In fact, \protect\citetSupp{botev2017truncated_normal} cautions that the algorithm may not work well beyond the dimension $d > 100$ except in limited situations.%
}

\subsection{Numerical results}
\label{sec:existing_sampler_numerical_results} 

\myedit{ComparisonAgainstOtherSamplersNumericalResults}{ \phantom{}%
	\protect\citetSupp{zhang2022hdtg} compare the three algorithms on truncated Gaussians with dimension $d = 100, 400, 800, \thinnerspace \text{and}\ 1{,}600$ with covariance matrices taken from
	i) the uniform \textsc{lkj} distribution \protect\citepSupp{lewandowski2009LKJ_prior}, 
	ii) the compound symmetric example of Section~\ref{sec:synthetic_examples}, and 
	iii) the phylogenetic posterior of Section~\ref{sec:phylogenetic_probit} with sub-sampling to reduce the dimension.
	Their general finding is that Zigzag-\Nuts{} has a clear advantage in higher dimensions, except in the compound symmetric case where the pre-processing step of \protect\citetSupp{pakman2014truncated_normal_hmc} to decorrelate the parameters leads to better performances.
	Additionally, \protect\citetSupp{zhang2022laplace_gause_hmc} demonstrate the use of Hamiltonian zigzag and the splitting technique \protect\citepSupp{shahbaba2014split_hmc, nishimura2020discontinuous_hmc} to improve efficiency by jointly updating highly correlated parameters, 
	an opportunity not offered by the other two algorithms. 
	
	Here we provide an additional insight on the algorithms' relative performance by benchmarking them on the full 11,235-dimensional posterior from the phylogenetic example of Section~\ref{sec:phylogenetic_probit}. 
	For the algorithm of \protect\citetSupp{pakman2014truncated_normal_hmc}, we use an implementation from the \texttt{hdtg} package \protect\citepSupp{zhang2022hdtg} optimized with \texttt{Rcpp} \protect\citepSupp{eddelbuettel2011rcpp} since the original implementation from the \texttt{tmg} package has been removed from \textsc{cran} on May 22nd, 2021.
	We measure its performance under three tuning parameter settings;
	from the recommended range $\integrationTime \in [\pi / 8, \pi / 2]$, we take a shorter, medium, and longer range $\integrationTime \sim \unifDist(\pi / 8, \pi / 4)$, $\sim \unifDist(3 \pi / 16, 3 \pi / 8)$, and $\sim \unifDist(\pi / 8, \pi / 4)$ while incorporating the jittering to prevent periodic behaviors \protect\citepSupp{neal2010hmc, zhang2022hdtg}.
	For the algorithm of \protect\citetSupp{botev2017truncated_normal}, we use the implementation from the \texttt{TruncatedNormal} package on \textsc{cran}.
	
	Table~\ref{tab:comparison_against_other_samplers_ess_per_time} summarizes the numerical results.
	The table shows the absolute \ess{} per hour, instead of the relative ones shown in Table~\ref{tab:ess_per_time_phylo_probit}, to provide a better sense of the actual computational burdens. 
	Zigzag-\Nuts{} outperforms the optimally tuned version of \protect\citetSupp{pakman2014truncated_normal_hmc}'s by a factor of 262 and of 59.6 in terms of the minimum \ess{} and of the \ess{} along the principal component.
	\protect\citetSupp{botev2017truncated_normal}'s fails to generate any sample after 3 days, so we only report the corresponding upper-bound on its \ess{}.
	As in Section~\ref{sec:phylogenetic_probit}, the \ess{}'s for Zigzag-\Nuts{} are calculated and averaged from five replicates of 1,500 iterations from stationarity.
	The \ess{}'s for \protect\citetSupp{pakman2014truncated_normal_hmc}'s are calculated in the same manner, using five replicates of 100 iterations.
	The reported \ess{}'s account for the pre-processing costs, though they contribute little to the overall costs in this example since the subsequent computations take far longer under the specified chain lengths.%
}

\begin{table}
\centering
\caption{%
	\Ess{} per hour under the phylogenetic probit posterior of Section~\ref{sec:phylogenetic_probit}.
	The simulation settings and \ess{} calculations are identical to those in Section~\ref{sec:phylogenetic_probit} except for two minor differences: 
	1) to simplify implementation of the additional algorithms and ensure their fair comparison against zigzag-\Nuts{}, we use the R implementation of zigzag-\Nuts{} available in the \texttt{hdtg} package \protect\citepSupp{zhang2022hdtg} instead of  the Java implementation; and
	2) to avoid incurring unnecessary expense on Amazon Elastic Compute Cloud, we carried out the computations on the Joint High Performance Computing Exchange (\jhpce{}) cluster at Johns Hopkins University (\url{https://jhpce.jhu.edu/}), allocating an \textsc{amd epyc} 7702 processor and 32 gigabytes of memory to each chain.
	Neither of these modifications should have material effects on the relative performances of the two algorithms.
}
\vspace*{.3\baselineskip}
\label{tab:comparison_against_other_samplers_ess_per_time}
\begin{tabular}[t]{lcccc}
\toprule
\multirow{2}{*}{Phylogenetic probit \hspace*{-1ex} \rule{0pt}{12pt}} 
& \multicolumn{2}{c}{\Ess{} per hour} \\
\cmidrule(l{3pt}r{3pt}){2-3} 
\cmidrule(l{3pt}r{3pt}){4-5}
& \hspace{2ex} min \hspace{2ex} & PC\\
\midrule
\rowcolor{defaultNutsColor} Zigzag-\Nuts{}
& 160 & 398 \\
\rowcolor{harmonicHmcColor} \protect\citetSupp{pakman2014truncated_normal_hmc} with $\integrationTime \sim \unifDist(\pi / 8, \pi / 4 )$ & 0.244 & 2.46  \\
\rowcolor{harmonicHmcColor} \protect\citetSupp{pakman2014truncated_normal_hmc} with $\integrationTime \sim \unifDist(3 \pi / 16, 3 \pi / 8 )$ & 0.286 & 4.03 \\
\rowcolor{harmonicHmcColor} \protect\citetSupp{pakman2014truncated_normal_hmc} with $\integrationTime \sim \unifDist(\pi / 4, \pi / 2 )$ & 0.610 & 6.68 \\
\rowcolor{botevColor} \protect\citetSupp{botev2017truncated_normal} & $\leq 0.0139$ & $\leq 0.0139$ \\
\bottomrule
\end{tabular}
\end{table}


\newcommand{\acceptThreshold}{\texttt{accept\_thresh}}
\newcommand{\nAcceptableStates}{\texttt{n\_acceptable}}
\newcommand{\acceptProb}{\texttt{accept\_prob}}
\newcommand{\uturned}{\texttt{u\_turned}}
\newcommand{\height}{\texttt{height}}
\newcommand{\direction}{\texttt{direc}}
\newcommand{\bpositionInput}{\bposition_{\textrm{curr}}}
\newcommand{\bpositionFront}{\bposition_{\textrm{front}}}
\newcommand{\bpositionRear}{\bposition_{\textrm{rear}}}
\newcommand{\bpositionSample}{\bposition_{\textrm{next}}}
\newcommand{\bmomentumInput}{\bmomentum_{\textrm{curr}}}
\newcommand{\bmomentumFront}{\bmomentum_{\textrm{front}}}
\newcommand{\bmomentumRear}{\bmomentum_{\textrm{rear}}}
\newcommand{\bmomentumSample}{\bmomentum_{\textrm{next}}}

\section{No-U-turn algorithm for reversible dynamics}
\label{sec:nuts_with_reversible_dynamics}

\myedit{NutsForReversibleDynamics}{%
	In presenting the no-U-turn algorithm, \protect\citetSupp{hoffman2014nuts} focus on its use with the leapfrog integrator for simulating Hamiltonian dynamics based on Gaussian momentum.
	The algorithm in fact applies more generally, however, to any \hmc{}-like sampler based on reversible volume-preserving dynamics.
	In essence, out of any given reversible map, the no-U-turn algorithm constructs a transition kernel that preserves the joint position-momentum distribution while avoiding computationally wasteful simulation of dynamics beyond a U-turn in the simulated trajectory. 
	The auxiliary variable technically need not have the interpretation as actual momentum in Hamiltonian dynamics; 
	see \protect\citetSupp{fang2014compressible_hmc} for a more general notion of reversible dynamics and their use in Monte Carlo simulation.
	
	To describe the no-U-turn algorithm in this more general setting, let $\solutionOp(\bposition, \bmomentum)$ denote a reversible volume-preserving map; 
	i.e.\ $\momentumFlipOp \circ \solutionOp = (\momentumFlipOp \circ \solutionOp)^{-1}$ for the momentum flip operator $\momentumFlipOp(\bposition, \bmomentum) = (\bposition, - \bmomentum)$ and the Jacobian of $\solutionOp$ has its determinant equal to $\pm 1$.
	The map $\solutionOp$ corresponds to the leapfrog integrator in the context of \protect\citetSupp{hoffman2014nuts} and to Hamiltonian zigzag simulated for $\baseIntegrationTime$ duration in the context of Section~\ref{sec:zigzag_nuts}. 
	Algorithm~\ref{alg:no_u_turn} and \ref{alg:no_u_turn_sub_functions} describes the no-U-turn transition kernel in terms of a generic map $\solutionOp$. 
	When used with the leapfrog integrator, the algorithm coincides with the transition kernel in Algorithm~3 of \protect\citetSupp{hoffman2014nuts}. 
	
	Given a current state $(\bpositionInput, \bmomentumInput)$, the no-U-turn transition kernel searches for an appropriate next state along a trajectory $\{\solutionOp^n(\bpositionInput, \bmomentumInput) \text{ for } n = 0, \pm 1, \pm 2, \ldots\}$ endowed with a binary tree structure.
	To ensure the transition kernel's symmetry, the algorithm doubles the trajectory either forward or backward according to $\direction \gets \unifDist(\{-1, + 1\})$ and checks for U-turn occurrences between the front- and rear-most states of every subtrees within the binary tree. 
	The trajectory simulation is terminated when a U-turn behavior is detected.
	Constructed out of a deterministic map, the no-U-turn transition kernel is reducible on its own; 
	each iteration therefore must be accompanied by a momentum refreshment step $\bmomentum \sim \momentumMarginal(\cdot)$ as in \hmc{}.
	For further details on the no-U-turn algorithm, we refer readers to \protect\citetSupp{hoffman2014nuts}.%
}

\myeditNoPrint{NutsAlgCounter}{\arabic{algorithm}}
{\spacingset{1.1}
\setlength{\intextsep}{\baselineskip}
\algrenewcommand\algorithmicindent{1.5em}
\begin{algorithm}[H]
	\caption[caption]{%
		Construction, based on the no-U-turn criterion, of a transition kernel from a reversible volume-preserving map $\solutionOp(\cdot)$ using sub-routines described in Algorithm~\ref{alg:no_u_turn_sub_functions}
	}
	\label{alg:no_u_turn}
	\begin{algorithmic}[1]
		\Function{noUturnTransition}{$\solutionOp(\cdot), \bpositionInput, \bmomentumInput$}
			\State $\acceptThreshold \sim \unifDist(0, \pi(\bpositionInput, \bmomentumInput))$ 
			\State $\nAcceptableStates \gets 1$
        	\State $\uturned \gets \mathtt{False}$
        	\State $\height \gets 0$
        	\State $(\bpositionFront, \bmomentumFront) 
        		\gets (\bpositionRear, \bmomentumRear) 
        		\gets (\bpositionSample, \bmomentumSample) 
        		\gets(\bpositionInput, \bmomentumInput)$
        	\While{not $\uturned$}
        		\State $\direction \gets \unifDist(\{-1, + 1\})$
        		\State \begin{varwidth}[t]{\linewidth}
        			$(\bpositionFront, \thinnerspace \bmomentumFront), \thinnerspace (\bpositionRear, \bmomentumRear), \thinnerspace (\bpositionSample, \bmomentumSample), \thinnerspace \uturned, \thinnerspace \nAcceptableStates$ \\
	            	\hspace*{\algorithmicindent} $\gets \textproc{doubleTree}$%
	            	  	\begin{varwidth}[t]{\linewidth}
	            		$(\solutionOp(\cdot),\thinnerspace \direction, \thinnerspace \bpositionFront, \thinnerspace \bmomentumFront, \thinnerspace \bpositionRear, \thinnerspace \bmomentumRear, \thinnerspace \bpositionSample, \thinnerspace \bmomentumSample,$\\ 
	            		$\hphantom{(}\thinnerspace \uturned, \thinnerspace \nAcceptableStates, \thinnerspace \height, \thinnerspace \acceptThreshold)$
	            	 	\end{varwidth}
           			\end{varwidth}
           		\State $\height \gets \height + 1$
        	\EndWhile
        	\State \Return $(\bpositionSample, \bmomentumSample)$
	    \EndFunction
	    \\

\algblockdefx{DoubleTrajectoryFunction}{EndDoubleTrajectoryFunction}{%
   	\algorithmicfunction\ \textproc{doubleTree}%
  	\begin{varwidth}[t]{\linewidth}
	$(\solutionOp(\cdot),\thinnerspace \direction, \thinnerspace \bpositionFront, \thinnerspace \bmomentumFront, \thinnerspace \bpositionRear, \thinnerspace \bmomentumRear, \thinnerspace \bpositionSample, \thinnerspace \bmomentumSample,$\\ 
	$\hphantom{(}\thinnerspace \uturned, \thinnerspace \nAcceptableStates, \thinnerspace \height, \thinnerspace \acceptThreshold)$
 	\end{varwidth}
}{end function}
\makeatletter
\ifthenelse{\equal{\ALG@noend}{t}}%
  {\algtext*{EndDoubleTrajectoryFunction}}
  {}%
\makeatother

	    \DoubleTrajectoryFunction
       		\If{$\direction == + 1$}
       			\State \begin{varwidth}[t]{\linewidth}
        			$(\bpositionFront, \thinnerspace \bmomentumFront), \thinnerspace \underline{\hphantom{(\bpositionRear, \bmomentumRear)}}, \thinnerspace (\bpositionSample', \bmomentumSample'), \thinnerspace \uturned', \thinnerspace \nAcceptableStates'$ \\
	            	\hspace*{\algorithmicindent} $\gets \textproc{buildNextTree}(\solutionOp(\cdot), \thinnerspace +1, \thinnerspace \bpositionFront, \thinnerspace \bmomentumFront, \thinnerspace \height, \thinnerspace \acceptThreshold)$
           			\end{varwidth} 
       		\Else
       			\State \begin{varwidth}[t]{\linewidth}
        			$\underline{\hphantom{(\bpositionFront, \thinnerspace \bmomentumFront)}}, \thinnerspace (\bpositionRear, \bmomentumRear), \thinnerspace (\bpositionSample', \bmomentumSample'), \thinnerspace \uturned', \thinnerspace \nAcceptableStates'$ \\
	            	\hspace*{\algorithmicindent} $\gets \textproc{buildNextTree}(\solutionOp(\cdot), \thinnerspace - 1, \thinnerspace \bpositionRear, \thinnerspace \bmomentumRear, \thinnerspace \height, \thinnerspace \acceptThreshold)$
           			\end{varwidth} 
       		\EndIf
       		\State $\uturned \gets \uturned \text{ or } \uturned' \text{ or }  \textproc{checkUturn}(\bpositionFront, \thinnerspace \bmomentumFront, \thinnerspace \bpositionRear, \thinnerspace \bmomentumRear)$ 
       		\If{not $\uturned'$} 
       			\State \commentMarker\ Merge current and next trees provided no U-turn occurred within the next
       			\If{$\min\{1, \thinnerspace \nAcceptableStates' \thinnerspace / \,  \nAcceptableStates\} > u \sim \unifDist(0, 1)$}
       				\State $(\bpositionSample, \bmomentumSample) \gets (\bpositionSample', \bmomentumSample')$
       			\EndIf 
        		\State $\nAcceptableStates \gets \nAcceptableStates + \nAcceptableStates'$
       		\EndIf
       		\State \Return $(\bpositionFront, \thinnerspace \bmomentumFront), \thinnerspace (\bpositionRear, \bmomentumRear), \thinnerspace (\bpositionSample, \bmomentumSample), \thinnerspace \uturned, \thinnerspace \nAcceptableStates$
       	\EndDoubleTrajectoryFunction
       	\\
   	   	\Function{checkUturn}{$\bpositionFront, \thinnerspace \bmomentumFront, \thinnerspace \bpositionRear, \thinnerspace \bmomentumRear$}
   	   		\State \Return $\langle \bpositionFront - \bpositionRear, \thinnerspace \bmomentumFront \rangle < 0$ \,or\, $\langle \bpositionRear - \bpositionFront, \thinnerspace - \bmomentumRear \rangle < 0$
   		\EndFunction
	\end{algorithmic}
\end{algorithm}

\begin{algorithm}[H]
	\caption[caption]{Sub-routines for the no-U-turn transition of Algorithm~\ref{alg:no_u_turn}}
	\label{alg:no_u_turn_sub_functions}
	\begin{algorithmic}[1]
	\Function{buildNextTree}{$\solutionOp(\cdot), \thinnerspace \direction, \thinnerspace \bposition_0, \bmomentum_0, \thinnerspace \height, \thinnerspace \acceptThreshold$}
    	\If{$\height == 0$} 
	    	\State \Return $\textproc{buildSingletonTree}(\solutionOp(\cdot), \thinnerspace \direction, \thinnerspace \bposition_0, \bmomentum_0, \thinnerspace \acceptThreshold)$
    	\EndIf
    	\State \commentMarker\ Build the first half of the tree
    	\State \begin{varwidth}[t]{\linewidth}
      			$(\bpositionFront, \thinnerspace \bmomentumFront), \thinnerspace (\bpositionRear, \bmomentumRear), \thinnerspace (\bpositionSample, \bmomentumSample), \thinnerspace \uturned, \thinnerspace \nAcceptableStates$ \\
           	\hspace*{\algorithmicindent} $\gets \textproc{buildNextTree}(\solutionOp(\cdot), \bposition_0, \bmomentum_0, \thinnerspace \direction, \thinnerspace \height - 1, \thinnerspace \acceptThreshold)$
         	\end{varwidth}
		\If{$\uturned$} 
	       	\State \commentMarker\ No point of simulating dynamics any further if U-turn has already occurred
	       	\State \Return $(\bpositionFront, \thinnerspace \bmomentumFront), \thinnerspace (\bpositionRear, \bmomentumRear), \thinnerspace (\bpositionSample, \bmomentumSample), \thinnerspace \uturned, \thinnerspace \nAcceptableStates$
		\EndIf
		\State \commentMarker\ Build the latter half of the tree and merge with the first half
   		\If{$\direction == + 1$}
 			\State \begin{varwidth}[t]{\linewidth}
  			$(\bpositionFront, \thinnerspace \bmomentumFront), \thinnerspace \underline{\hphantom{(\bpositionRear, \bmomentumRear)}}, \thinnerspace (\bpositionSample', \bmomentumSample'), \thinnerspace \uturned', \thinnerspace \nAcceptableStates'$ \\
       	\hspace*{\algorithmicindent} $\gets \textproc{buildNextTree}(\solutionOp(\cdot), \thinnerspace +1, \thinnerspace \bpositionFront, \thinnerspace \bmomentumFront, \thinnerspace \height - 1, \thinnerspace \acceptThreshold)$
     			\end{varwidth} 
 		\Else
 			\State \begin{varwidth}[t]{\linewidth}
  			$\underline{\hphantom{(\bpositionFront, \thinnerspace \bmomentumFront)}}, \thinnerspace (\bpositionRear, \bmomentumRear), \thinnerspace (\bpositionSample', \bmomentumSample'), \thinnerspace \uturned', \thinnerspace \nAcceptableStates'$ \\
       	\hspace*{\algorithmicindent} $\gets \textproc{buildNextTree}(\solutionOp(\cdot), \thinnerspace -1, \thinnerspace \bpositionRear, \thinnerspace \bmomentumRear, \thinnerspace \height - 1, \thinnerspace \acceptThreshold)$
     			\end{varwidth} 
   		\EndIf
      	\State $\uturned \gets \uturned \text{ or } \uturned' \text{ or }  \textproc{checkUturn}(\bpositionFront, \thinnerspace \bmomentumFront, \thinnerspace \bpositionRear, \thinnerspace \bmomentumRear)$ 
   		\If{not $\uturned'$}
 			\If{$\thinnerspace \nAcceptableStates' \thinnerspace / \,  (\nAcceptableStates + \nAcceptableStates') > u \sim \unifDist(0, 1)$}
 				\State $(\bpositionSample, \bmomentumSample) \gets (\bpositionSample', \bmomentumSample')$
 			\EndIf 
  		\State $\nAcceptableStates \gets \nAcceptableStates + \nAcceptableStates'$
 		\EndIf
      	\State \Return $(\bpositionFront, \thinnerspace \bmomentumFront), \thinnerspace (\bpositionRear, \bmomentumRear), \thinnerspace (\bpositionSample, \bmomentumSample), \thinnerspace \uturned, \thinnerspace \nAcceptableStates$
   	\EndFunction \\
   	\Function{buildSingletonTree}{$\solutionOp(\cdot), \thinnerspace \direction, \thinnerspace \bposition_0, \bmomentum_0, \thinnerspace \acceptThreshold$}
   		\State $(\bposition_1, \thinnerspace \direction * \bmomentum_1) \gets \solutionOp(\bposition_0, \thinnerspace \direction * \bmomentum_0)$
   		\State $\nAcceptableStates \gets \indicator\{ \pi(\bposition_1, \bmomentum_1) > \acceptThreshold \}$
   		\State $\uturned \gets \mathtt{False}$
   		\State $(\bpositionFront, \bmomentumFront) 
       		\gets (\bpositionRear, \bmomentumRear) 
       		\gets (\bpositionSample, \bmomentumSample) 
       		\gets (\bposition_1, \bmomentum_1)$
       	\State \Return $(\bpositionFront, \thinnerspace \bmomentumFront), \thinnerspace (\bpositionRear, \bmomentumRear), \thinnerspace (\bpositionSample, \bmomentumSample), \thinnerspace \uturned, \thinnerspace \nAcceptableStates$
   	\EndFunction 
	\end{algorithmic}
\end{algorithm}
}

\section{Choosing base integration time for Zigzag-\Nuts{}}
\label{sec:base_integration_time_for_nuts}
In Section~\ref{sec:zigzag_nuts}, we discuss how to automatically tune an integration time for Hamiltonian zigzag via the no-U-turn algorithm. 
As we mention there, the base integration time $\baseIntegrationTime = \eigenValue_{\min}^{-1/2}(\bPhi) \baseIntegrationTimeMultiplier$ for $\baseIntegrationTimeMultiplier = 0.1$ works well in a broad range of problems but is not always optimal.
We explain and illustrate this behavior using the posteriors of Section~\ref{sec:simulation}.

As discussed in Section~\ref{sec:synthetic_examples}, the compound symmetric posterior has its probability tightly concentrated along the principal component and is otherwise symmetric in all the other directions.
Due to this extreme structure, it takes a while even for Hamiltonian zigzag to gain momentum and start traveling along the least constrained direction. 
In particular, while Hamiltonian zigzag eventually finds its way over a sufficiently long time period, it also undergoes many local U-turns at a shorter time scale (Figure~\ref{fig:hamiltonian_traj_and_sq_traveled_dist_on_cs_target}).
Such local behavior can cause the no-U-turn algorithm to terminate its trajectory simulation prematurely, leading to an integration time that is too small and sub-optimal \protect\citepSupp{neal2012nuts_blog_post}. 
Since the algorithm checks for U-turns only at discrete time points $t = 2^\ell \baseIntegrationTime$ for $\ell \geq 0$, however,
we can to a reasonably extent prevent 
premature trajectory terminations in Zigzag-\Nuts{} by choosing $\baseIntegrationTime$ larger than time scales of local U-turns. 
For this reason, Zigzag-\Nuts{} performs better with the choice $\baseIntegrationTimeMultiplier > 0.1$ for these compound symmetric posteriors.
As seen in Table~\ref{tab:additional_ess_per_time_compound_symmetric}, we observe $4$ to $7$-fold increase in \ess{} for the $\rho = 0.9$ case when moving from $\baseIntegrationTimeMultiplier = 0.1$ to $\baseIntegrationTimeMultiplier = 1$.
\begin{figure}
\centering
	\begin{minipage}{.425\linewidth}
	\includegraphics[width=\linewidth]{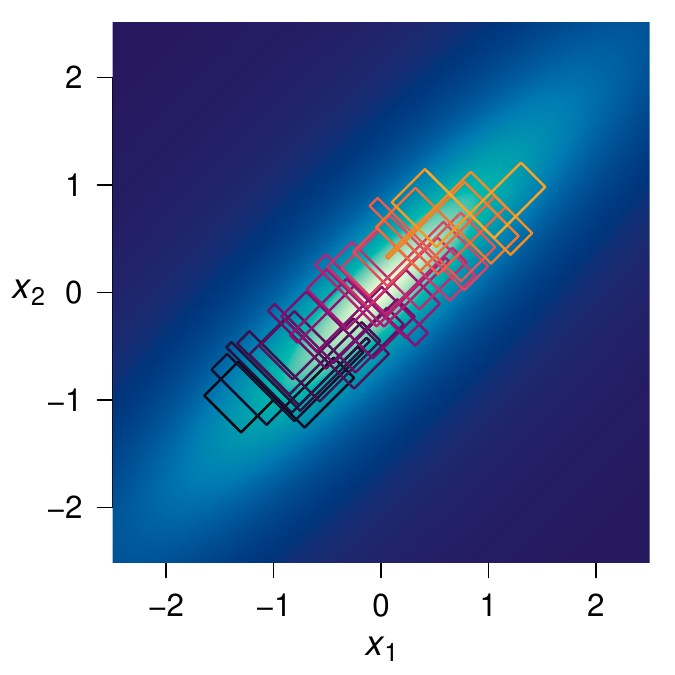}
	\end{minipage}
	\nobreak\hspace{.1em} 
	\begin{minipage}{.555\linewidth}
	\includegraphics[width=\linewidth]{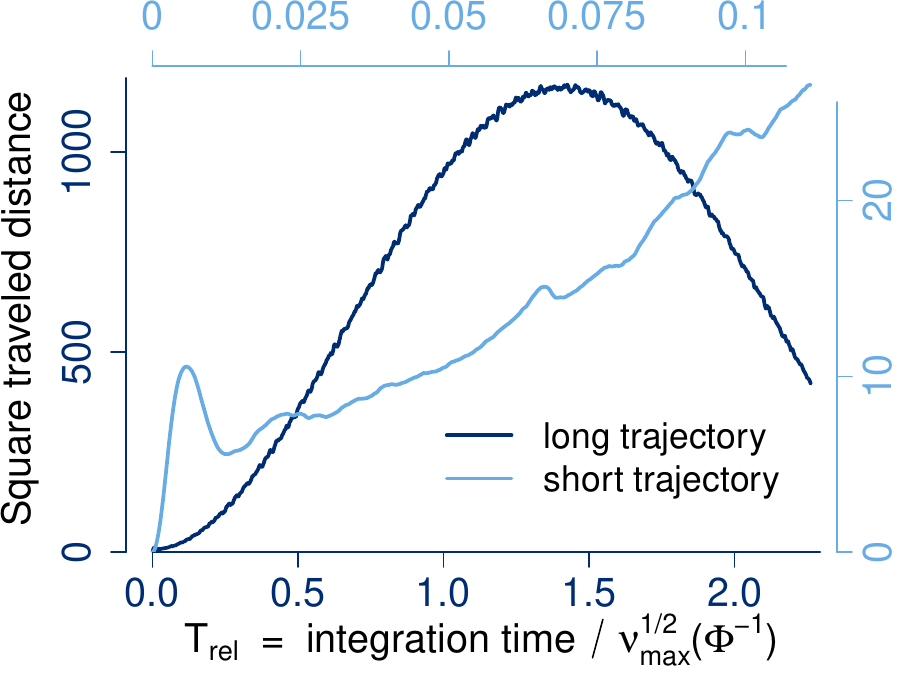}
	\end{minipage}
\caption{%
	Trajectory of the first two position coordinates (left) and squared distance $\| \bposition(t) - \bposition_0 \|^2$ (right) of Hamiltonian zigzag without momentum refreshment on the $1{,}024$-dimensional truncated compound symmetric target with $\rho = 0.99$.
	The dynamics is simulated for $4.5 \times 10^5$ linear segments, starting from a pre-specified position $\position_i = 0.5$ and random momentum $\momentum_i \sim \laplaceDist(\textrm{scale} = 1)$.
	In the left figure, the line segment colors change from darkest to lightest as the dynamics evolves.
	To the right, the squared traveled distance is plotted as a function of the relative integration time. 
	To help examine the wiggliness at shorter time scales, we zoom into the first $5\%$ of the trajectory and plot it in the lighter blue color.
}
\label{fig:hamiltonian_traj_and_sq_traveled_dist_on_cs_target}
\end{figure}

\begin{table}
\caption{%
	Relative \ess{} per computing time under the compound symmetric posteriors.
	The table follows the same format as Table~\ref{tab:ess_per_time_compound_symmetric} but has additional rows for the results based on Zigzag-\Nuts{} with $\baseIntegrationTimeMultiplier = 1$.
	\label{tab:additional_ess_per_time_compound_symmetric}
}
\centering
\begin{tabular}[t]{lcccc}
\toprule
& \multicolumn{4}{c}{Relative \Ess{} per time}\\
\cmidrule(l{3pt}r{3pt}){2-5}
\multicolumn{1}{c}{Compound symmetric} 
& \multicolumn{2}{c}{$\rho = 0.9$} & \multicolumn{2}{c}{$\rho = 0.99$} \\
\cmidrule(l{3pt}r{3pt}){2-3} \cmidrule(l{3pt}r{3pt}){4-5}
\multicolumn{1}{c}{} & $x_1$ & PC & $x_1$ & PC \\
\midrule
\multicolumn{1}{c}{Case: $d = 256$} & \multicolumn{4}{c}{} \\
\rowcolor{markovColor} Markovian & 1 & 1 & 1 & 1\\
\rowcolor{defaultNutsColor} Zigzag-\Nuts{} $\left( \baseIntegrationTimeMultiplier = 0.1 \right)$ 
& 4.5 & 4.6 & 41 & 40 \\
\rowcolor{altNutsColor} Zigzag-\Nuts{} $\left( \baseIntegrationTimeMultiplier = 1 \right)$
& 15 & 18 & 55 & 55 \\
\rowcolor{white} 
\multicolumn{1}{c}{Case: $d = 1{,}024$} & \multicolumn{4}{c}{} \rule{0pt}{12pt}\\
\rowcolor{defaultNutsColor} Zigzag-\Nuts{} $\left( \baseIntegrationTimeMultiplier = 0.1 \right)$ 
& 4.7 & 4.5 & 54 & 54 \\
\rowcolor{altNutsColor} Zigzag-\Nuts{} $\left( \baseIntegrationTimeMultiplier = 1 \right)$ 
& 27 & 33 & 78 & 78 \\
\bottomrule
\end{tabular}
\end{table}

\FloatBarrier
\section{Performance comparison of two zigzags in platform-independent metric}
\label{sec:ess_per_event}

As pointed out in Section~\ref{sec:simulation_setup_and_efficiency_metric}, it is a considerable challenge to account for computational costs of different \mcmc{} algorithms in a practically informative manner. 
As one alternative to \ess{} per time, here we consider \ess{} per velocity switch event as a platform-independent performance metric.
This metric is motivated by the similarity in required calculations for the two zigzags' inter-event simulations (Algorithm~\ref{alg:hamiltonian_zigzag_simulation} and \ref{alg:markovian_zigzag_simulation}). 
\Ess{} per event is also of theoretical interest since, theoretically speaking, the main difference between the two zigzags arise from the difference in their inter-event time distributions (Line~\ref{line:coord_wise_event_time} of Algorithm~\ref{alg:hamiltonian_zigzag_simulation} and \ref{alg:markovian_zigzag_simulation}).

In case of truncated Gaussians (Algorithm~\ref{alg:hamiltonian_zigzag_for_truncated_normal} and \ref{alg:markovian_zigzag_for_truncated_normal}), our code profiling reveals generating a large number of pseudo-random numbers (Line~\ref{line::uniform_rv_for_truncated_gaussian_markovian_zigzag}) for Markovian zigzag to be the most computationally expensive step.
With the trajectory simulation optimized for truncated Gaussian targets, the rest of the algorithms requires less computational efforts. 
This fact gives Hamiltonian zigzag an additional advantage in terms of computational speed and hence of \ess{} per time as presented in Section~\ref{sec:simulation}.
In contrast, our use of \ess{} per event here tilts the comparison in favor of Markovian zigzag.

The numerical results of Section~\ref{sec:simulation} are presented, now in terms of the alternative performance metric, in Table~\ref{tab:ess_per_event_compound_symmetric} and \ref{tab:ess_per_event_phylo_probit}. 
Overall, we observe the same pattern here as in Section~\ref{sec:simulation} ---
Hamiltonian zigzag outperforms Markovian one except in the i.i.d.\ Gaussian case, with greater efficiency gain at stronger correlations.
The main difference, for the reason we describe above, is that efficiency gains are smaller when measured in \ess{} per event rather than in \ess{} per time (cf.\ Table~\ref{tab:ess_per_time_compound_symmetric} and \ref{tab:ess_per_time_phylo_probit}).
We see that relative \ess{} per time is typically greater than that per event by a factor of $3$ to $5$.\footnote{%
	This factor varies more in the i.i.d.\ Gaussian case (under $\rho = 0$ in Table~\ref{tab:ess_per_time_compound_symmetric} and \ref{tab:ess_per_event_compound_symmetric}), ranging from $2.0$ --- for the case of Zigzag-\Nuts{} $\left( \baseIntegrationTimeMultiplier = 0.1 \right)$ and $\nParam = 1{,}024$ --- to $8.2$ --- for the case of Zigzag-\Hmc{} and $\nParam = 256$. 
	This greater variations is due to the exceptional simplicity of the i.i.d.\ Gaussian target;
	this particular target is so simple that aspects of the algorithms, which ordinarily contribute little to overall computational costs, becomes significant.
}

\begin{table}
	\caption{%
		\Ess{} per velocity switch event --- relative to that of Markovian zigzag --- under the compound symmetric posteriors.
		Except for the difference in $\ess{}$ standardizations, the table follows the same format as Table~\ref{tab:ess_per_time_compound_symmetric}.
		\label{tab:ess_per_event_compound_symmetric}
	}
	\centering
	\begin{tabular}[t]{lccccccc}
		\toprule
		& \multicolumn{5}{c}{Relative \Ess{} per event}\\
		\cmidrule(l{3pt}r{3pt}){2-6}
		\multicolumn{1}{c}{Compound symmetric} 
		& \multicolumn{1}{c}{$\rho = 0$} & \multicolumn{2}{c}{$\rho = 0.9$} & \multicolumn{2}{c}{$\rho = 0.99$} \\
		\cmidrule(l{3pt}r{3pt}){2-2} \cmidrule(l{3pt}r{3pt}){3-4} \cmidrule(l{3pt}r{3pt}){5-6}
		\multicolumn{1}{c}{} & $x_1$  & $x_1$ & PC & $x_1$ & PC \\
		\midrule
		\multicolumn{1}{c}{Case: $d = 256$} & \multicolumn{5}{c}{} \\
		\rowcolor{markovColor} Markovian & 1 & 1 & 1 & 1 & 1\\
		\rowcolor{defaultNutsColor} Zigzag-\Nuts{} $\left( \baseIntegrationTimeMultiplier = 0.1 \right)$ 
		& 0.27 & 1.2 & 1.3 & 8.0 & 8.0 \\ 
		\rowcolor{altNutsColor} Zigzag-\Nuts{} $\left( \baseIntegrationTimeMultiplier = 1 \right)$
		& 0.37 & 3.2 & 3.8 & 9.6 & 9.7 \\
		\rowcolor{hzzManualColor} Zigzag-\Hmc{} $\left( \integrationTime_\textrm{rel} = \sqrt{2} \right)$ 
		& 0.67 & 8.3 & 12 & 34 & 34 \\ 
		\rowcolor{white} 
		\multicolumn{1}{c}{Case: $d = 1{,}024$} & \multicolumn{5}{c}{} \rule{0pt}{12pt}\\
		\rowcolor{defaultNutsColor} Zigzag-\Nuts{} $\left( \baseIntegrationTimeMultiplier = 0.1 \right)$ 
		& 0.29 & 1.9 & 1.8 & 15 & 15 \\
		\rowcolor{altNutsColor} Zigzag-\Nuts{} $\left( \baseIntegrationTimeMultiplier = 1 \right)$ 
		& 0.33 & 7.5 & 9.1 & 19 & 19 \\ 
		\rowcolor{hzzManualColor} Zigzag-\Hmc{} $\left( \integrationTime_\textrm{rel}  = \sqrt{2} \right)$ 
		& 0.68 & 16 & 24 & 71 & 71 \\
		\bottomrule
	\end{tabular}
\end{table}
\begin{table}
	\caption{%
		Relative \ess{} per velocity switch event under the phylogenetic probit posterior.
		Except for the difference in $\ess{}$ standardizations, the table follows the same format as Table~\ref{tab:ess_per_time_phylo_probit}.
		\label{tab:ess_per_event_phylo_probit}
	}
	\centering
	\begin{tabular}[t]{lcccc}
		\toprule
		\multirow{2}{*}{\hspace{2em} Phylogenetic probit \rule{0pt}{12pt}} 
		& \multicolumn{2}{c}{Relative \Ess{} per event} \\
		\cmidrule(l{3pt}r{3pt}){2-3} 
		\cmidrule(l{3pt}r{3pt}){4-5}
		& \hspace{2ex} min \hspace{2ex} & PC\\
		\midrule
		\rowcolor{markovColor} Markovian & 1 & 1\\
		\rowcolor{defaultNutsColor} Zigzag-\Nuts{} $\left( \baseIntegrationTimeMultiplier = 0.1 \right)$ 
		& 2.0 & 5.9\\
		\bottomrule
	\end{tabular}
\end{table}

\section{Two zigzags' performance on rotated versions of compound symmetric targets}
\label{sec:rotated_compound_symmetric}

\myedit{ZigzagUnderRotatedCSSupplement}{%
	Here we answer the question posed in the last paragraph of Section~\ref{sec:synthetic_examples} by comparing the two zigzags' performances on rotated version of the compound symmetric targets. 
	We consider Gaussian targets with covariance matrices $\bPhi^{-1} = (1 - \rho) \Id + \rho \thinnerspace \principalComp \principalComp^\transpose$ whose principal components $\principalComp \in \mathbb{R}^\nParam$ are drawn uniformly from the $(\nParam - 1)$-dimensional sphere.
	We opt not to truncate the targets in this simulation since truncations can change targets correlation structures and get in the way of studying how the orientations of targets affect the zigzags' performances. 
	We compare the two zigzags under three different rotations with randomly generated $\principalComp$'s as well as under no rotation with $\principalComp = (1, \ldots, 1) / \sqrt{\nParam}$.
	
	The results are summarized in Table~\ref{tab:ess_per_time_rotated_compound_symmetric} and \ref{tab:ess_per_time_rotated_compound_symmetric_comparison_across_rotations}, focusing respectively on how the samplers' efficiencies vary across the different correlation strengths and across the different rotations of the compound symmetric Gaussian.
	Table~\ref{tab:ess_per_time_rotated_compound_symmetric} report \ess{}'s relative to Markovian zigzag to facilitate comparison among the algorithms.
	Table~\ref{tab:ess_per_time_rotated_compound_symmetric_comparison_across_rotations} on the other hand shows \ess{} in absolute terms to provide better sense of the computational challenge posed by the extreme correlation.
	We report \ess{}'s of zigzag \nuts{} under the same relative base integration time $\baseIntegrationTime_\textrm{rel} = 0.1$ as in Section~\ref{sec:synthetic_examples} as well as under $\baseIntegrationTime_\textrm{rel} = 1$, where the latter choice can yield significantly improved performance in the presence of such extreme correlation as in the compound symmetric case (Supplement~\ref{sec:base_integration_time_for_nuts}).
	For zigzag \hmc{}, we attempt $\integrationTime_\textrm{rel} = 2^{k / 2}$ for $k = -2, -1, 0, 1, 2$ as in Section~\ref{sec:synthetic_examples}.
	We report \ess{}'s under the same relative integration time $\integrationTime_\textrm{rel} = \sqrt{2}$ as in Section~\ref{sec:synthetic_examples} as well as under $\integrationTime_\textrm{rel} = 2$, where the latter choice delivers better performance here, likely due to the untruncated nature of the targets warranting slightly longer travel times for comprehensive exploration.
	
	Table~\ref{tab:ess_per_time_rotated_compound_symmetric} shows an essentially identical pattern as in Table~\ref{tab:ess_per_time_compound_symmetric} of Section~\ref{sec:synthetic_examples}, with Hamiltonian zigzag demonstrating increasingly superior performance over its Markovian counterpart as the correlation increases.
	Table~\ref{tab:ess_per_time_rotated_compound_symmetric_comparison_across_rotations} shows that each sampler's performance change little across different rotations of the target. 
	In summary, the results here indicate that the advantage of Hamiltonian zigzag does not depend on a specific orientation of a given target.%
}

\begin{table}
\centering
\caption{%
	\Ess{} per computing time, relative to that of Markovian zigzag, under rotated compound symmetric targets.
	The results are summarized in the same manner as in Table~\ref{tab:ess_per_time_compound_symmetric} of Section~\ref{sec:synthetic_examples}.
	Instead of the first coordinate \ess{}, however, we report the minimum \ess{} across the coordinates under the label ``min'' since the coordinates are no longer exchangeable after rotation.
	Under each target-sampler combination, the ``min'' and ``PC'' \ess{} are similar to each other, suggesting the principal component direction indeed constitutes the bottleneck in the chains' explorations.
	The results are obtained, for the same reasons as explained in Table~\ref{tab:comparison_against_other_samplers_ess_per_time}, using the algorithms' implementations provided by the \texttt{hdtg} package \protect\citepSupp{zhang2022hdtg} and using computing resources from Johns Hopkins' \jhpce{} cluster.
}
\label{tab:ess_per_time_rotated_compound_symmetric}
\begin{tabular}[t]{lcccc}
\toprule
\multirow{3}{*}{
	\hspace{1ex}
	\begin{tabular}[c]{@{}c@{}}
	\\[-20pt]
	Rotated \\[-8pt]
	compound symmetric
	\end{tabular}
}
& \multicolumn{4}{c}{Relative \Ess{} per time}\\
\cmidrule(l{3pt}r{3pt}){2-5}
& \multicolumn{2}{c}{$\rho = 0.9$} & \multicolumn{2}{c}{$\rho = 0.99$} \\
\cmidrule(l{3pt}r{3pt}){2-3} \cmidrule(l{3pt}r{3pt}){4-5}
\multicolumn{1}{c}{} & min & PC & min & PC \\
\midrule
\multicolumn{1}{c}{Case: $d = 256$} & \multicolumn{4}{c}{} \\
\rowcolor{markovColor} Markovian & 1 & 1 & 1 & 1 \\
\rowcolor{defaultNutsColor} Zigzag-\Nuts{} $\left( \baseIntegrationTime_\textrm{rel}  = 0.1 \right)$
& 3.7 & 3.7 & 63 & 63 \\ 
\rowcolor{defaultNutsColor} Zigzag-\Nuts{} $\left( \baseIntegrationTime_\textrm{rel}  = 1 \right)$
& 25 & 26 & 75 & 74 \\ 
\rowcolor{hzzManualColor} Zigzag-\Hmc{} $\left( \integrationTime_\textrm{rel} = \sqrt{2} \right)$
& 42 & 43 & 140 & 140 \\
\rowcolor{hzzManualColor} Zigzag-\Hmc{} $\left( \integrationTime_\textrm{rel} = 2 \right)$
& 75 & 77 & 250 & 250 \\
\rowcolor{white} 
\multicolumn{1}{c}{Case: $d = 1{,}024$} & \multicolumn{4}{c}{} \rule{0pt}{12pt}\\
\rowcolor{defaultNutsColor} Zigzag-\Nuts{} $\left( \baseIntegrationTime_\textrm{rel}  = 0.1 \right)$
& 6.5 & 6.7 & 130 & 140 \\ 
\rowcolor{defaultNutsColor} Zigzag-\Nuts{} $\left( \baseIntegrationTime_\textrm{rel}  = 1 \right)$
& 63 & 63 & 180 & 180 \\ 
\rowcolor{hzzManualColor} Zigzag-\Hmc{} $\left( \integrationTime_\textrm{rel} = \sqrt{2} \right)$
& 100 & 100 & 310 & 310 \\
\rowcolor{hzzManualColor} Zigzag-\Hmc{} $\left( \integrationTime_\textrm{rel} = 2 \right)$
& 180 & 180 & 560 & 570 \\
\bottomrule
\end{tabular}
\end{table}

\begin{table}
\centering
\caption{%
	\Ess{} per minute along principal component, compared across different rotations of the compound symmetric target with $\rho = 0.99$.
	The coordinate-wise minimum \ess{}'s are similar to those along PC and are hence omitted for space.
	The ``none'' label indicates the unrotated one with principal component $\principalComp = (1, \ldots, 1) / \sqrt{\nParam}$ and corresponding covariance structure as in Eq.~\eqref{eq:compound_symmetry_covariance}.
	The results are obtained, for the same reasons as explained in Table~\ref{tab:comparison_against_other_samplers_ess_per_time}, using the algorithms' implementations in the \texttt{hdtg} package \protect\citepSupp{zhang2022hdtg} and using computing resources from Johns Hopkins' \jhpce{} cluster.
}
\label{tab:ess_per_time_rotated_compound_symmetric_comparison_across_rotations}
\begin{tabular}[t]{lcccc}
\toprule
\multirow{3}{*}{
	\hspace{1ex}
	\begin{tabular}[c]{@{}c@{}}
	Rotated \\[-8pt]
	compound symmetric
	\end{tabular}
}
& \multicolumn{4}{c}{\Ess{} per minute along PC}\\
\cmidrule(l{3pt}r{3pt}){2-5}
& \multicolumn{4}{c}{
	Rotation type
} \\[-6pt]
& none & ver.~1 & ver.\ 2 & ver.\ 3 \\
\midrule
\multicolumn{1}{c}{Case: $d = 256$, $\rho = 0.99$} & \multicolumn{4}{c}{} \\
\rowcolor{markovColor} Markovian & 2.9 & 2.8 & 2.7 & 2.8 \\
\rowcolor{defaultNutsColor} Zigzag-\Nuts{} $\left( \baseIntegrationTime_\textrm{rel}  = 0.1 \right)$\hspace*{0ex}
& 180 & 180 & 180 & 190 \\ 
\rowcolor{defaultNutsColor} Zigzag-\Nuts{} $\left( \baseIntegrationTime_\textrm{rel}  = 1 \right)$\hspace*{0ex}
& 220 & 210 & 210 & 210 \\ 
\rowcolor{hzzManualColor} Zigzag-\Hmc{} $\left( \integrationTime_\textrm{rel} = \sqrt{2} \right)$\hspace*{-1.5ex}
& 410 & 390 & 410 & 400 \\
\rowcolor{hzzManualColor} Zigzag-\Hmc{} $\left( \integrationTime_\textrm{rel} = 2 \right)$\hspace*{-1.5ex}
& 710 & 700 & 690 & 690 \\
\rowcolor{white} 
\multicolumn{1}{c}{Case: $d = 1{,}024$, $\rho = 0.99$} & \multicolumn{4}{c}{} \rule{0pt}{12pt}\\
\rowcolor{markovColor} Markovian & 0.045 & 0.047 & 0.054 & 0.043 \\
\rowcolor{defaultNutsColor} Zigzag-\Nuts{} $\left( \baseIntegrationTime_\textrm{rel}  = 0.1 \right)$\hspace*{0ex}
& 6.2 & 6.4 & 6.2 & 5.9 \\ 
\rowcolor{defaultNutsColor} Zigzag-\Nuts{} $\left( \baseIntegrationTime_\textrm{rel}  = 1 \right)$\hspace*{0ex}
& 8.3 & 8.5 & 8.3 & 8.2 \\ 
\rowcolor{hzzManualColor} Zigzag-\Hmc{} $\left( \integrationTime_\textrm{rel} = \sqrt{2} \right)$\hspace*{-1.5ex}
& 15 & 15 & 15 & 15 \\
\rowcolor{hzzManualColor} Zigzag-\Hmc{} $\left( \integrationTime_\textrm{rel} = 2 \right)$\hspace*{-1.5ex}
& 26 & 27 & 26 & 27 \\
\bottomrule
\end{tabular}
\end{table}

\FloatBarrier

\section{Preconditioning target under Hamiltonian zigzag}
\label{sec:preconditioning}
\myedit{PreconditioningSimulation}{%
	Performances of \hmc{} algorithms can often be improved by \textit{preconditioning}; 
	i.e.\ targeting the density of a linearly transformed parameter $\tilde{\bposition} = \bm{A} \bposition$ for an invertible matrix $\bm{A}$ \protect\citepSupp{stan_manual}. 
	Incidentally, this can be shown as equivalent to specifying a momentum distribution with non-identity covariance instead of actually transforming the parameter of interest \protect\citepSupp{neal2010hmc, nishimura2020discontinuous_hmc}.
	Empirical and theoretical evidence support the use of transformation $\bm{A} = \bm{\Sigma}^{-1/2}$ or $\diag\!\left( \bm{\Sigma} \right)^{-1/2}$ for $\bm{\Sigma} = \textrm{Var}(\bposition)$ to improve efficiency \protect\citepSupp{girolami2011riemann_manifold_hmc, stan_manual}.
	
	In our phylogenetic application, we might take $\bPhi^{-1} = \bPhi(\traitCovariance, \phylogeny)^{-1}$ as an approximation of $\textrm{Var}(\bposition)$ by ignoring effects of the truncation and consider preconditioning via $\bm{A} = \bPhi^{1/2}$ or $\diag\!\left( \bPhi^{-1}  \right)^{-1/2}$.
	Neither choice is practical, however, as both require expensive matrix computations on $\bPhi$ to obtain necessary quantities. 
	Moreover, the diagonal preconditioning would have only a minor effect anyway because the phylogenetic model is parametrized in a way that the diagonals of $\bPhi^{-1}$ are on similar scales.
	In fact, by explicitly computing $\bPhi^{-1}$, we find that its diagonals range between $1.44$ and $2.00$ in its values.
	
	More generally for a (truncated) Gaussian target parameterized by a precision matrix $\bPhi$, $\diag(\bPhi)^{-1}$ could be a reasonable enough approximation of $\diag(\bPhi^{-1})$ to improve the efficiency of zigzag \hmc{} and \nuts{} through preconditioning.
	On the other hands, the diagonals of $\bPhi$ represent the conditional variances while the diagonals of $\bPhi^{-1}$ represent the marginal variances, so the two diagonals differ substantially in the presence of substantial correlation among parameters.
	For the purpose of preconditioning, therefore, $\diag(\bPhi^{-1})$ can be a poor enough approximation of $\diag(\bPhi)^{-1}$ to actually hurt the samplers' efficiency.
	We in fact observe the preconditioning via $\bm{A} = \diag\!\left( \bPhi  \right)^{1/2}$ to be more harmful rather than helpful in our phylogenetic application (Table~\ref{tab:ess_of_zigzag_with_preconditioning}).
	We thus caution against the default use of $\bm{A} = \diag\!\left( \bPhi  \right)^{1/2}$ for preconditioning;
	it is important that $\bm{A}$ actually approximates $\bm{\Sigma}^{-1/2}$ or $\diag\!\left( \bm{\Sigma} \right)^{-1/2}$.%
}

\FloatBarrier
\begin{table}[htb]
\centering
\caption{%
	Relative \ess{} per computing time under the phylogenetic probit posterior of Section~\ref{sec:phylogenetic_probit}. 
	``Zigzag-\Nuts{} (w/o preconditioning)'' is identical to the algorithm compared to Markovian zigzag in Table~\ref{tab:ess_per_time_phylo_probit}.
	The preconditioned version uses the linear transformation $\bm{A} = \diag\!\left( \bPhi  \right)^{1/2}$ based on the conditional variances, which actually hurts the sampler's efficiency. 
	The results are obtained, for the same reasons as explained in Table~\ref{tab:comparison_against_other_samplers_ess_per_time}, using the algorithms' implementations in the \texttt{hdtg} package \protect\citepSupp{zhang2022hdtg} and using computing resources from Johns Hopkins' \jhpce{} cluster. 
	The simulation settings and \ess{} calculations are otherwise identical to those in Section~\ref{sec:phylogenetic_probit}.
}
\label{tab:ess_of_zigzag_with_preconditioning}
\begin{tabular}[t]{lcccc}
\toprule
\multirow{2}{*}{Phylogenetic probit \hspace*{-1ex} \rule{0pt}{12pt}} 
& \multicolumn{2}{c}{Relative \Ess{} per time} \\
\cmidrule(l{3pt}r{3pt}){2-3} 
\cmidrule(l{3pt}r{3pt}){4-5}
& \hspace{2ex} min \hspace{2ex} & PC\\
\midrule
\rowcolor{defaultNutsColor} Zigzag-\Nuts{} (w/o preconditioning)
& 2.4 & 2.1\\
\rowcolor{altNutsColor} Zigzag-\Nuts{} (with preconditioning) & 1 & 1\\
\bottomrule\\
\end{tabular}
\end{table}

{\spacingset{1.4}
\bibliographystyleSupp{agsm}
\bibliographySupp{zigzag_hmc}
}

\end{document}